\documentclass[aps,twocolumn,superscriptaddress]{revtex4-1}%
\usepackage{amsfonts}
\usepackage{amsmath}
\usepackage{amssymb}
\usepackage{graphicx}
\usepackage{mathtools}
\usepackage[colorlinks=true,linkcolor=blue,citecolor=red,plainpages=false,pdfpagelabels]%
{hyperref}%
\usepackage{color,xcolor}
\providecommand{\U}[1]{\protect\rule{.1in}{.1in}}
\newtheorem{theorem}{Theorem}

\newtheorem{corollary}[theorem]{Corollary}

\newtheorem{definition}[theorem]{Definition}

\newtheorem{lemma}[theorem]{Lemma}

\newtheorem{proposition}[theorem]{Proposition}
\newtheorem{remark}[theorem]{Remark}

\newenvironment{proof}[1][Proof]{\noindent\textbf{#1.} }{\ \rule{0.5em}{0.5em}}
\newcommand{\bra}[1]{\langle #1|}
\newcommand{\ket}[1]{|#1\rangle}

\newcommand{\Tr}{\operatorname{Tr}}

\allowdisplaybreaks
\begin{document}
\preprint{ }
\title[ ]{Second-order coding rates for key distillation in quantum key distribution}

\author{Sumeet Khatri}
\affiliation{Hearne Institute for Theoretical Physics, Department of Physics and Astronomy, and Center for Computation and Technology, Louisiana State University, Baton Rouge, Louisiana 70803, USA}

\author{Eneet Kaur}
\affiliation{Hearne Institute for Theoretical Physics, Department of Physics and Astronomy, and Center for Computation and Technology, Louisiana State University, Baton Rouge, Louisiana 70803, USA}

\author{Saikat Guha}
\affiliation{College of Optical Sciences, University of Arizona, 1630 E. University Blvd.,
Tucson, Arizona 85719, USA}
\affiliation{Department of Electrical and Computer Engineering, University of Arizona, 1230
E Speedway Blvd., Tucson, Arizona 85719, USA}

\author{Mark M. Wilde}
\affiliation{Hearne Institute for Theoretical Physics, Department of Physics and Astronomy, and Center for Computation and Technology, Louisiana State University, Baton Rouge, Louisiana 70803, USA}


\begin{abstract}
The security of quantum key distribution has traditionally been analyzed in
either the asymptotic or non-asymptotic regimes. In this paper, we provide a
bridge between these two regimes, by determining second-order coding rates for
key distillation in quantum key distribution under collective attacks. Our
main result is a formula that characterizes the backoff from the known
asymptotic formula for key distillation---our formula incorporates the
reliability and security of the protocol, as well as the mutual information
variances to the legitimate receiver and the eavesdropper. In order to
determine secure key rates against collective attacks, one should perform a
joint optimization of the Holevo information and the Holevo information
variance to the eavesdropper. We show how to do so by analyzing several
examples, including the six-state, BB84, and continuous-variable quantum key
distribution protocols (the last involving Gaussian modulation of coherent
states along with heterodyne detection). The technical contributions of this
paper include one-shot and second-order analyses of private communication over
a compound quantum wiretap channel with fixed marginal and key distillation
over a compound quantum wiretap source with fixed marginal. We also establish
the second-order asymptotics of the smooth max-relative entropy of quantum
states acting on a separable Hilbert space, and we derive a formula for the
Holevo information variance of a Gaussian ensemble of Gaussian states.

\end{abstract}
\date{\today}
\startpage{1}
\endpage{10}
\maketitle
\tableofcontents

\section{Introduction}

One of the near-term applications of quantum information science is quantum
key distribution (QKD) \cite{bb84,E91}. By making use of an insecure quantum
channel and a public, authenticated classical channel, two parties can share
an information-theoretically secure secret key, which can be used later on for
one-time pad encryption of a private message. There has been significant
progress on this topic in the decades since it was conceived
\cite{SBCDLP09,Lut14,LCT14,DLQY16,XMZLP19}.

A critical challenge for this research area is to establish security proofs
for quantum key distribution. In particular, we are interested in determining
the maximum possible rates that are achievable in principle and guaranteed to
be secure against any possible eavesdropper allowed by quantum mechanics.
Security proofs have been developed for discrete-variable quantum key
distribution (DV-QKD) protocols, both in the asymptotic
\cite{LC99,SP00,M01,RGK05,Koashi_2006}\ and non-asymptotic regimes
\cite{Renner2005,TLGR12,Hayashi_2012,Tomamichel2017largelyself}. There are
also advanced numerical approaches for addressing the asymptotic security of
DV-QKD\ protocols \cite{CML16,Winick2018reliablenumerical}. Additionally,
security proofs have been developed for continuous-variable quantum key
distribution (CV-QKD) protocols \cite{GG02,Grangier2002}, in the asymptotic
\cite{NGA06,PC2006,RC09}\ and non-asymptotic regimes
\cite{PhysRevLett.109.100502,PhysRevA.90.042325,PhysRevLett.118.200501}.
Furthermore, security proofs have appeared for asymptotic security of
discrete-modulation protocols for CV-QKD
\cite{Zhao2009,Bradler2018,KGW19,GGDL19,LUL19}.

In this work, we address the security of quantum key distribution in a regime
that represents a bridge between the asymptotic and non-asymptotic regimes.
Namely, we employ the methods of second-order asymptotics \cite{li12,TH12}\ in
order to determine secure key rates for the key distillation step of a quantum
key distribution protocol, under the assumption that an adversary applies a
collective attack. Our development here covers both DV-QKD\ and
CV-QKD\ protocols, due to various technical advances that we make in this
paper. Second-order quantum information theory grew out of earlier
developments in classical information theory \cite{Hay09,polyanskiy10}, and it
has since been explored extensively for various quantum communication tasks
\cite{TT13,DL15,DTW14,WRG16,TBR15,WTB16,wilde2017position,OMW19}. The goal of
a second-order information-theoretic analysis is to determine the extra
backoff from or overhead on the rates that are achievable in the asymptotic
case, which are due to finite-size effects. In the context of quantum key
distribution, the approach has been used for analyzing information
reconciliation \cite{Tomamichel2017}, for proving upper bounds on secure key
rates \cite{WTB16,KW17a}, and to address the entropy accumulation problem in
device-independent QKD~\cite{DF19}. See also \cite{Hay06} for earlier work on this topic.

One contribution of our paper is that it is possible to evaluate distillable
key rates in a regime that goes beyond a first-order asymptotic analysis,
which is the typical case studied in several of the aforementioned works
\cite{LC99,SP00,M01,Koashi_2006,CML16,Winick2018reliablenumerical,NGA06,PC2006,RC09,Zhao2009,Bradler2018,KGW19,GGDL19,LUL19}%
. It has been found in several preceding information-theoretic contributions
that a second-order analysis gives excellent agreement with what is actually
achievable in the finite-size regime \cite{Hay09,polyanskiy10,TBR15,WTB16}.
Thus, it is expected that our second-order analysis should agree well with
secure distillable key rates that are achievable in principle in the
finite-size regime.

To summarize our main contribution to QKD security analysis, suppose that
$\mathcal{S}$ is an \textquotedblleft uncertainty set\textquotedblright%
\ indexing the states of the eavesdropper Eve that are consistent with the
observed measurement results of the sender Alice and the receiver Bob. Let
$p_{XY}(x,y)$ be the probability distribution estimated by Alice and Bob after
the parameter estimation step of a QKD protocol \footnote{In some QKD
protocols, the entire joint distribution of Alice's and Bob's classical data
$p_{XY}(x,y)$ need not be estimated in order to distill a secure key. A well
known example is CV-QKD, where only two scalar parameters need to be estimated
in order to derive a key-rate lower bound under a collective-attack
assumption, even in a finite key-length regime \cite{Lev15}. It is possible
that, for such protocols, the second-order correction to the key distillation
rate that we present in this paper could be extended to a formulation that
does not require the estimation of $p_{XY}(x,y)$.}. For fixed $s\in
\mathcal{S}$, let $\left\{  p_{XY}(x,y),\rho_{E}^{x,y,s}\right\}  _{x,y,s}$ be
an ensemble of states that is consistent with the measurement results of Alice
and Bob.\ Here we are assuming that Eve employs a collective attack, meaning
that she applies the same quantum channel to every transmission of Alice. Then
our contribution is that the following is the rate at which secret key bits
can be generated in the key distillation step of a direct reconciliation
protocol, for sufficiently large $n$, such that Bob's decoding error
probability is no larger than $\varepsilon_{\operatorname{I}}\in(0,1)$ and the
security parameter of the key is no larger than $\varepsilon
_{\operatorname{II}}\in(0,1)$:%
\begin{multline}
I(X;Y)+\sqrt{\frac{1}{n}V(X;Y)}\Phi^{-1}(\varepsilon_{\operatorname{I}%
})\label{eq:key-rate-intro}\\
-\sup_{s\in\mathcal{S}}\left[  I(X;E)_{s}-\sqrt{\frac{1}{n}V(X;E)_{s}}%
\Phi^{-1}({\varepsilon_{\operatorname{II}}^2})\right] \\
+O\!\left(  \frac{\log n}{n}\right)  ,
\end{multline}
where
\begin{align}
\Phi^{-1}(\varepsilon)  &  \coloneqq \sup\left\{  a\in\mathbb{R}\ |\ \Phi
(a)\leq\varepsilon\right\}  ,\label{eq:inverse-cdf}\\
\Phi(a)  &  \coloneqq \frac{1}{\sqrt{2\pi}}\int_{-\infty}^{a}dx\ \exp\!\left(
\frac{-x^{2}}{2}\right)  .
\end{align}
In the above, $I(X;Y)$ is the Alice--Bob mutual information, $V(X;Y)$ is the
Alice--Bob mutual information variance, $I(X;E)_{s}$ is the Alice--Eve Holevo
information for fixed $s$, and $V(X;E)_{s}$ is the Alice--Eve Holevo
information variance for fixed $s$. The first-order term is given by
$I(X;Y)-\sup_{s\in\mathcal{S}}I(X;E)_{s}$, which is the standard information
quantity considered in asymptotic security analyses (see, e.g., \cite{Lut14}).
In such a first-order asymptotic analysis, the quantity $I(X;Y)-\sup
_{s\in\mathcal{S}}I(X;E)_{s}$ is typically understood as the asymptotic rate
at which an imperfect ensemble can be converted to a perfectly secure key.
However, this conversion is only perfect in a precise sense in the asymptotic
limit. What the mutual information variance $V(X;Y)$ and the Holevo
information variance $V(X;E)_{s}$ characterize are the fluctuations about the
first-order terms, which are due to finite-size effects, much like the
variance of a random variable characterizes the speed of convergence toward
the mean in the central-limit theorem \cite{KS10,S11}. The second-order terms
$\sqrt{\frac{1}{n}V(X;Y)}\Phi^{-1}(\varepsilon_{\operatorname{I}})$ and
$\sqrt{\frac{1}{n}V(X;E)_{s}}\Phi^{-1}({\varepsilon_{\operatorname{II}}^2})$ are
negative for $\varepsilon_{\operatorname{I}},{\varepsilon_{\operatorname{II}%
}^2}<1/2$, and thus they characterize the back-off from the asymptotic limit for
sufficiently large, yet finite $n$. In Section~\ref{sec:key-dist-QKD}, we
define all of these information quantities formally and explain them in more detail.

For the key distillation step of a reverse reconciliation protocol, the key
rate is given by%
\begin{multline}
I(X;Y)+\sqrt{\frac{1}{n}V(X;Y)}\Phi^{-1}(\varepsilon_{\operatorname{I}%
})\label{eq:key-rate-intro-2}\\
-\sup_{s\in\mathcal{S}}\left[  I(Y;E)_{s}-\sqrt{\frac{1}{n}V(Y;E)_{s}}%
\Phi^{-1}({\varepsilon_{\operatorname{II}}^2})\right] \\
+O\!\left(  \frac{\log n}{n}\right)  .
\end{multline}
The information quantities appearing in the above formula and their
interpretations are essentially the same as those given in
\eqref{eq:key-rate-intro}, with the difference being the substitutions
$I(X;E)_{s}\rightarrow I(Y;E)_{s}$ and $V(X;E)_{s}\rightarrow V(Y;E)_{s}$ to
account for reverse reconciliation.

The formulas in \eqref{eq:key-rate-intro} and \eqref{eq:key-rate-intro-2}
apply regardless of whether the variables $X$ and $Y$ are continuous or
discrete, and whether the system $E$ is finite- or infinite-dimensional. It is
thus an advantage of our results that the expressions are given in terms of
mutual informations and their variances rather than conditional entropies and
their variances, as the former ones are well defined for infinite-dimensional
systems (see the discussion in \cite[Section~IV-B]{LGG10}). Furthermore, the
above formulas apply whenever $n$ is large enough so that $n\sim
\max\{\varepsilon_{\operatorname{I}}^{-2},\varepsilon_{\operatorname{II}}%
^{-2}\}$. However, as we stated above, in prior work, this kind of
second-order approximation has given excellent agreement with actual
finite-size achievability statements \cite{Hay09,polyanskiy10,TBR15,WTB16}.

When performing the optimization given by
\begin{equation}
\sup_{s\in\mathcal{S}}\left[ I(Y;E)_{s}-\sqrt{\frac{1}{n}V(Y;E)_{s}}\Phi
^{-1}({\varepsilon_{\operatorname{II}}^2})\right] ,
\end{equation}
it is possible for large enough $n$ to employ a perturbative approach, in
which the first-order term $I(Y;E)_{s}$ is optimized first, and then among all
of the ensembles achieving the optimum first-order term, one further optimizes
the second-order term $-\sqrt{\frac{1}{n}V(Y;E)_{s}}\Phi^{-1}({\varepsilon
_{\operatorname{II}}^2})$. See \cite[Lemmas~63 and 64]{polyanskiy10} for a
justification of this approach.

In Section~\ref{sec:app-QKD}, we evaluate achievable secure key rates for
three standard QKD\ protocols:\ the six-state DV-QKD\ protocol \cite{B98,BG99}%
, the BB84 DV-QKD protocol \cite{bb84}, and the CV-QKD protocol involving
Gaussian-modulation coherent-state encodings along with heterodyne detection
\cite{GG02,Grangier2002}. We find various analytical expressions for the
achievable second-order coding rates, and we plot their performance in order
to have a sense of what rates are achievable in principle.

We remark here that the contribution of our paper goes beyond the setting of
quantum key distribution and applies more broadly in the context of private
communication in quantum information theory. A starting point for our
development is \cite{wilde2017position}, but here our result applies more
broadly to the scenario in which the Alice--Eve correlations are not precisely
known but instead chosen from an ``uncertainty set.'' Furthermore, our results
apply to the scenario in which the underlying quantum states act on an
infinite-dimensional, separable Hilbert space (see \cite{HZ12}\ for a review
of this setting of quantum information theory). See
Sections~\ref{sec:priv-comm-comp-wiretap}
and~\ref{sec:key-dist-compound-source} for a one-shot analysis and
Section~\ref{sec:second-order-priv-comm-key-dist} for a second-order analysis.
The information-theoretic setting that we consider here (secure communication
against collective attacks) is strongly related to universal private
communication and secret key distillation \cite{DH10}, as well as the compound
wiretap channel and compound wiretap source settings \cite{CCD12,BCCD14,BJ16}.
We note here that the compound wiretap setting has been of considerable
interest in classical information theory recently \cite{SBP15}.

As a side note, it is unclear to us from reading the literature whether
researchers working on security analyses against collective attacks in
QKD\ and those working on universal and compound information-theoretic secrecy
questions are fully aware of each other, and so one byproduct of our work
could be to develop more interaction between these communities. Related to
this, the QKD\ security proof community has been consistently applying the
Devetak--Winter formula \cite{Devetak2005}\ to analyze secret key rates for
collective attacks in QKD, in spite of the fact that the Devetak--Winter
protocol from \cite{Devetak2005} does not generally apply to such a scenario
and instead only applies to a known, fixed attack. The results of our paper
also clarify and bridge this gap, and we discuss all of these points in more
detail in Section~\ref{sec:col-att-QKD-comp-wiretap}.

On a technical level, an important contribution of our paper is to determine
the second-order asymptotics of the smooth max-relative entropy (see
Appendix~\ref{app:2nd-order-smooth-dmax}), and this is the main reason why our
security analysis applies to the infinite-dimensional case. This result also
has implications for the distinguishability dilution task in the resource
theory of asymmetric distinguishability \cite{WW19states}. We also extend some
other known relations from the finite- to infinite-dimensional case (see
Appendices~\ref{app:relation-smooth-max-MIs},
\ref{sec:smooth-univ-convex-split},  \ref{app:duality-smooth-dmax}, and \ref{app:smooth-dmax-hypo-test}). Combined with the second-order asymptotics
of the hypothesis testing relative entropy for the infinite-dimensional case
\cite{DPR15,KW17a,OMW19}, we are then led to our claim concerning second-order
coding rates for private communication, secret key distillation, and key
distillation in quantum key distribution.

Another technical contribution of our paper is a formula for the Holevo
information variance of a Gaussian ensemble of quantum Gaussian states (see
Proposition~\ref{prop:gaussian-formulas-HI-Hi-var} in
Appendix~\ref{app:Holevo-info-and-var-Gaussian}). This formula is useful in a
second-order analysis of CV-QKD protocols, and we expect it to be useful in
other contexts besides those considered here. We also derive a novel
expression for the Holevo information of a Gaussian ensemble of quantum
Gaussian states (see Proposition~\ref{prop:gaussian-formulas-HI-Hi-var}).

The rest of the paper proceeds as follows. We first consider the second-order
analysis of the key distillation step of quantum key distribution in
Section~\ref{sec:app-QKD}, and therein we analyze the approach for the
important examples mentioned above (six state, BB84, CV-QKD). After that, we
then develop the information-theoretic compound wiretap settings and protocols
and the corresponding secure rates in detail. In
Section~\ref{sec:col-att-QKD-comp-wiretap}, we provide a historical discussion
of the compound wiretap setting and collective attacks in quantum key
distribution, with the stated goal of bridging the communities working on
these related topics. We finally conclude in Section~\ref{sec:conclusion} with
a summary and a list of open questions. The appendices provide details of
various technical contributions that are useful for this work and might be of
independent interest.

\section{Second-order coding rates for key distillation in quantum key
distribution}

\label{sec:app-QKD}We begin by presenting one of our main results, which is
the application of the second-order coding rates in \eqref{eq:key-rate-intro}
and \eqref{eq:key-rate-intro-2} to the key distillation step of a quantum key
distribution protocol. The claim here rests upon the one-shot
information-theoretic key distillation protocol from
Section~\ref{sec:key-dist-compound-source}\ and its second-order expansion in
Section~\ref{sec:second-order-priv-comm-key-dist}, the details of which are
presented in these later sections. Here we first review a generic
prepare-and-measure QKD protocol and then state how \eqref{eq:key-rate-intro}
and \eqref{eq:key-rate-intro-2} apply in this context. We finally analyze the
approach in the context of the six-state DV-QKD protocol \cite{B98,BG99}, the
BB84 DV-QKD\ protocol \cite{bb84}, and the reverse-reconciliation
CV-QKD\ protocol \cite{Grangier2002}. All source code files needed to generate the plots in this section are available as ancillary files with the arXiv posting of this paper.

We should clarify that our main emphasis, as indicated above, is on the key
distillation step of a quantum key distribution protocol and how to
incorporate a second-order coding rate analysis. As part of this, we are
assuming that the parameter estimation step of a QKD\ protocol yields reliable
estimates of the classical channel from Alice to Bob that is induced by the
quantum part of the QKD\ protocol. We expand more upon this point in what
follows, and we note up front here that it is an open question to incorporate
a full second-order analysis that includes the parameter estimation step in
addition to the key distillation step.

\subsection{Generic prepare-and-measure QKD\ protocol}

\label{sec:gen-pm-prot}Let us recall the structure of a generic
prepare-and-measure protocol for quantum key distribution, in which the
adversary applies a collective attack. The protocol consists of three steps: a
quantum part, parameter estimation, and key distillation (the last combines
information reconciliation and privacy amplification into a single step). We
consider both direct and reverse reconciliation settings for key distillation.

\subsubsection{Quantum part}

The quantum part of the protocol consists of the following steps:

\begin{enumerate}
\item It begins with the sender Alice picking a state $\rho_{A}^{x}$\ randomly
from an ensemble%
\begin{equation}
\mathcal{E}_{A}\coloneqq \left\{  p_{X}(x),\rho_{A}^{x}\right\}  _{x\in\mathcal{X}},
\end{equation}
where $p_{X}$ is a probability distribution and each state $\rho_{A}^{x}$ is
described by a density operator acting on an input Hilbert space
$\mathcal{H}_{A}$.

\item Alice transmits the system $A$ through an unknown quantum channel
$\mathcal{N}_{A\rightarrow B}$, with input system $A$ and output system $B$.
It is assumed that the channel itself is controlled by the eavesdropper Eve,
and the output system $B$ is given to the legitimate receiver Bob. It is a
standard feature of quantum information that every quantum channel has an
isometric extension \cite{S55} (see also, e.g., \cite{Wil17}), from which the
original channel can be realized by a partial trace over an environment
system. Thus, there exists an isometric channel $\mathcal{U}_{A\rightarrow
BE}^{\mathcal{N}}$ extending the channel $\mathcal{N}_{A\rightarrow B}$\ such
that%
\begin{equation}
\mathcal{N}_{A\rightarrow B}=\operatorname{Tr}_{E}\circ\mathcal{U}%
_{A\rightarrow BE}^{\mathcal{N}}.
\end{equation}
In the worst-case scenario, the eavesdropper has full access to the
environment system $E$, and so we assume that she does (as is standard in
QKD\ security proof analyses).

\item Upon receiving the system $B$, the receiver Bob performs a measurement
channel $\mathcal{M}_{B\rightarrow Y}$, which is uniquely specified by a
positive operator-valued measure (POVM) $\{\Lambda_{B}^{y}\}_{y\in\mathcal{Y}%
}$, such that%
\begin{equation}
\Lambda_{B}^{y}\geq0\quad\forall y\in\mathcal{Y},\qquad\sum_{y\in\mathcal{Y}%
}\Lambda_{B}^{y}=I_{B}.
\end{equation}
According to the Born rule, the measurement channel $\mathcal{M}_{B\rightarrow
Y}$ gives the outcome $y$ with probability $\operatorname{Tr}[\Lambda^{y}_{B}
\omega_{B}]$ if the input state is $\omega_{B}$.
\end{enumerate}

One round of this protocol leads to the following ensemble:%
\begin{align}
\mathcal{E}_{\text{QKD}}  &  \coloneqq \mathcal{E}_{\text{QKD}}(\mathcal{E}%
_{A},\mathcal{U}_{A\rightarrow BE}^{\mathcal{N}},\mathcal{M}_{B\rightarrow
Y})\\
&  \coloneqq \left\{  p_{XY}(x,y),\rho_{E}^{x,y}\right\}  _{x\in\mathcal{X}%
,y\in\mathcal{Y}},
\end{align}
where the joint distribution $p_{XY}(x,y)$ is given as%
\begin{align}
p_{XY}(x,y)  &  =p_{X}(x)p_{Y|X}(y|x),\\
p_{Y|X}(y|x)  &  \coloneqq \operatorname{Tr}[\Lambda_{B}^{y}\mathcal{N}_{A\rightarrow
B}(\rho_{A}^{x})], \label{eq:QKD-induced-classical-channel}%
\end{align}
and the eavesdropper states $\rho_{E}^{x,y}$ are as follows:%
\begin{equation}
\rho_{E}^{x,y}\coloneqq \frac{\operatorname{Tr}_{B}[\Lambda_{B}^{y}\mathcal{U}%
_{A\rightarrow BE}^{\mathcal{N}}(\rho_{A}^{x})]}{p_{Y|X}(y|x)}.
\end{equation}
Observe how the protocol induces a classical channel $p_{Y|X}(y|x)$ from Alice
to Bob via the Born rule in \eqref{eq:QKD-induced-classical-channel}.

The above steps are repeated $m=k+n$ times, where $k$ is the number of rounds
used for parameter estimation and $n$ is the number of rounds used for key
distillation. The ensemble shared between Alice, Bob, and Eve after these $m$
rounds is as follows:%
\begin{align}
\mathcal{E}_{\text{QKD}}^{\otimes m}  &  \coloneqq \mathcal{E}_{\text{QKD}}^{\otimes
m}(\mathcal{E}_{A}^{\otimes m},(\mathcal{U}_{A\rightarrow BE}^{\mathcal{N}%
})^{\otimes m},(\mathcal{M}_{B\rightarrow Y})^{\otimes m}%
)\label{eq:tensor-power-ensemble-1}\\
&  \coloneqq \left\{  p_{X^{m}Y^{m}}(x^{m},y^{m}),\rho_{E^{m}}^{x^{m},y^{m}}\right\}
_{x^{m}\in\mathcal{X}^{m},y^{m}\in\mathcal{Y}^{m}},
\end{align}
where%
\begin{align}
p_{X^{m}Y^{m}}(x^{m},y^{m})  &  \coloneqq \prod\limits_{i=1}^{m}p_{X_{i}Y_{i}}%
(x_{i},y_{i}),\\
\rho_{E^{m}}^{x^{m},y^{m}}  &  \coloneqq \rho_{E_{1}}^{x_{1},y_{1}}\otimes
\cdots\otimes\rho_{E_{m}}^{x_{m},y_{m}},\\
\rho_{E_{i}}^{x_{i},y_{i}}  &  \coloneqq \frac{\operatorname{Tr}_{B_{i}}%
[\Lambda_{B_{i}}^{y_{i}}\mathcal{U}_{A_{i}\rightarrow B_{i}E_{i}}%
^{\mathcal{N}}(\rho_{A_{i}}^{x_{i}})]}{p_{Y_{i}|X_{i}}(y_{i}|x_{i})},\\
p_{Y_{i}|X_{i}}(y_{i}|x_{i})  &  \coloneqq \operatorname{Tr}[\Lambda_{B_{i}}^{y_{i}%
}\mathcal{N}_{A_{i}\rightarrow B_{i}}(\rho_{A_{i}}^{x_{i}})].
\label{eq:tensor-power-ensemble-last}%
\end{align}
A critical assumption that we make in the above is that Eve employs a
\textit{collective attack}, meaning that the isometric channel she applies
over the $m$ rounds is the tensor-power channel $(\mathcal{U}_{A\rightarrow
BE}^{\mathcal{N}})^{\otimes m}$.

\paragraph{Channel twirling}

\label{sec:channel-twirl}In order to simplify or symmetrize the collective
attacks that need to be considered, Alice and Bob can employ an additional
symmetrization of the protocol in each round, called channel twirling,
introduced in a different context in \cite{BDSW96}. Let $\{U_{A}^{g}%
\}_{g\in\mathcal{G}}$ and $\{V_{B}^{g}\}_{g\in\mathcal{G}}$ be unitary
representations of a group $\mathcal{G}$, such that the unitaries act on the
input Hilbert space $\mathcal{H}_{A}$ and output Hilbert space $\mathcal{H}%
_{B}$, respectively. Then before sending out her state $\rho_{A}^{x}$, Alice
can select $g$ uniformly at random from the group~$\mathcal{G}$, apply the
unitary $U_{A}^{g}$ to her state, communicate the value of $g$ over a public
classical communication channel to Bob, who then performs $V_{B}^{g\dag}$ on
his system before acting with his measurement. If Alice and Bob then discard
the value of $g$, this twirling procedure transforms the original quantum
channel $\mathcal{N}_{A\rightarrow B}$ to the following symmetrized quantum
channel:%
\begin{equation}
\overline{\mathcal{N}}_{A\rightarrow B}(\omega_{A})\coloneqq \frac{1}{\left\vert
\mathcal{G}\right\vert }\sum_{g\in\mathcal{G}}(\mathcal{V}_{A}^{g\dag}%
\circ\mathcal{N}_{A\rightarrow B}\circ\mathcal{U}_{A}^{g})(\omega_{A}),
\label{eq:twirled-channel}%
\end{equation}
where%
\begin{equation}
\mathcal{U}_{A}^{g}(\cdot)\coloneqq U_{A}^{g}(\cdot)U_{A}^{g\dag},\qquad
\mathcal{V}_{B}^{g}(\cdot)\coloneqq V_{B}^{g}(\cdot)V_{B}^{g\dag}.
\end{equation}

Although channel twirling produces a worse (noisier) channel from the original
one, the main benefit is that the number of parameters that are needed to
specify $\overline{\mathcal{N}}_{A\rightarrow B}$ can be far fewer than the
number needed to specify $\mathcal{N}_{A\rightarrow B}$. For example, if the
original channel $\mathcal{N}_{A\rightarrow B}$ is a qubit channel and the
unitaries consist of the Pauli operators, then the resulting twirled channel
is a Pauli channel and thus specified by only three parameters. If the
original channel $\mathcal{N}_{A\rightarrow B}$\ is a single-mode bosonic
channel and the unitaries consist of the four equally spaced phase rotations
$\left\{  0,\pi/2,\pi,3\pi/2\right\}  $, then the resulting channel is a phase
covariant (phase insensitive)\ channel \cite{KGW19}\ and has fewer parameters
that characterize it.

\paragraph{Finite-dimensional assumptions}

\label{sec:fd-assumptions}In the trusted device scenario that we are dealing
with here, if the Hilbert space $\mathcal{H}_{A}$ is finite dimensional, then
we are making an implicit assumption that the union of the supports of the
states $\rho_{A}^{x}$\ are fully contained in the Hilbert space $\mathcal{H}%
_{A}$ and there is no leakage outside of it. In the case that system $B$ is
finite dimensional (i.e., the Hilbert space $\mathcal{H}_{B}$ is finite
dimensional), then we are making an implicit assumption that the measurement
operators $\{\Lambda_{B}^{y}\}_{y\in\mathcal{Y}}$ satisfy $\sum_{y\in
\mathcal{Y}}\Lambda_{B}^{y}=I_{B}$ and that $I_{B}$ is indeed the identity
operator for $\mathcal{H}_{B}$. Thus, if both the input and output Hilbert
spaces $\mathcal{H}_{A}$ and $\mathcal{H}_{B}$ are finite dimensional, then by
the Choi--Kraus theorem (see, e.g., \cite{Wil17}), it is not necessary for
Eve's system to be any larger than $\dim(\mathcal{H}_{A})\cdot\dim
(\mathcal{H}_{B})$. So it follows that the finite dimensional assumption leads
to strong constraints about Eve's attack, which may not necessarily hold in
practice and thus should be stated up front. This is especially the case when
dealing with photonic states and measurements acting on subspaces of an
infinite-dimensional bosonic Fock space.

\subsubsection{Sifting}

\label{sec:sifting}

Some QKD\ protocols, such as the six-state and BB84 protocols, incorporate a
sifting step, in which a fraction of the data generated by the protocol is
discarded. The main reason for incorporating sifting is that some input-output
pairs in $\mathcal{X}\times\mathcal{Y}$ do not give any useful information
about the channel (Eve's attack) and thus can be discarded. For example, if
the input state is $|0\rangle\!\langle0|$ or $|1\rangle\!\langle1|$ and Bob
measures in the basis $\{|+\rangle\!\langle+|,|-\rangle\!\langle-|\}$, then even
in the noiseless case of no eavesdropping, the outcome of the measurement is
independent of the input and thus yields no useful information.

This sifting step can be described mathematically as a filter onto a subset
$\mathcal{F}\subseteq\mathcal{X}\times\mathcal{Y}$, with the sifting
probability given by%
\begin{equation}
p_{\text{sift}}\coloneqq \sum_{\left(  x,y\right)  \in\mathcal{F}}p_{XY}(x,y),
\end{equation}
and the conditioned ensemble by%
\begin{equation}
\mathcal{E}_{\text{QKD}}^{\text{sift}}\coloneqq \left\{  p_{XY}^{\prime}(x,y),\rho
_{E}^{x,y}\right\}  _{\left(  x,y\right)  \in\mathcal{F}},
\end{equation}
where%
\begin{equation}
p_{XY}^{\prime}(x,y)\coloneqq p_{XY}(x,y)/p_{\text{sift}}.
\end{equation}
As a consequence of sifting, some number $m^{\prime}$ of the original $m$
rounds are selected, and the remaining systems of Alice, Bob, and Eve are
described by the tensor-power ensemble:%
\begin{multline}
(\mathcal{E}_{\text{QKD}}^{\text{sift}})^{\otimes m^{\prime}}\coloneqq \\
\left\{  p_{X^{m^{\prime}}Y^{m^{\prime}}}^{\prime}(x^{m^{\prime}}%
,y^{m^{\prime}}),\rho_{E^{m^{\prime}}}^{x^{m^{\prime}},y^{m^{\prime}}%
}\right\}  _{(x^{m^{\prime}},y^{m^{\prime}})\in\mathcal{F}^{m^{\prime}}},
\end{multline}
with the above quantities defined similarly as in~\eqref{eq:tensor-power-ensemble-1}--\eqref{eq:tensor-power-ensemble-last}.

In the discussion that follows in Section~\ref{sec:parameter-est-QKD}, we
simply relabel $m^{\prime}$ as $m$ and the distribution $p_{XY}^{\prime}(x,y)$
as $p_{XY}(x,y)$, with it understood that elements of $p_{XY}(x,y)$ with
$(x,y)$ outside of $\mathcal{F}$ are set to zero. The distribution $p_{X}$ is
the marginal of $p_{XY}$, and the conditional distribution $p_{Y|X}$ is a
classical channel computed from $p_{XY}$ as usual via $p_{XY}(x,y)/p_{X}(x)$.

\subsubsection{Parameter estimation}

\label{sec:parameter-est-QKD}After $m$ rounds of the protocol are complete,
$k$ of the $XY$ classical systems are randomly selected by Alice and Bob for
parameter estimation, in order to estimate the set $\mathcal{S}$ of possible
collective attacks of Eve. To be clear, the classical systems used for
parameter estimation are $X_{i_{1}}$, $Y_{i_{1}}$, \ldots, $X_{i_{k}}$,
$Y_{i_{k}}$ for some randomly selected (without replacement) $i_{1}$, \ldots,
$i_{k}\in\left\{  1,\ldots,n\right\}  $.

In what follows, we assume that $k$ is large enough such that Alice and Bob
get a very reliable (essentially exact) estimate of the classical channel
$p_{Y|X}(y|x)$. This assumption is strong, but as stated above, our main focus
in this paper is on analyzing second-order coding rates in the key
distillation step of the quantum key distribution protocol. In
Section~\ref{sec:future-routes-PE}, we discuss various routes for addressing
this problem.

The goal of the parameter estimation step is to estimate the classical channel
$p_{Y|X}(y|x)$ reliably in order to produce an estimate of the unknown quantum
channel $\mathcal{N}_{A\rightarrow B}$ connecting Alice and Bob. Doing so then
allows them to estimate the isometric channel $\mathcal{U}_{A\rightarrow
BE}^{\mathcal{N}}$, up to an information theoretically irrelevant isometry
acting on the system $E$. In more detail, note that any other isometric
extension of the original channel $\mathcal{N}_{A\rightarrow B}$\ is related
to $\mathcal{U}_{A\rightarrow BE}^{\mathcal{N}}$ by an isometric channel
acting on the system $E$. That is, suppose that $\mathcal{V}_{A\rightarrow
BE^{\prime}}^{\mathcal{N}}$ is another isometric channel satisfying
$\mathcal{N}_{A\rightarrow B}=\operatorname{Tr}_{E^{\prime}}\circ
\mathcal{V}_{A\rightarrow BE^{\prime}}^{\mathcal{N}}$. Then there exists an
isometric channel $\mathcal{W}_{E\rightarrow E^{\prime}}$ such that
$\mathcal{V}_{A\rightarrow BE^{\prime}}^{\mathcal{N}}=\mathcal{W}%
_{E\rightarrow E^{\prime}}\circ\mathcal{U}_{A\rightarrow BE}^{\mathcal{N}}$
\cite{Wil17}. However, Eve's information about the classical systems $X$ and
$Y$ is the same regardless of which particular isometric extension is
considered. Thus, as stated above, the goal of the parameter estimation step
is to estimate the channel $\mathcal{N}_{A\rightarrow B}$ by employing the
estimate of $p_{Y|X}(y|x)$.

As a result of the parameter estimation step, Alice and Bob determine an
uncertainty set $\mathcal{S}$, each element~$s$ of which indexes a quantum
channel $\mathcal{N}_{A\rightarrow B}^{s}$ that is consistent with the
classical channel $p_{Y|X}(y|x)$, in the sense that%
\begin{equation}
p_{Y|X}(y|x)=\operatorname{Tr}[\Lambda_{B}^{y}\mathcal{N}_{A\rightarrow B}%
^{s}(\rho_{A}^{x})],
\end{equation}
for all $s\in\mathcal{S}$, $x\in\mathcal{X}$, and $y\in\mathcal{Y}$. It is
generally not possible to determine the actual quantum channel $\mathcal{N}%
_{A\rightarrow B}$ exactly, so that $\left\vert \mathcal{S}\right\vert >1$.
However, if the input ensemble $\mathcal{E}_{A}=\left\{  p_{X}(x),\rho_{A}%
^{x}\right\}  _{x\in\mathcal{X}}$ and the POVM\ $\{\Lambda_{B}^{y}%
\}_{y\in\mathcal{Y}}$ form a tomographically complete set \cite{CN97,PCZ97},
then it is possible to determine an exact estimate of the actual, unknown
quantum channel $\mathcal{N}_{A\rightarrow B}$ from the classical channel
$p_{Y|X}(y|x)$. That is, in the tomographically complete case, there exists an
invertible linear map relating $p_{Y|X}(y|x)$ and $\mathcal{N}_{A\rightarrow
B}$. So in this special case, there is a unique quantum channel $\mathcal{N}%
_{A\rightarrow B}$ corresponding to the classical channel $p_{Y|X}(y|x)$.

In some parameter estimation protocols, Alice and Bob do not estimate the
entries of $p_{Y|X}(y|x)$ for all $x\in\mathcal{X}$ and $y\in\mathcal{Y}$, but
they instead do so for a subset of $\mathcal{X}\times\mathcal{Y}$ (for example
as a consequence of sifting). Furthermore, they could reduce the number of
parameters that need to be estimated by employing additional symmetrization of
$p_{Y|X}(y|x)$, in which some of the entries are averaged or simple functions
of them are computed. These latter approaches are commonly employed in the
parameter estimation phase of the BB84 and six-state protocols, in which
quantum bit error rates (QBERs)\ are estimated in lieu of all entries of
$p_{Y|X}(y|x)$. Another example for which the full $p_{Y|X}(y|x)$ is not
typically estimated is CV-QKD, where only two scalar parameters are estimated
in order to derive a key-rate lower bound under a collective-attack
assumption, even in a finite key-length regime~\cite{Lev15}.

\subsubsection{Key distillation}

\label{sec:key-dist-QKD}After the parameter estimation step, the ensemble
characterizing each of the $n$ remaining systems shared by Alice, Bob, and Eve
is as follows:%
\begin{equation}
\mathcal{E}_{\text{QKD}}^{s}\coloneqq \left\{  p_{XY}(x,y),\rho_{E}^{x,y,s}\right\}
_{x\in\mathcal{X},y\in\mathcal{Y},s\in\mathcal{S}},
\label{eq:ensemble-QKD-for-key-dist}%
\end{equation}
where%
\begin{align}
p_{XY}(x,y)  &  =p_{X}(x)p_{Y|X}(y|x),\\
p_{Y|X}(y|x)  &  \coloneqq \operatorname{Tr}[\Lambda_{B}^{y}\mathcal{N}_{A\rightarrow
B}^{s}(\rho_{A}^{x})],
\end{align}
and the eavesdropper states $\rho_{E}^{x,y,s}$ are as follows:%
\begin{equation}
\rho_{E}^{x,y,s}\coloneqq \frac{\operatorname{Tr}_{B}[\Lambda_{B}^{y}\mathcal{U}%
_{A\rightarrow BE}^{\mathcal{N}^{s}}(\rho_{A}^{x})]}{p_{Y|X}(y|x)}.
\end{equation}
The full ensemble for the $n$ remaining systems is an $n$-fold tensor power of
the above, similar to that given in
\eqref{eq:tensor-power-ensemble-1}--\eqref{eq:tensor-power-ensemble-last},
except with the substitutions $m\rightarrow n$ and $\mathcal{N}\rightarrow
\mathcal{N}^{s}$. Note that the classical channel $p_{Y|X}(y|x)$ is known and
independent of $s$, due to our assumption of a collective attack and that $k$
is large enough so that Alice and Bob can estimate $p_{Y|X}(y|x)$ reliably.
That is, there could be many quantum channels $\mathcal{N}_{A\rightarrow
B}^{s}$ that lead to the same classical channel $p_{Y|X}(y|x)$ if the input
preparation and the output measurements are not tomographically complete.
Furthermore, the distribution $p_{X}(x)$ is known because Alice controls the
random selection of the states~$\{\rho_{A}^{x}\}_{x\in\mathcal{X}}$.

The ensemble in \eqref{eq:ensemble-QKD-for-key-dist} is then a particular
instance of the information-theoretic model presented later on in
Section~\ref{sec:key-dist-compound-source}. Specifically, the ensemble in
\eqref{eq:ensemble-QKD-for-key-dist} is an instance of the compound wiretap
source with fixed marginal from \eqref{eq:cqq-state}, in which the system $B$
in \eqref{eq:cqq-state} is in correspondence with the classical system $Y$ in
\eqref{eq:ensemble-QKD-for-key-dist}. As a result, we can apply
\eqref{eq:second-order-key-rate} (itself a consequence of
Theorem~\ref{thm:one-shot-key-dist}) to conclude that the following rate is
achievable for key distillation for the remaining $n$ rounds, by using direct
reconciliation:%
\begin{multline}
I(X;Y)_{\mathcal{E}^{s}}+\sqrt{\frac{1}{n}V(X;Y)_{\mathcal{E}^{s}}}\Phi
^{-1}(\varepsilon_{\operatorname{I}}%
)\label{eq:direct-reconciliation-second-order}\\
-\sup_{s\in\mathcal{S}}\left[  I(X;E)_{\mathcal{E}^{s}}-\sqrt{\frac{1}%
{n}V(X;E)_{\mathcal{E}^{s}}}\Phi^{-1}({\varepsilon_{\operatorname{II}}^2})\right]
\\
+O\!\left(  \frac{\log n}{n}\right)  ,
\end{multline}
for sufficiently large $n$ and with decoding error probability not larger than
$\varepsilon_{\operatorname{I}}$ and security parameter not larger than
$\varepsilon_{\operatorname{II}}$ (these latter two quantities are defined in
Section~\ref{sec:key-dist-compound-source-one-shot}). In the above, the first
two terms depend on the probability distribution $p_{XY}$ in
\eqref{eq:ensemble-QKD-for-key-dist}, but they have no dependence on $s$. The
inverse cumulative Gaussian distribution function $\Phi^{-1}(\cdot)$ is
defined in \eqref{eq:inverse-cdf}. The classical mutual information is defined
as%
\begin{equation}
I(X;Y)_{\mathcal{E}^{s}}\coloneqq \sum_{x\in\mathcal{X},y\in\mathcal{Y}}%
p_{XY}(x,y)\log_{2}\!\left(  \frac{p_{XY}(x,y)}{p_{X}(x)p_{Y}(y)}\right)  ,
\end{equation}
the classical mutual information variance as \cite{S62,Hay09,polyanskiy10}%
\begin{multline}
V(X;Y)_{\mathcal{E}^{s}}\coloneqq \\
\sum_{x\in\mathcal{X},y\in\mathcal{Y}}p_{XY}(x,y)\left[  \log_{2}\!\left(
\frac{p_{XY}(x,y)}{p_{X}(x)p_{Y}(y)}\right)  -I(X;Y)_{\mathcal{E}^{s}}\right]
^{2},
\end{multline}
the Holevo information as \cite{Holevo73}%
\begin{equation}
I(X;E)_{\mathcal{E}^{s}}\coloneqq \sum_{x\in\mathcal{X}}p_{X}(x)D(\rho_{E}^{x,s}%
\Vert\rho_{E}^{s}),
\end{equation}
and the Holevo information variance as%
\begin{multline}
V(X;E)_{\mathcal{E}^{s}}\coloneqq \\
\sum_{x\in\mathcal{X}}p_{X}(x)\left[  V(\rho_{E}^{x,s}\Vert\rho_{E}%
^{s})+\left[  D(\rho_{E}^{x,s}\Vert\rho_{E}^{s})\right]  ^{2}\right]  -\left[
I(X;E)_{\mathcal{E}^{s}}\right]  ^{2},
\end{multline}
where%
\begin{align}
\rho_{E}^{x,s}  &  \coloneqq \sum_{y\in\mathcal{Y}}p_{Y|X}(y|x)\rho_{E}^{x,y,s},\\
\rho_{E}^{s}  &  \coloneqq \sum_{x\in\mathcal{X}}p_{X}(x)\rho_{E}^{x,s}.
\end{align}
The quantum relative entropy of states $\omega$ and $\tau$ is defined as
\cite{U62}%
\begin{equation}
D(\omega\Vert\tau)\coloneqq \operatorname{Tr}[\omega\left(  \log_{2}\omega-\log
_{2}\tau\right)  ], \label{eq:rel-ent-q}%
\end{equation}
and the quantum relative entropy variance as \cite{li12,TH12}%
\begin{equation}
V(\omega\Vert\tau)\coloneqq \operatorname{Tr}[\omega\left(  \log_{2}\omega-\log
_{2}\tau-D(\omega\Vert\tau)\right)  ^{2}]. \label{eq:rel-ent-q-var}%
\end{equation}
These quantities are defined more generally in
\eqref{eq:rel-ent-sep}--\eqref{eq:rel-ent-var-sep} for states acting on
separable Hilbert spaces. If either alphabet $\mathcal{X}$ or $\mathcal{Y}$ is
continuous, then the corresponding sum is replaced with an integral.

If Alice and Bob employ reverse reconciliation instead for key distillation,
then the following distillable key rate is achievable:%
\begin{multline}
I(X;Y)_{\mathcal{E}^{s}}+\sqrt{\frac{1}{n}V(X;Y)_{\mathcal{E}^{s}}}\Phi
^{-1}(\varepsilon_{\operatorname{I}}%
)\label{eq:reverse-reconciliation-second-order}\\
-\sup_{s\in\mathcal{S}}\left[  I(Y;E)_{\mathcal{E}^{s}}-\sqrt{\frac{1}%
{n}V(Y;E)_{\mathcal{E}^{s}}}\Phi^{-1}({\varepsilon_{\operatorname{II}}^2})\right]
\\
+O\!\left(  \frac{\log n}{n}\right)  ,
\end{multline}
with the Holevo information defined as%
\begin{equation}
I(Y;E)_{\mathcal{E}^{s}}\coloneqq \sum_{y\in\mathcal{Y}}p_{Y}(y)D(\rho_{E}^{y,s}%
\Vert\rho_{E}^{s}), \label{eq:Holevo-info-def-rel-ent}%
\end{equation}
and the Holevo information variance as%
\begin{multline}
V(Y;E)_{\mathcal{E}^{s}}\coloneqq \label{eq:holevo-info-var-reverse}\\
\sum_{y\in\mathcal{Y}}p_{Y}(y)\left[  V(\rho_{E}^{y,s}\Vert\rho_{E}%
^{s})+\left[  D(\rho_{E}^{y,s}\Vert\rho_{E}^{s})\right]  ^{2}\right]  -\left[
I(Y;E)_{\mathcal{E}^{s}}\right]  ^{2},
\end{multline}
where%
\begin{equation}
\rho_{E}^{y,s}\coloneqq \sum_{x\in\mathcal{X}}p_{X|Y}(x|y)\rho_{E}^{x,y,s}.
\end{equation}

\begin{remark}
	{We emphasize that the formulas in \eqref{eq:direct-reconciliation-second-order} and \eqref{eq:reverse-reconciliation-second-order} represent the number of secret key bits per sifted bit. If we include all of the rounds of the protocol as part of the key rate, then the formulas in \eqref{eq:direct-reconciliation-second-order} and \eqref{eq:reverse-reconciliation-second-order} must be multiplied by the fraction of bits used for sifting; see, e.g., \cite[Eq.~(121)]{Tomamichel2017largelyself}.}
\end{remark}

\paragraph{On reconciliation efficiency and privacy amplification overhead}

It is common in the key distillation step of a first-order asymptotic analysis
to incorporate a reconciliation efficiency parameter $\beta\in\lbrack0,1]$,
which recognizes the fact that it is never possible in any realistic scheme to
achieve the Shannon limit $I(X;Y)_{\mathcal{E}^{s}}$. That is, the information
reconciliation rate is written as $\beta I(X;Y)_{\mathcal{E}^{s}}$. Typically,
the reconciliation efficiency $\beta$ is chosen as a constant independent of
the channel, the blocklength $n$, and the decoding error probability
$\varepsilon_{\operatorname{I}}$. As argued in \cite{Tomamichel2017}, this
approximation is rather rough, and one can instead employ a second-order
analysis to get a much better approximation of the reconciliation efficiency.
What we find from \eqref{eq:direct-reconciliation-second-order} and
\eqref{eq:reverse-reconciliation-second-order}\ is that the ideal
reconciliation efficiency is characterized in terms of $p_{XY}$, $n$, and
$\varepsilon_{\operatorname{I}}$ as follows:%
\begin{equation}
\beta(p_{XY},n,\varepsilon_{\operatorname{I}})\coloneqq 1+\frac{\sqrt{\frac{1}%
{n}V(X;Y)_{\mathcal{E}^{s}}}\Phi^{-1}(\varepsilon_{\operatorname{I}}%
)}{I(X;Y)_{\mathcal{E}^{s}}},
\end{equation}
so that%
\begin{multline}
\beta(p_{XY},n,\varepsilon_{\operatorname{I}})I(X;Y)_{\mathcal{E}^{s}}\\
=I(X;Y)_{\mathcal{E}^{s}}+\sqrt{\frac{1}{n}V(X;Y)_{\mathcal{E}^{s}}}\Phi
^{-1}(\varepsilon_{\operatorname{I}}).
\end{multline}
(Keep in mind that $\Phi^{-1}(\varepsilon_{\operatorname{I}})<0$ for
$\varepsilon_{\operatorname{I}}<1/2$.) One can also find very good fits of the
information reconciliation performance for actual codes by allowing for
constants $\beta_{1}$ and $\beta_{2}$, each not necessarily equal to one, to
characterize the reconciliation efficiency empirically as
follows~\cite{Tomamichel2017}:
\begin{equation}
\beta(p_{XY},n,\varepsilon_{\operatorname{I}},\beta_{1},\beta_{2})\coloneqq \beta
_{1}+\beta_{2}\frac{\sqrt{\frac{1}{n}V(X;Y)_{\mathcal{E}^{s}}}\Phi
^{-1}(\varepsilon_{\operatorname{I}})}{I(X;Y)_{\mathcal{E}^{s}}}.
\end{equation}
Thus, the formula above is a more useful guideline for reconciliation efficiency.

What we also notice in \eqref{eq:direct-reconciliation-second-order} and
\eqref{eq:reverse-reconciliation-second-order} is the presence of terms
related to privacy amplification overhead. Privacy amplification overhead
beyond the term $\sup_{s\in\mathcal{S}}I(Y;E)_{\mathcal{E}^{s}}$ is not
typically taken into account in first-order asymptotic security analyses, in
spite of the fact that other factors are inevitably necessary. To be clear,
the privacy amplification overhead is a factor $\gamma>1$ that multiplies the
asymptotic, first-order term $\sup_{s\in\mathcal{S}}I(Y;E)_{\mathcal{E}^{s}}$.
By employing \cite[Lemma~63]{polyanskiy10}, the following expansion in $n$
holds for direct-reconciliation privacy amplification for sufficiently
large$~n$:%
\begin{equation}
\sup_{s\in\mathcal{S}}\left[  I(X;E)_{\mathcal{E}^{s}}\right]  -\sqrt{\frac
{1}{n}V(X,\mathcal{S}^{\ast})}\Phi^{-1}({\varepsilon_{\operatorname{II}%
}^2})+o(1/\sqrt{n}),
\end{equation}
where%
\begin{align}
V(X,\mathcal{S}^{\ast})  &  \coloneqq \left\{
\begin{array}
[c]{ccc}%
\sup_{s\in\mathcal{S}^{\ast}}V(X;E)_{\mathcal{E}^{s}} & \text{if} &
{\varepsilon_{\operatorname{II}}^2}\leq\frac{1}{2}\\
\inf_{s\in\mathcal{S}^{\ast}}V(X;E)_{\mathcal{E}^{s}} & \text{else} &
\end{array}
\right.  ,\\
\mathcal{S}^{\ast}  &  \coloneqq \arg\max_{s\in\mathcal{S}}I(X;E)_{\mathcal{E}^{s}}.
\end{align}
In the above, we are applying the perturbative approach of \cite[Lemma~63]%
{polyanskiy10}, in which we optimize the first-order term $I(X;E)_{\mathcal{E}%
^{s}}$, and then among all of the optimizers of this first-order term, we are
optimizing the second-order term $V(X;E)_{\mathcal{E}^{s}}$. Thus, the ideal
privacy amplification overhead is given by%
\begin{equation}
\gamma(\mathcal{S},n,\varepsilon_{\operatorname{II}})\coloneqq 1-\frac{\sqrt{\frac
{1}{n}V(X,\mathcal{S}^{\ast})}\Phi^{-1}({\varepsilon_{\operatorname{II}}^2})}%
{\sup_{s\in\mathcal{S}}\left[  I(X;E)_{\mathcal{E}^{s}}\right]  },
\end{equation}
so that%
\begin{multline}
\gamma(\mathcal{S},n,\varepsilon_{\operatorname{II}})\sup_{s\in\mathcal{S}%
}\left[  I(X;E)_{\mathcal{E}^{s}}\right]  =\\
\sup_{s\in\mathcal{S}}\left[  I(X;E)_{\mathcal{E}^{s}}\right]  -\sqrt{\frac
{1}{n}V(X,\mathcal{S}^{\ast})}\Phi^{-1}({\varepsilon_{\operatorname{II}}^2}).
\end{multline}
As above, we could also allow for a more refined expression as follows, in
terms of constants $\gamma_{1}$ and $\gamma_{2}$, in order to fit the
performance of realistic privacy amplification protocols:%
\begin{equation}
\gamma(\mathcal{S},n,\varepsilon_{\operatorname{II}},\gamma_{1},\gamma
_{2})\coloneqq \gamma_{1}-\gamma_{2}\frac{\sqrt{\frac{1}{n}V(X,\mathcal{S}^{\ast}%
)}\Phi^{-1}({\varepsilon_{\operatorname{II}}^2})}{\sup_{s\in\mathcal{S}}\left[
I(X;E)_{\mathcal{E}^{s}}\right]  }.
\end{equation}

\subsection{Examples}

\subsubsection{General setup for the six-state and BB84 DV-QKD\ protocols}

We begin this example section by providing the general setup for both the
ideal, trusted six-state \cite{B98,BG99} and BB84 \cite{bb84} DV-QKD
protocols, which are particular instances of the generic prepare-and-measure
protocol presented in Section~\ref{sec:gen-pm-prot}.

In both of these protocols, the random variable $X$ for Alice's encoding is a
joint random variable consisting of random variables $X_{1}$ and $X_{2}$,
where $X_{1}$ represents Alice's basis choice, and $X_{2}$ is the binary
random variable corresponding to the state taken from the chosen basis. The
random variables $X_{1}$ and $X_{2}$ are independent. We similarly have that
the output random variable $Y$ for Bob is a joint random variable consisting
of random variables $Y_{1}$ and $Y_{2}$, where $Y_{1}$ represents the choice
of measurement basis, and $Y_{2}$ represents the outcome of the measurement.

Let the alphabet $\mathcal{B}$ contain the possible basis choices. The random
variables $X_{1}$ and $Y_{1}$ take values in $\mathcal{B}$. For the six-state
protocol, $\mathcal{B}_{\text{six-state}}=\{0,1,2\}$, with ``0'' denoting the
$X$-basis, ``1'' the $Z$-basis, and ``2'' the $Y$-basis. For the BB84
protocol, $\mathcal{B}_{\text{BB84}}=\{0,1\}$. Then, let $q_{b}^{A}$ and
$q_{b}^{B}$ be the probabilities that Alice and Bob, respectively, choose the
basis $b\in\mathcal{B}$. In other words,
\begin{equation}
q_{b}^{A}\coloneqq\Pr[X_{1}=b],\quad q_{b}^{B}\coloneqq\Pr[Y_{1}=b].
\end{equation}

Let us define the following:
\begin{align}
\Pi_{0}^{0}  &  \coloneqq|+\rangle\!\langle+|\equiv\rho_{A}^{0,0},\\
\Pi_{1}^{0}  &  \coloneqq|-\rangle\!\langle-|\equiv\rho_{A}^{0,1},\\
\Pi_{0}^{1}  &  \coloneqq|0\rangle\!\langle0|\equiv\rho_{A}^{1,0},\\
\Pi_{1}^{1}  &  \coloneqq|1\rangle\!\langle1|\equiv\rho_{A}^{1,1},\\
\Pi_{0}^{2}  &  \coloneqq\left\vert +i\right\rangle \!\langle+i|\equiv\rho
_{A}^{2,0},\\
\Pi_{1}^{2}  &  \coloneqq\left\vert -i\right\rangle \!\langle-i|\equiv\rho
_{A}^{2,1},
\end{align}
where%
\begin{align}
|+\rangle &  \coloneqq (|0\rangle+|1\rangle)/\sqrt{2},\\
|-\rangle &  \coloneqq (|0\rangle-|1\rangle)/\sqrt{2},\\
\left\vert +i\right\rangle  &  \coloneqq (|0\rangle+i|1\rangle)/\sqrt{2},\\
\left\vert -i\right\rangle  &  \coloneqq (|0\rangle-i|1\rangle)/\sqrt{2},
\end{align}

Now, Alice chooses the basis $b_{A}\in\mathcal{B}$ with probability $q_{b_{A}%
}^{A}$, and with probability $\frac{1}{2}$ chooses one of the two states
$\{\rho_{A}^{b_{A},0},\rho_{A}^{b_{A},1}\}$ in the basis to send to Bob. These
choices are independent, and so we have that
\begin{equation}
p_{X_{1}X_{2}}(b_{A},x)= q_{b_{A}}^{A}\cdot\frac{1}{2}.
\end{equation}
The encoding ensemble $\mathcal{E}_{A}$ is thus
\begin{equation}
\mathcal{E}_{A} \coloneqq  \{p_{X_{1}X_{2}}(b_{A},x), \rho^{b_{A},x}_{A} \}_{b_{A}
\in\mathcal{B}, x \in\{0,1\}}. \label{eq:six-state-ensemble}%
\end{equation}

The decoding POVM\ for Bob is%
\begin{equation}
\left\{  q^{B}_{b_{B}} \Pi^{b_{B}}_{y}\right\}  _{b_{B} \in\mathcal{B}, y
\in\{0,1\}} , \label{eq:six-state-meas}%
\end{equation}
which is equivalent to Bob picking the basis $b_{B}$ at random according to
$q^{B}_{b_{B}}$ and then performing the measurement~$\{\Pi^{b_{B}}_{0},
\Pi^{b_{B}}_{1}\}$.

The protocol is finite dimensional, meaning that the assumptions stated in
Section~\ref{sec:fd-assumptions}, come into play. Thus, the attack
$\mathcal{N}_{A\rightarrow B}$ that Eve applies is a qubit channel.

The channel twirling that Alice and Bob perform in this case is a Pauli
channel twirl \cite{BDSW96}. That is, before sending out her state, Alice
applies, uniformly at random, one of the Pauli operators $I$, $X$, $Y$, or $Z$
and Bob applies the corresponding Pauli after receiving the state from the
channel. Thus, both $\{U_{A}^{g}\}_{g\in\mathcal{G}}$ and $\{V_{B}^{g}%
\}_{g\in\mathcal{G}}$, as discussed in Section~\ref{sec:channel-twirl}, are
the Pauli group $\{I,X,Y,Z\}$. Then the resulting twirled channel is as given
in \eqref{eq:twirled-channel}, and it is well known that a Pauli twirl of a
qubit channel leads to a Pauli channel \cite{DHCB05}:%
\begin{equation}
\overline{\mathcal{N}}_{A\rightarrow B}(\omega_{A})=p_{I}\omega_{A}%
+p_{X}X\omega_{A}X+p_{Y}Y\omega_{A}Y+p_{Z}Z\omega_{A}Z,
\label{eq:pauli-channel}%
\end{equation}
where $p_{I},p_{X},p_{Y},p_{Z}\geq0$ and $p_{I}+p_{X}+p_{Y}+p_{Z}=1$.

Note that it is not necessary for Alice and Bob to apply the Pauli channel
twirl in an active way in the quantum domain. Since the encoding ensemble in
\eqref{eq:six-state-ensemble} is invariant with respect to the Pauli group
(for both the six-state and BB84 protocols), Alice can keep track of the twirl
in classical processing. Similarly, the decoding POVM\ of Bob in
\eqref{eq:six-state-meas} is invariant with respect to the Pauli group, so
that Bob can keep track of the twirl in classical processing.

At the end of the quantum part of the protocol, the induced classical channel
is as follows
\begin{equation}
p_{Y_{1}Y_{2}|X_{1}X_{2}}(b_{B},y|b_{A},x)=q_{b_{B}}^{B}\operatorname{Tr}%
[\Pi_{y}^{b_{B}}\overline{\mathcal{N}}_{A\rightarrow B}(\rho_{A}^{b_{A},x})].
\end{equation}
The joint probability distribution shared by Alice and Bob is then as
follows:
\begin{align}
&  p_{X_{1}X_{2}Y_{1}Y_{2}}(b_{A},x,b_{B},y)\nonumber\\
&  =p_{Y_{1}Y_{2}|X_{1}X_{2}}(b_{B},y|b_{A},x)p_{X_{1}X_{2}}(b_{A},x)\\
&  =\frac{1}{2}q_{b_{A}}^{A}q_{b_{B}}^{B}\operatorname{Tr}[\Pi_{y}^{b_{B}%
}\overline{\mathcal{N}}_{A\to B}(\rho_{A}^{b_{A},x})].
\label{eq-PM_QKD_full_dist}%
\end{align}

Both the six-state and BB84 protocols involve a sifting step, as discussed in
Section~\ref{sec:sifting}. Only the data are kept for which the sender and
receiver's basis bits agree. The probability that Alice and Bob choose the
same basis is given by
\begin{equation}
p_{\text{sift}}\coloneqq \sum_{b\in\mathcal{B}}q_{b}^{A}q_{b}^{B}%
\end{equation}
The resulting probability distribution shared by Alice and Bob is
\begin{equation}
p_{X_{1}X_{2}Y_{1}Y_{2}}^{\text{sift}}(b,x,b,y)\coloneqq \frac{q_{b}^{A}%
q_{b}^{B}}{2p_{\text{sift}}}\operatorname{Tr}[\Pi_{y}^{b}\overline
{\mathcal{N}}_{A\to B}(\rho_{A}^{b,x})],
\end{equation}
and it is for this (conditional) probability distribution that parameter
estimation occurs and using which key distillation occurs in both the
six-state and BB84 protocols.

The full classical-classical-quantum state of Alice, Bob, and the
eavesdropper, can be written via an isometric extension $\mathcal{U}_{A\to
BE}^{\overline{\mathcal{N}}}$ of the channel $\overline{\mathcal{N}}_{A\to B}%
$. Specifically,
\begin{multline}
\rho_{X_{1}X_{2}Y_{1}Y_{2}E}^{\text{sift}}=\\
\sum_{b\in\mathcal{B}}\sum_{x,y=0}^{1} \frac{q_{b}^{A}q_{b}^{B}}%
{p_{\text{sift}}} p_{X_{2}Y_{2}|X_{1}Y_{1}}(x,y|b,b)\ |b,b\rangle\!\langle
b,b|_{X_{1}Y_{1}}\\
\otimes|x,y\rangle\!\langle x,y|_{X_{2}Y_{2}}\otimes\rho_{E}^{b,x,y},
\end{multline}
where
\begin{align}
p_{X_{2}Y_{2}|X_{1}Y_{1}}(x,y|b,b)  &  =\frac{1}{2}\operatorname{Tr}[\Pi
_{y}^{b}\overline{\mathcal{N}}_{A\to B}(\rho_{A}^{b,x})],\\
\rho_{E}^{b,x,y}  &  =\frac{\operatorname{Tr}_{B}[\Pi_{y}^{b}\mathcal{U}_{A\to
BE}^{\overline{\mathcal{N}}}(\rho_{A}^{b,x})]}{p_{X_{2}Y_{2}|X_{1}Y_{1}%
}(x,y|b,b)}.
\end{align}

The channel parameters $p_{X}$, $p_{Y}$, and $p_{Z}$ in
\eqref{eq:pauli-channel} can be rewritten in terms of three quantum bit error
rates (QBERs) $Q_{x}$, $Q_{y}$, and $Q_{z}$, which in each case corresponds to
the expected probability that Bob measures a different state from what Alice
sent with respect to a given basis:
\begin{align}
Q_{x}  &  \coloneqq \frac{1}{2}\left(  \langle-|\overline{\mathcal{N}}(|+\rangle
\langle+|)|-\rangle+\langle+|\overline{\mathcal{N}}(|-\rangle\!\langle
-|)|+\rangle\right)  ,\\
Q_{y}  &  \coloneqq \frac{1}{2}\left\langle -i\right\vert \overline{\mathcal{N}%
}(\left\vert +i\right\rangle \!\left\langle +i\right\vert )\left\vert
-i\right\rangle \nonumber\\
&  \qquad\qquad+\frac{1}{2}\left\langle +i\right\vert \overline{\mathcal{N}%
}(\left\vert -i\right\rangle \! \left\langle -i\right\vert )\left\vert
+i\right\rangle ,\\
Q_{z}  &  \coloneqq \frac{1}{2}\left(  \langle1|\overline{\mathcal{N}}(|0\rangle
\langle0|)|1\rangle+\langle0|\overline{\mathcal{N}}(|1\rangle\!\langle
1|)|0\rangle\right)  .
\end{align}
The probabilities $p_{X}$, $p_{Y}$, and $p_{Z}$ are then related to $Q_{x}$,
$Q_{y}$, and $Q_{z}$ as follows:%
\begin{align}
p_{X}  &  =\frac{1}{2}\left(  Q_{z}-Q_{x}+Q_{y}\right)
,\label{eq:QBER-to-pauli-1}\\
p_{Y}  &  =\frac{1}{2}\left(  Q_{x}-Q_{y}+Q_{z}\right)  ,\\
p_{Z}  &  =\frac{1}{2}\left(  Q_{y}-Q_{z}+Q_{x}\right)  .
\label{eq:QBER-to-pauli-3}%
\end{align}
See \cite[Chapter~2]{Kh16}\ for a derivation of these equations. In the
six-state protocol, it is possible to estimate all of the QBERs reliably,
while in BB84, it is only possible to estimate $Q_{x}$ and $Q_{z}$ reliably.

Another further symmetrization of the protocol, in addition to channel
twirling, is possible if Alice and Bob discard the basis information in
$X_{1}$ and $Y_{1}$. This is commonly employed in both the six-state and BB84
protocols in order to simplify their analysis. Discarding the basis
information corresponds to tracing out the registers $X_{1}$ and $Y_{1}$
containing the basis information for Alice and Bob, respectively:%
\begin{multline}
\rho_{X_{2}Y_{2}E}^{\text{sift}}\coloneqq\operatorname{Tr}_{X_{1}Y_{1}}%
[\rho_{X_{1}X_{2}Y_{1}Y_{2}E}^{\text{sift}}]=\\
\sum_{x,y=0}^{1}\sum_{b\in\mathcal{B}}\frac{q_{b}^{A}q_{b}^{B}}{p_{\text{sift}%
}}p_{X_{2}Y_{2}|X_{1}Y_{1}}(x,y|b,b)|x,y\rangle\!\langle x,y|_{X_{2}Y_{2}}\\
\otimes\rho_{E}^{b,x,y},
\end{multline}

This discarding of basis information in either the six-state or BB84 protocols
is equivalent to a further channel twirl. Let us consider first the six-state
protocol, and let $T$ denote the following unitary%
\begin{equation}
T\coloneqq |+\rangle\!\langle0|-i|-\rangle\!\langle1|=\frac{1}{\sqrt{2}}%
\begin{bmatrix}
1 & -i\\
1 & i
\end{bmatrix}
,
\end{equation}
so that $T$ is a unitary changing the Pauli basis as
\begin{align}
TXT^{\dag}  &  =Y,\\
TYT^{\dag}  &  =Z,\\
TZT^{\dag}  &  =X.
\end{align}
Due to the fact that this gate swaps information encoded into the Pauli
eigenstates around, discarding of basis information in the six-state protocol
is equivalent to a further channel twirl as \cite[Section~2.2.7]{Myhr10}%
\begin{align}
&  \overline{\mathcal{N}}_{A\rightarrow B}^{Q}(\omega_{A})\nonumber\\
&  \coloneqq \frac{1}{3}\sum_{j\in\left\{  0,1,2\right\}  }T^{j\dag}\overline
{\mathcal{N}}_{A\rightarrow B}(T^{j}\omega_{A}T^{j\dag})T^{j}\\
&  =\left(  1-\frac{3Q}{2}\right)  \omega_{A}+\frac{Q}{2}\left(  X\omega
_{A}X+Y\omega_{A}Y+Z\omega_{A}Z\right)  ,
\label{eq:depolar-eve-attack-6-state}%
\end{align}
where%
\begin{equation}
Q=\frac{1}{3}\left(  Q_{x}+Q_{y}+Q_{z}\right)  .
\end{equation}
The channel $\overline{\mathcal{N}}_{A\rightarrow B}^{Q}$ above is the well
known quantum depolarizing channel. It is completely positive for
$Q\in[0,2/3]$, and it is entanglement breaking when $Q\in[1/3,2/3]$. An
entanglement breaking channel is not capable of distilling secret key, as
argued in \cite{CLL04}, with a strong limitation for the finite-key regime
established in \cite{WTB16}.

Now let us consider the BB84 protocol. Let $H$ denote the following unitary
Hadamard transformation:%
\begin{equation}
H\coloneqq |+\rangle\!\langle0|+|-\rangle\!\langle1|=\frac{1}{\sqrt{2}}%
\begin{bmatrix}
1 & 1\\
1 & -1
\end{bmatrix}
,
\end{equation}
so that $H$ changes the Pauli $X$ and $Z$ bases as $HXH^{\dag}=Z$ and
$HZH^{\dag}=X$.\ Due to the fact that the $H$ gate swaps information encoded
into the Pauli $X$ and $Z$ eigenstates around, discarding of basis information
in the BB84 protocol is equivalent to a further channel twirl as
\cite[Section~2.2.7]{Myhr10}%
\begin{align}
\overline{\mathcal{N}}_{A\rightarrow B}^{\text{BB84},Q}(\omega_{A})  &
\coloneqq \frac{1}{2}\sum_{j\in\left\{  0,1\right\}  }H^{j\dag}\overline{\mathcal{N}%
}_{A\rightarrow B}(H^{j}\omega_{A}H^{j\dag})H^{j}\\
&  =\left(  1-2Q+s\right)  \omega_{A}+\left(  Q-s\right)  X\omega
_{A}X\nonumber\\
&  \qquad+sY\omega_{A}Y+\left(  Q-s\right)  Z\omega_{A}Z,
\end{align}
where in this case%
\begin{align}
Q  &  =\frac{1}{2}\left(  Q_{x}+Q_{z}\right)  ,\\
s  &  =Q-Q_{y}/2.
\end{align}
The channel $\overline{\mathcal{N}}_{A\rightarrow B}^{\text{BB84},Q}$ is thus
known as the BB84 channel in the literature \cite{SS08}. 
Observe that $s
\in[2Q-1,Q]$ in order for complete positivity to hold. By
computing the Choi state of this channel and using the condition for
separability from \cite{ADH08}, we find that the BB84 channel is entanglement
breaking for $s\in[0,1/2] $ and $Q\in[(s+1/2)/2,(s+1)/2]$. As before, it is
not possible to distill secret key from the BB84 channel when it is
entanglement breaking.

We call a protocol in which the basis information is discarded a
\textquotedblleft coarse-grained protocol\textquotedblright. We find analytic
expressions for the key distillation rate of the coarse-grained BB84 and
six-state protocols in the following sections.

\subsubsection{Six-state protocol}

For the six-state protocol, we have $\mathcal{B}=\mathcal{B}_{\text{six-state}%
}=\{0,1,2\}$, corresponding to the $X$, $Y$, and $Z$ bases. We typically take
$q_{b}^{A}=\frac{1}{3}=q_{b}^{B}$ for all $b\in\mathcal{B}$, so that
$p_{\text{sift}}=\frac{1}{3}$. The joint distribution shared by Alice and Bob
after sifting is as follows:%
\begin{equation}
p_{X_{1}X_{2}Y_{1}Y_{2}}^{\text{6-state}|\text{sift}}(b,x,b,y)=\frac{q_{b}%
^{A}q_{b}^{B}}{2p_{\text{sift}}}\operatorname{Tr}[\Pi_{y}^{b}\overline
{\mathcal{N}}_{A\rightarrow B}(\rho_{A}^{b,x})],
\end{equation}
for $b\in\{0,1,2\}$ and $x,y\in\{0,1\}$. Each of the relevant entries is given
by
\begin{align}
p_{X_{1}X_{2}Y_{1}Y_{2}}^{\text{6-state}|\text{sift}}(0,0,0,0)  &  =\frac
{1}{6}(1-Q_{x}),\\
p_{X_{1}X_{2}Y_{1}Y_{2}}^{\text{6-state}|\text{sift}}(0,0,0,1)  &  =\frac
{1}{6}Q_{x},\\
p_{X_{1}X_{2}Y_{1}Y_{2}}^{\text{6-state}|\text{sift}}(0,1,0,0)  &  =\frac
{1}{6}Q_{x},\\
p_{X_{1}X_{2}Y_{1}Y_{2}}^{\text{6-state}|\text{sift}}(0,1,0,1)  &  =\frac
{1}{6}(1-Q_{x}),\\
p_{X_{1}X_{2}Y_{1}Y_{2}}^{\text{6-state}|\text{sift}}(1,0,1,0)  &  =\frac
{1}{6}(1-Q_{z}),\\
p_{X_{1}X_{2}Y_{1}Y_{2}}^{\text{6-state}|\text{sift}}(1,0,1,1)  &  =\frac
{1}{6}Q_{z},\\
p_{X_{1}X_{2}Y_{1}Y_{2}}^{\text{6-state}|\text{sift}}(1,1,1,0)  &  =\frac
{1}{6}Q_{z},\\
p_{X_{1}X_{2}Y_{1}Y_{2}}^{\text{6-state}|\text{sift}}(1,1,1,1)  &  =\frac
{1}{6}(1-Q_{z}),\\
p_{X_{1}X_{2}Y_{1}Y_{2}}^{\text{6-state}|\text{sift}}(2,0,2,0)  &  =\frac
{1}{6}(1-Q_{y}),\\
p_{X_{1}X_{2}Y_{1}Y_{2}}^{\text{6-state}|\text{sift}}(2,0,2,1)  &  =\frac
{1}{6}Q_{y},\\
p_{X_{1}X_{2}Y_{1}Y_{2}}^{\text{6-state}|\text{sift}}(2,1,2,0)  &  =\frac
{1}{6}Q_{y},\\
p_{X_{1}X_{2}Y_{1}Y_{2}}^{\text{6-state}|\text{sift}}(2,1,2,1)  &  =\frac
{1}{6}(1-Q_{y}).
\end{align}

Let us define the average QBER as
\begin{equation}
Q\coloneqq\frac{1}{3}(Q_{x}+Q_{y}+Q_{z}).
\end{equation}
If Alice and Bob discard the basis information in $X_{1}$ and $Y_{1}$, then
the resulting probability distribution is
\begin{align}
p_{X_{2}Y_{2}}^{\text{6-state}|\text{sift}}(0,0)  &  =\frac{1}{2}%
(1-Q),\label{eq-6state_pXY_discard_1}\\
p_{X_{2}Y_{2}}^{\text{6-state}|\text{sift}}(0,1)  &  =\frac{1}{2}%
Q,\label{eq-6state_pXY_discard_2}\\
p_{X_{2}Y_{2}}^{\text{6-state}|\text{sift}}(1,0)  &  =\frac{1}{2}%
Q,\label{eq-6state_pXY_discard_3}\\
p_{X_{2}Y_{2}}^{\text{6-state}|\text{sift}}(1,1)  &  =\frac{1}{2}(1-Q).
\label{eq-6state_pXY_discard_4}%
\end{align}
In other words, when discarding the basis information, Alice and Bob's data
can be characterized by the single parameter $Q$.

For the six-state protocol, there is a reliable estimate of Eve's collective
attack. This means that the uncertainty set $\mathcal{S}$ discussed in
Section~\ref{sec:parameter-est-QKD} has cardinality equal to one. With the
further assumption of discarding basis information, the average QBER $Q$
uniquely identifies the attack of Eve, as in
\eqref{eq:depolar-eve-attack-6-state}. By following the derivations in
Appendix~\ref{app:mut-hol-info-6-bb84}, we find that the mutual and Holevo
informations and variances as a function of the average QBER~$Q$ are as
follows:%
\begin{align}
I(X;Y)_{\rho}  &  =1-h_{2}(Q),\\
V(X;Y)_{\rho}  &  =Q(1-Q)\left(  \log_{2}\!\left(  \frac{1-Q}{Q}\right)
\right)  ^{2},\\
I(X;E)_{\rho}  &  =-\left(  1-\frac{3Q}{2}\right)  \log_{2}\!\left(
1-\frac{3Q}{2}\right) \nonumber\\
&  \qquad-\frac{3Q}{2}\log_{2}\!\left(  \frac{Q}{2}\right)  -h_{2}(Q),\\
V(X;E)_{\rho}  &  =Q+\left(  1-\frac{3Q}{2}\right)  \left(  \log_{2}\!\left(
\frac{1-\frac{3Q}{2}}{1-Q}\right)  \right)  ^{2}\nonumber\\
&  \qquad+\frac{Q}{2}\left(  \log_{2}\!\left(  \frac{\frac{Q}{2}}{1-Q}\right)
\right)  ^{2}-I(X;E)_{\rho}^{2},
\end{align}
where
\begin{equation}
	h_2(Q)\coloneqq -Q\log_2(Q)-(1-Q)\log_2(1-Q)
\end{equation}
is the binary entropy. The achievable distillable key with direct reconciliation is then given by
evaluating the general formula in
\eqref{eq:direct-reconciliation-second-order}:%
\begin{multline}
K_{\text{6-state}}(Q,n,\varepsilon_{\text{I}},\varepsilon_{\text{II}%
})=1+\left(  1-\frac{3Q}{2}\right)  \log_{2}\!\left(  1-\frac{3Q}{2}\right) \\
+\frac{3Q}{2}\log_{2}\!\left(  \frac{Q}{2}\right) \\
\quad+\sqrt{\frac{Q(1-Q)}{n}}\log_{2}\!\left(  \frac{1-Q}{Q}\right)  \Phi
^{-1}(\varepsilon_{\text{I}})\\
+\sqrt{\frac{V(X;E)_{\rho}}{n}}\Phi^{-1}({\varepsilon_{\text{II}}^2}).
\end{multline}
We note that the first-order rate term  in the above expression is equal to the known asymptotic key rate for the six-state protocol and coincides with the result from \cite{L01}.

Figure~\ref{fig:six-state-performance} plots the second-order coding rate for
key distillation using the six-state protocol. These rates can be compared
with the finite-key analysis from \cite{AMKB11}, but the comparison is not
necessarily fair, due to our assumption of reliable parameter estimation.

\begin{figure}
\centering
\includegraphics[width=\columnwidth]{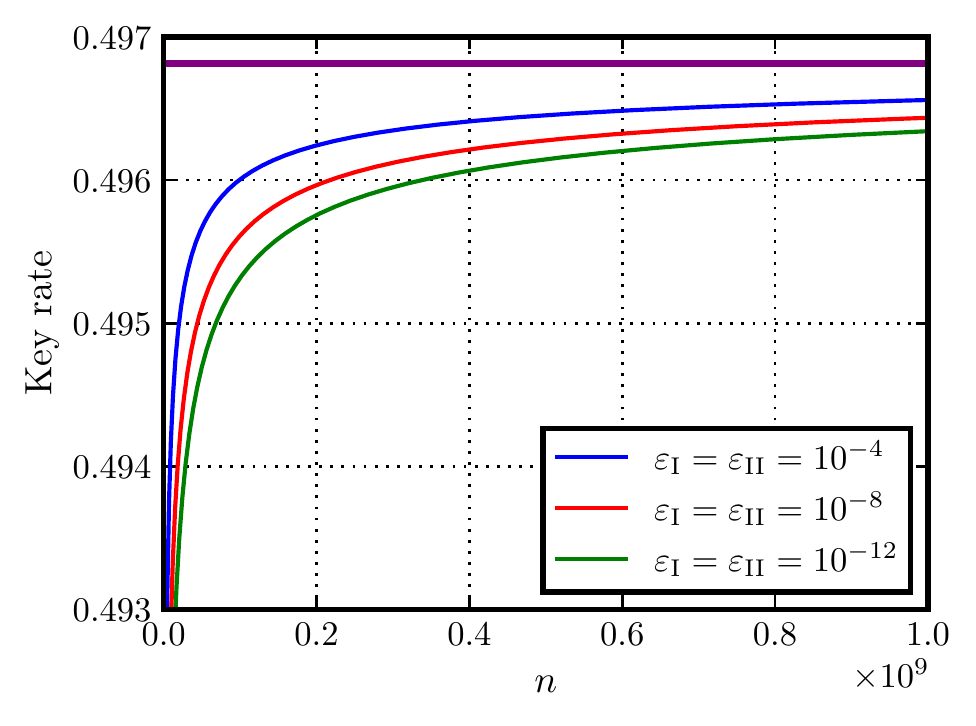}\caption{Distillable key
rates for the coarse-grained six-state protocol after sifting with $Q=0.05$.
The purple solid line indicates the asymptotic key rate.}%
\label{fig:six-state-performance}%
\end{figure}

We finally note that incorporating the methods of \cite{KR08} can improve the
key rate.

\subsubsection{BB84 DV-QKD protocol}

For the BB84 protocol, we have $\mathcal{B}=\mathcal{B}_{\text{BB84}}%
=\{0,1\}$, corresponding to the $X$ and $Z$ bases. We typically take
$q_{b}^{A}=\frac{1}{2}=q_{b}^{B}$ for all $b\in\mathcal{B}$, so that
$p_{\text{sift}}=\frac{1}{2}$. The relevant probability distribution is
\begin{equation}
p_{X_{1}X_{2}Y_{1}Y_{2}}^{\text{BB84}|\text{sift}}(b,x,b,y)=\frac{q_{b}%
^{A}q_{b}^{B}}{2p_{\text{sift}}}\operatorname{Tr}[\Pi_{y}^{b}\overline
{\mathcal{N}}_{A\to B}(\rho_{A}^{b,x})],
\end{equation}
for $b\in\{0,1\}$ and $x,y\in\{0,1\}$. We then have
\begin{align}
p_{X_{1}X_{2}Y_{1}Y_{2}}^{\text{BB84}|\text{sift}}(0,0,0,0)  &  =\frac{1}%
{4}(1-Q_{x}),\\
p_{X_{1}X_{2}Y_{1}Y_{2}}^{\text{BB84}|\text{sift}}(0,0,0,1)  &  =\frac{1}%
{4}Q_{x},\\
p_{X_{1}X_{2}Y_{1}Y_{2}}^{\text{BB84}|\text{sift}}(0,1,0,0)  &  =\frac{1}%
{4}Q_{x},\\
p_{X_{1}X_{2}Y_{1}Y_{2}}^{\text{BB84}|\text{sift}}(0,1,0,1)  &  =\frac{1}%
{4}(1-Q_{x}),\\
p_{X_{1}X_{2}Y_{1}Y_{2}}^{\text{BB84}|\text{sift}}(1,0,1,0)  &  =\frac{1}%
{4}(1-Q_{z}),\\
p_{X_{1}X_{2}Y_{1}Y_{2}}^{\text{BB84}|\text{sift}}(1,0,1,1)  &  =\frac{1}%
{4}Q_{z},\\
p_{X_{1}X_{2}Y_{1}Y_{2}}^{\text{BB84}|\text{sift}}(1,1,1,0)  &  =\frac{1}%
{4}Q_{z},\\
p_{X_{1}X_{2}Y_{1}Y_{2}}^{\text{BB84}|\text{sift}}(1,1,1,1)  &  =\frac{1}%
{4}(1-Q_{z})
\end{align}
%The mutual information of the full distribution $p_{X_{1}X_{2}Y_{1}Y_{2}%
%}^{\text{BB84}|\text{sift}}$ is
%\begin{equation}
%I(X_{1}X_{2};Y_{1}Y_{2})_{\rho^{\text{BB84}|\text{sift}}}=\frac{1}{2}%
%(4-h_{2}(Q_{x})-h_{2}(Q_{z})).
%\end{equation}

Let us define the average QBER as
\begin{equation}
Q\coloneqq\frac{1}{2}(Q_{x}+Q_{z}).
\end{equation}
If Alice and Bob discard the basis information in $X_{1}$ and $Y_{1}$, then
the probability distribution is
\begin{align}
p_{X_{2}Y_{2}}^{\text{BB84}|\text{sift}}(0,0)  &  =\frac{1}{2}%
(1-Q),\label{eq-BB84_pXY_discard_1}\\
p_{X_{2}Y_{2}}^{\text{BB84}|\text{sift}}(0,1)  &  =\frac{1}{2}%
Q,\label{eq-BB84_pXY_discard_2}\\
p_{X_{2}Y_{2}}^{\text{BB84}|\text{sift}}(1,0)  &  =\frac{1}{2}%
Q,\label{eq-BB84_pXY_discard_3}\\
p_{X_{2}Y_{2}}^{\text{BB84}|\text{sift}}(1,1)  &  =\frac{1}{2}(1-Q).
\label{eq-BB84_pXY_discard_4}%
\end{align}
In other words, when discarding the basis information, Alice and Bob's
classical data can be characterized using the single parameter $Q$.

For the parameter estimation step of the protocol, it is possible for Alice
and Bob to determine the QBERs $Q_{x}$ and $Q_{z}$\ reliably. However, it is
not possible to estimate $Q_{y}$ because the encoding and decoding do not
involve the eigenbasis of the Pauli $Y$ operator. Thus, the uncertainty set
$\mathcal{S}$ in this case consists of all Pauli channels satisfying
\eqref{eq:QBER-to-pauli-1}--\eqref{eq:QBER-to-pauli-3} for fixed $Q_{x}$ and
$Q_{z}$. If we further simplify the protocol by throwing away the basis
information, then the analysis simplifies. By following the derivations in
Appendix~\ref{app:mut-hol-info-6-bb84}, we find that the mutual and Holevo
informations and variances as a function of the average QBER~$Q$ and the
optimization parameter $s\in\left[  0,Q\right]  $ are as follows:%
\begin{align}
I(X;Y)_{\rho}  &  =1-h_{2}(Q),\\
V(X;Y)_{\rho}  &  =Q(1-Q)\left(  \log_{2}\!\left(  \frac{1-Q}{Q}\right)
\right)  ^{2},\\
I(X;E)_{\rho}  &  =H(\{1-2Q+s,Q-s,Q-s,s\})\nonumber\\
&  \qquad-h_{2}(Q),\\
V(X;E)_{\rho}  &  =(1-2Q+s)\left(  \log_{2}\!\left(  \frac{1-2Q+s}%
{1-Q}\right)  \right)  ^{2}\nonumber\\
&  \qquad+(Q-s)\left(  \log_{2}\!\left(  \frac{Q-s}{1-Q}\right)  \right)
^{2}\nonumber\\
&  \qquad+(Q-s)\left(  \log_{2}\!\left(  \frac{Q-s}{Q}\right)  \right)
^{2}\nonumber\\
&  \qquad+s\left(  \log_{2}\!\left(  \frac{s}{Q}\right)  \right)
^{2}-I(X;E)_{\rho}^{2}.
\end{align}

The achievable key rate for direct reconciliation is then given by evaluating
the general formula in \eqref{eq:direct-reconciliation-second-order} and
performing an optimization over the parameter $s\in\left[  0,Q\right]  $. We note that the first-order rate term $I(X;Y)-I(X;E)$  is equal to the known asymptotic key rate for the BB84 protocol and coincides with the result from \cite{SP00}.

Figure~\ref{fig:BB84-performance} plots the second-order coding rate for key
distillation using the BB84 protocol. These rates can be compared with the
finite-key analysis from \cite{Tomamichel2017largelyself}, but the comparison
is not necessarily fair, due to our assumption of reliable parameter estimation.

\begin{figure}
\centering
\includegraphics[width=\columnwidth]{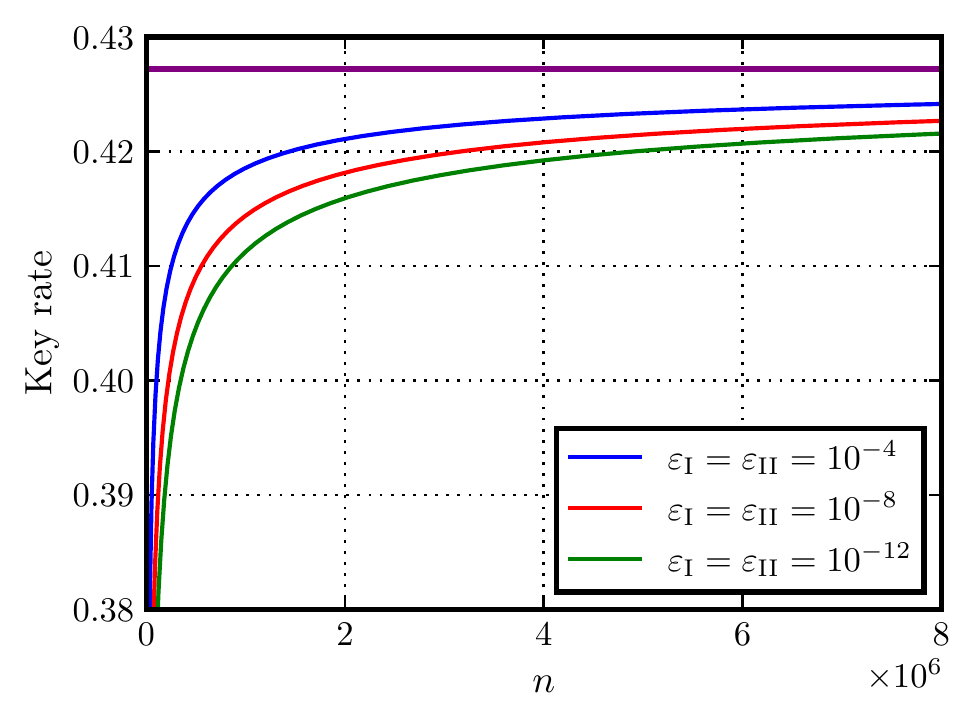} \caption{Distillable key rates
for the coarse-grained BB84 protocol after sifting with $Q=0.05$. The purple
solid line indicates the asymptotic key rate.}%
\label{fig:BB84-performance}%
\end{figure}

We again note that incorporating the methods of \cite{smith:170502} can
improve the key rate.

\subsubsection{CV-QKD\ protocol with coherent states and heterodyne detection}

We finally analyze the performance of the CV-QKD\ Gaussian modulation protocol
with reverse reconciliation. This protocol involves Gaussian modulation of
coherent states by Alice and heterodyne detection by Bob, and it is a
particular instance of the prepare-and-measure protocol from
Section~\ref{sec:gen-pm-prot}.

The encoding ensemble $\mathcal{E}_{A}$ in this case is%
\begin{equation}
\left\{  p_{\overline{n}}(\alpha),|\alpha\rangle\!\langle\alpha|\right\}
_{\alpha\in\mathbb{C}}, \label{eq:CV-Gaussian-encoding-ens}%
\end{equation}
where $\overline{n} > 0$, $p_{\overline{n}}(\alpha)$ is an isotropic complex
Gaussian distribution:%
\begin{equation}
p_{\overline{n}}(\alpha)\coloneqq \frac{1}{\pi\overline{n}}\exp\left(  -\left\vert
\alpha\right\vert ^{2}/\overline{n}\right)  ,
\label{eq:CV-QKD-encoding-ensemble}%
\end{equation}
and $|\alpha\rangle\!\langle\alpha|$ is a coherent state \cite{GK04}\ such that%
\begin{equation}
|\alpha\rangle\coloneqq e^{-\left\vert \alpha\right\vert ^{2}/2}\sum_{n=0}^{\infty
}\frac{\alpha^{n}}{\sqrt{n!}}|n\rangle,
\end{equation}
with $\left\{  |n\rangle\right\}  _{n=0}^{\infty}$ the orthonormal photon
number basis. The decoding POVM\ for Bob is heterodyne detection \cite{GK04},
given by%
\begin{equation}
\left\{  \frac{1}{\pi}|\alpha\rangle\!\langle\alpha|\right\}  _{\alpha
\in\mathbb{C}}. \label{eq:Bob-decoding-POVM}%
\end{equation}
Let $X$ be a complex random variable corresponding to the choice of $\alpha$
in Alice's encoding, and let $Y$ be a complex random variable corresponding to
Bob's measurement outcome.

The encoding and decoding act on a single bosonic mode, which implies that the
channel $\mathcal{N}_{A\rightarrow B}$\ that Eve employs is a single-mode
bosonic channel.

The channel twirling that Alice and Bob employ in this case is the phase
symmetrization from \cite{KGW19}. In particular, before Alice sends out a
coherent state, she applies the unitary $e^{-i\hat{n}\phi}$, with the phase
$\phi$ selected uniformly at random from $\left\{  0,\pi/2,\pi,3\pi/2\right\}
$. She communicates this choice to Bob, who then applies $e^{i\hat{n}\phi}$ to
his received mode. They both then discard the classical memory of which phase
$\phi$ was applied. This phase symmetrization significantly reduces the number
of parameters that Alice and Bob need to estimate in the parameter estimation
part of the protocol. Like the other protocols, this phase symmetrization need
not be applied in the quantum domain because the encoding ensemble in
\eqref{eq:CV-Gaussian-encoding-ens} and the decoding
POVM\ \eqref{eq:Bob-decoding-POVM} are invariant under the actions of Alice
and Bob mentioned above. Thus, it can be carried out in classical processing
that keeps track of the twirl.

During the parameter estimation step of the protocol, Alice and Bob estimate
\begin{align}
\gamma_{12}  &  \coloneqq \mathbb{E}\{\left(  X-\mathbb{E}\{X\}\right)  ^{\ast}\left(
Y-\mathbb{E}\{Y\}\right)  \},\label{eq:cv-params-1}\\
\gamma_{22}  &  \coloneqq \mathbb{E}\{\left\vert Y-\mathbb{E}\{Y\}\right\vert ^{2}\}.
\label{eq:cv-params-3}%
\end{align}
We suppose here that the photon number variance of the channel output system $B$ is finite, in order to ensure that reliable estimation of the parameters $\gamma_{12}$ and $\gamma_{22}$ is possible.
Here we also assume that they estimate $p_{Y|X}(y|x)$. The following parameter
is known
\begin{equation}
\gamma_{11} \coloneqq \mathbb{E}\{\left\vert X-\mathbb{E}\{X\}\right\vert ^{2}\},
\end{equation}
being a function of Alice's encoding ensemble. The uncertainty set
$\mathcal{S}$ in this case then consists of all single-mode bosonic channels, with system $B$ having finite photon number variance,
that lead to the estimates $\gamma_{12}$ and $\gamma_{22}$, as well as
$p_{Y|X}(y|x)$. The achievable key rate is given by evaluating the formula in
\eqref{eq:reverse-reconciliation-second-order}\ as a function of the encoding
in \eqref{eq:CV-Gaussian-encoding-ens}, the decoding in
\eqref{eq:Bob-decoding-POVM}, and the set $\mathcal{S}$ mentioned above. To
evaluate this formula, we apply the perturbative approach from \cite[Lemmas~63
and 64]{polyanskiy10}, which is valid for sufficiently large $n$. In this
approach, we optimize the first-order Holevo information term of Eve first,
and then among all of the channels optimizing the first-order term, we
optimize the second-order Holevo information variance of Eve. In this case,
the Gaussian extremality theorem of \cite{Wolf06}, as observed in
\cite{PC2006}, implies that a Gaussian attack achieves the optimal first-order
Holevo information of Eve. The same conclusion has been reached in
\cite{NGA06} by a different line of reasoning. In fact, by examining the proof
in \cite{NGA06} and employing the faithfulness of quantum relative entropy, we
conclude that for every non-Gaussian attack, there is a Gaussian attack that
achieves a strictly higher Holevo information of Eve. Thus, it suffices to
optimize Eve's Holevo information over Gaussian attacks exclusively, and among
all of these Gaussian attacks, there is a unique one that achieves the optimum
\cite[Appendix~D]{Lev15}, which is a thermal channel consistent with the
observed estimates $\gamma_{11}$, $\gamma_{12}$, and $\gamma_{22}$. So
applying the perturbative approach mentioned above, there is no need to
perform a further optimization of the Holevo information variance of Eve, and
we just evaluate it with respect to the optimal and unique Gaussian attack.

We stress here that, even though the actual attack of Eve need not be a
Gaussian channel, the Gaussian optimality theorem and the perturbative
extension of it is useful in order to bound Eve's information from above.
Without this result, the optimization would be too difficult (perhaps
impossible) because the underlying Hilbert space is infinite-dimensional.

Let us analyze the performance of the above protocol in a particular physical
scenario. Suppose that the underlying physical channel is indeed a thermal
channel $\mathcal{L}_{A\rightarrow B}^{\eta,N_{B}}$ \cite{Ser17},
characterized by transmissivity $\eta\in(0,1)$ and environment thermal photon
number $N_{B}$. However, Alice and Bob are not aware of this, as is the usual
case in a QKD protocol. They execute the above protocol, and after the
parameter estimation phase, suppose that they have estimated the classical
channel $p_{Y|X}(y|x)$ and the parameters in
\eqref{eq:cv-params-1}--\eqref{eq:cv-params-3}. From the parameters in
\eqref{eq:cv-params-1}--\eqref{eq:cv-params-3}, they conclude that the optimal
Gaussian attack of Eve is a thermal channel of transmissivity $\eta$ and
environment thermal photon number $N_{B}$. Let $\mathcal{U}_{A\rightarrow
BE}^{\mathcal{L}^{\eta,N_{B}}}$ be an isometric channel extending
$\mathcal{L}_{A\rightarrow B}^{\eta,N_{B}}$. An isometric channel extending
$\mathcal{L}_{A\rightarrow B}^{\eta,N_{B}}$ can be physically realized by the
action of a beamsplitter of transmissivity $\eta$ acting on the input mode $A$
and one share $E_{1}$ of a two-mode squeezed vacuum state \cite{CG06,CGH06},
as%
\begin{equation}
\mathcal{U}_{A\rightarrow BE}^{\mathcal{L}^{\eta,N_{B}}}(\omega_{A}%
)=\mathcal{B}_{AE_{1}\rightarrow BE_{1}}^{\eta}(\omega_{A}\otimes\psi
_{E_{1}E_{2}}^{N_{B}}),
\end{equation}
where $\omega_{A}$ is the input state, $\mathcal{B}_{AE_{1}\rightarrow BE_{1}%
}^{\eta}$ denotes the beamsplitter channel of transmissivity $\eta$, the
system $E$ consists of modes $E_{1}$ and $E_{2}$, and the two-mode squeezed
vacuum state $\psi_{E_{1}E_{2}}^{N_{B}}\coloneqq |\psi^{N_{B}}\rangle\!\langle
\psi^{N_{B}}|_{E_{1}E_{2}}$ is defined from%
\begin{equation}
|\psi^{N_{B}}\rangle_{E_{1}E_{2}}\coloneqq \frac{1}{\sqrt{N_{B}+1}}\sum_{n=0}^{\infty
}\left(  \sqrt{\frac{N_{B}}{N_{B}+1}}\right)  ^{n}|n\rangle_{E_{1}}%
|n\rangle_{E_{2}}.
\end{equation}

We now calculate the various quantities in
\eqref{eq:reverse-reconciliation-second-order} for this case and for a reverse
reconciliation protocol. We begin by determining the mutual information
$I(X;Y)$ and the mutual information variance $V(X;Y)$. The classical channel
$p_{Y|X}(y|x)$ induced by heterodyne detection of a Gaussian-distributed
ensemble of coherent states transmitted through the thermal channel
$\mathcal{L}_{A\rightarrow B}^{\eta,N_{B}}$\ is as follows (see, e.g.,
\cite{Guha04,Savov12}):%
\begin{equation}
p_{Y|X}(y|x)=\frac{\exp\left(  -\frac{\left\vert y-\sqrt{\eta}x\right\vert
^{2}}{\pi\left(  \left(  1-\eta\right)  N_{B}+1\right)  }\right)  }{\pi\left(
\left(  1-\eta\right)  N_{B}+1\right)  }, \label{eq:induced-Gaussian-channel}%
\end{equation}
where $x,y\in\mathbb{C}$. As mentioned previously, we assume that this
classical channel $p_{Y|X}$ is reliably estimated during the parameter
estimation step \footnote{When estimating $p_{Y|X}$,  we can  bin the outcomes of the joint distribution $p_{XY}$ to get a discrete distribution $p_{X_m Y_m}$. Then we can estimate this finite distribution $p_{X_m Y_m}$ instead, which is reasonable since there are a finite number of parameters, and plug the values into the classical mutual information $I(X_m;Y_m)$ and classical mutual information variance $V(X_m;Y_m)$ for a second-order coding rate for information reconciliation. Then the limits $I(X_m;Y_m) \to I(X;Y)$ and $V(X_m;Y_m)\to V(X;Y)$ hold as the bin size gets smaller.}. Then it follows that%
\begin{align}
I(X;Y)  &  =\log_{2}( 1+P) ,\\
V(X;Y)  &  = \frac{1}{\ln^{2} 2 } \frac{P\left(  P+2\right)  }{\left(
P+1\right)  ^{2}},
\end{align}
where%
\begin{equation}
P\coloneqq \frac{\eta\bar{n}}{\left(  1-\eta\right)  N_{B}+1}.
\end{equation}
The formulas above follow from the fact that the channel $p_{Y|X}(y|x)$ in
\eqref{eq:induced-Gaussian-channel} can be understood as two independent real
Gaussian channels each with signal-to-noise ratio $P$, and then by applying
Shannon's formula \cite{bell1948shannon}\ and the mutual information variance
formula from \cite{polyanskiy10,TT15} (with a prefactor of two to account for
the two independent channels).

We now turn to calculating the Holevo information $I(Y;E)$ and the Holevo
information variance $V(Y;E)$, the latter of which can be calculated as a
special case of Proposition~\ref{prop:gaussian-formulas-HI-Hi-var} below. For
a Gaussian-distributed ensemble of coherent states as in
\eqref{eq:CV-QKD-encoding-ensemble}, the expected input density operator is a
thermal state of the following form~\cite{Ser17}:%
\begin{equation}
\theta(\bar{n})\coloneqq \frac{1}{\bar{n}+1}\sum_{n=0}^{\infty}\left(  \frac{\bar{n}%
}{\bar{n}+1}\right)  ^{n}|n\rangle\!\langle n|,
\end{equation}
with covariance matrix $V(\bar{n})\coloneqq \left(  2\bar{n}+1\right)  I_{2}$, where
$I_{2}$ is the $2\times2$ identity matrix. It then follows that the covariance
matrix of the state $\rho_{BE}\coloneqq \mathcal{U}_{A\rightarrow BE}^{\mathcal{L}%
^{\eta,N_{B}}}(\theta(\bar{n}))$ is%
\begin{align}
V_{\rho_{BE}} & \coloneqq
\begin{bmatrix}
V_{B} & V_{BE}\\
V_{BE} & V_{E}%
\end{bmatrix}
\\
&=\left(  B(\eta)\oplus I_{2}\right)  \left(  V(\bar{n})\oplus V_{E}%
(N_{B})\right)  \left(  B^{T}(\eta)\oplus I_{2}\right)  ,
\end{align}
where%
\begin{align}
B(\eta) &  \coloneqq %
\begin{bmatrix}
\sqrt{\eta}I_{2} & \sqrt{1-\eta}I_{2}\\
-\sqrt{1-\eta}I_{2} & \sqrt{\eta}I_{2}%
\end{bmatrix}
,\\
V_{E}(N_{B}) &  \coloneqq %
\begin{bmatrix}
\left(  2N_{B}+1\right)  I_{2} & 2\sqrt{N_{B}(N_{B}+1)}\sigma_{Z}\\
2\sqrt{N_{B}(N_{B}+1)}\sigma_{Z} & \left(  2N_{B}+1\right)  I_{2}%
\end{bmatrix}
.
\end{align}
For this example, note that Eve's actual attack and Alice and Bob's estimation
of it coincide, due to our assumption that the underlying channel is a thermal
channel. We then obtain the covariance matrix of the reduced state $\rho
_{E}=\operatorname{Tr}_{B}[\rho_{BE}]$ from the above, which can be used to
calculate the entropy $H(E)_{\rho}$. The covariance matrix of the reduced
state $\rho_{E}^{y}$\ of Eve's system, given the outcome $y$ of Bob's
heterodyne detection is as follows:%
\begin{equation}
V_{E|Y}\coloneqq V_{E}-V_{BE}\left(  V_{B}+I_{2}\right)  ^{-1}V_{BE}^{T},
\end{equation}
and the probability of obtaining the outcome $y$ is%
\begin{equation}
p_{Y}(y)\coloneqq \frac{\exp\left(  -\frac{\left\vert y\right\vert ^{2}}{\pi\left(
\eta\bar{n}+\left(  1-\eta\right)  N_{B}+1\right)  }\right)  }{\pi\left(
\eta\bar{n}+\left(  1-\eta\right)  N_{B}+1\right)  }.
\end{equation}
This allows us to compute the conditional entropy as
\begin{equation}
H(E|Y)=\int dy\ p_{Y}(y)H(E)_{\rho^{y}}.
\end{equation}
Now combining $H(E)_{\rho}-H(E|Y)=I(Y;E)$, we can calculate Eve's Holevo information.

To calculate the Holevo information variance $V(Y;E)$, we apply the formula
from Proposition~\ref{prop:gaussian-formulas-HI-Hi-var} below. To understand
this formula, we first provide a definition for a Gaussian ensemble of quantum
Gaussian states.

\begin{definition}
[Gaussian ensemble]\label{def:gaussian-ensemble}A Gaussian ensemble of
$m$-mode Gaussian states consists of a Gaussian prior probability distribution
of the form%
\begin{align}
p_{Y}(y) &  \coloneqq \mathcal{N}(\mu,\Sigma)\\
&  \coloneqq \frac{\exp\left(  -\frac{1}{2}\left(  y-\mu\right)  ^{T}\Sigma
^{-1}\left(  y-\mu\right)  \right)  }{\sqrt{\left(  2\pi\right)  ^{2n}%
\det(\Sigma)}},
\end{align}
and a set $\left\{  \rho_{E}^{y}\right\}  _{y\in\mathbb{R}^{2n}}$ of Gaussian
states, each having $2m\times1$ mean vector $Wy+\nu$ and $2m\times2m$ quantum
covariance matrix $V$. In the above, $y$ is a $2n\times1$ vector, $\mu$ is the
$2n\times1$ mean vector of the random vector $Y$, $\Sigma$ is the $2n\times2n$
covariance matrix of the random vector$\ Y$, $W$ is a $2m\times2n$ matrix, and
$\nu$ is a $2m$-dimensional vector.
\end{definition}

\begin{proposition}
\label{prop:gaussian-formulas-HI-Hi-var}The Holevo information of the ensemble
in Definition~\ref{def:gaussian-ensemble} is given by%
\begin{multline}
I(Y;E)=\\
\frac{1}{2\ln2}\left[  \ln\!\left(  \frac{Z_{E}}{Z}\right)  +\frac{1}%
{2}\operatorname{Tr}[V\Delta]+\operatorname{Tr}[W\Sigma W^{T}G_{E}]\right]  ,
\end{multline}
and the Holevo information variance of the ensemble in
Definition~\ref{def:gaussian-ensemble} is given by%
\begin{multline}
V(Y;E)=\frac{1}{8\ln^{2}2}\left[  \operatorname{Tr}[\left(  \Delta V\right)
^{2}]+\operatorname{Tr}[\left(  \Delta\Omega\right)  ^{2}]\right]  \\
+\frac{1}{2\ln^{2}2}\left[  \operatorname{Tr}[W\Sigma W^{T}G_{E}%
VG_{E}]+\operatorname{Tr}[\left(  W\Sigma W^{T}G_{E}\right)  ^{2}]\right]  ,
\end{multline}
where%
\begin{align}
Z &  \coloneqq \det(\left[  V+i\Omega\right]  /2),\\
G &  \coloneqq 2i\Omega\operatorname{arcoth}(iV\Omega),\\
V_{E} &  \coloneqq V+2W\Sigma W^{T},\\
Z_{E} &  \coloneqq \det(\left[  V_{E}+i\Omega\right]  /2),\\
G_{E} &  \coloneqq 2i\Omega\operatorname{arcoth}(iV_{E}\Omega),\\
\Delta &  \coloneqq G_{E}-G.
\end{align}

\end{proposition}

See Appendix~\ref{app:Holevo-info-and-var-Gaussian} for a review of quantum
Gaussian states and measurements and for a proof of
Proposition~\ref{prop:gaussian-formulas-HI-Hi-var}.

Putting everything above together, we find that the achievable secret key rate
is given by evaluating the above quantities in the formula in
\eqref{eq:reverse-reconciliation-second-order}.
Figure~\ref{fig:CVQKD-performance} plots the second-order coding rate for key
distillation using the CV-QKD reverse reconciliation protocol. One can compare
these rates to the finite-key analysis presented in \cite{Lev15}, but the
comparison is not necessarily fair due to our assumption of reliable parameter estimation. {We also note that the first-order term in \eqref{eq:reverse-reconciliation-second-order} is equal to the known asymptotic key rate from \cite{Lev15}.}

\begin{figure}
\centering
\includegraphics[width=\columnwidth]{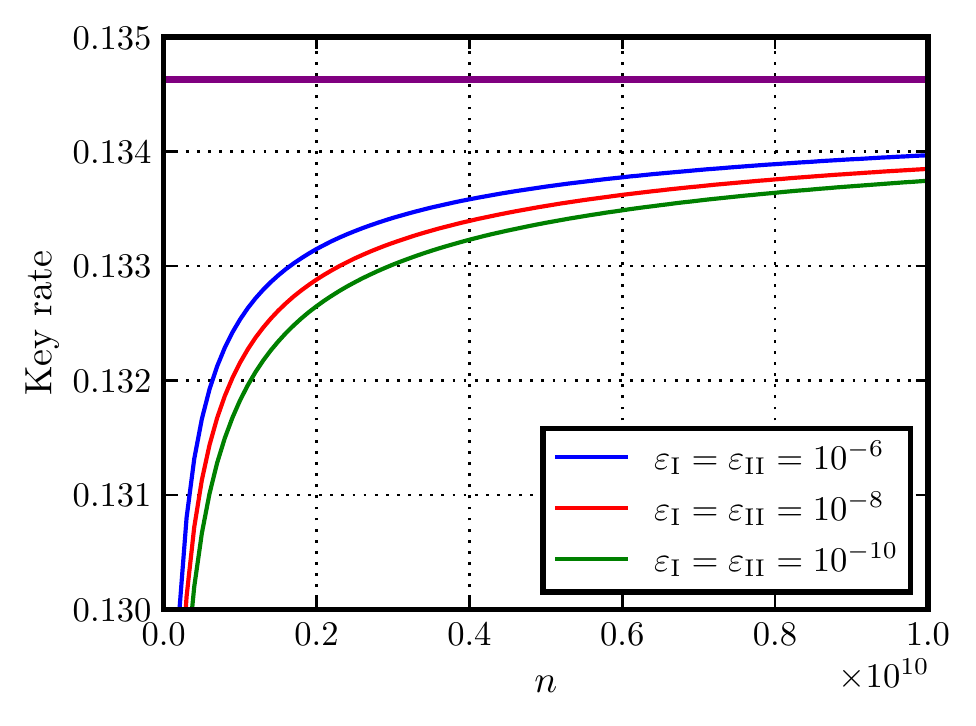}\caption{Distillable key rates
for the reverse reconciliation CV-QKD protocol, when conducted over a thermal
channel with transmissivity $\eta= 0.2$ and environment photon number $N_{B} =
10^{-2}$. The purple solid line indicates the asymptotic key rate.}%
\label{fig:CVQKD-performance}%
\end{figure}

\subsection{Routes to a full second-order analysis that includes parameter
estimation}

\label{sec:future-routes-PE}As we have mentioned in
Section~\ref{sec:parameter-est-QKD}, our second-order analysis above is
incomplete, in the sense that it does not incorporate details of the parameter
estimation step. That is, we have assumed that $k$ is large enough to allow
for a reliable estimate of the induced classical channel $p_{Y|X}$, and from
there, we applied a second-order analysis to the key distillation step. Here
we discuss a potential route around this issue and an avenue for future work
on this topic.

For finite-dimensional protocols, the collective attack of Eve is constrained
to a finite-dimensional quantum channel. This quantum channel is realized by a
unitary acting on a larger, yet finite-dimensional Hilbert space. If we limit
Eve to a finite set of quantum computations that can be realized from a
finite, universal gate set, which could in principle realize any unitary on
the space to any desired accuracy, then the resulting channel is selected from
a very large, yet finite and known set. In this case, the parameter estimation
procedure of \cite[Section~III-D]{Poly13} can be applied, and it has
negligible overhead. Indeed, it is possible to estimate the channel chosen
from a finite set with $k=O(\log n)$ and such that the decoding error
probability is increased by no more than $1/\sqrt{n}$, which has no impact on
the second-order terms in \eqref{eq:direct-reconciliation-second-order} and
\eqref{eq:reverse-reconciliation-second-order} for sufficiently large $n$.
However, the main drawback of this approach is that, even though this
constraint on the eavesdropper might be considered reasonable, it is still a
computational assumption and thus lies outside the quantum
information-theoretic security that is ultimately desired in quantum key
distribution. Furthermore, this does not address CV-QKD\ protocols in which
there is not a finite-dimensional assumption.

Another approach to address the issue is to incorporate recent advances in
universal coding up to the second order \cite{Hay19}. A universal code does
not depend on the particular channel over which communication is being
conducted, and it only requires an estimate of the first- and second-order
terms. This approach has been analyzed in detail in \cite{Hay19}, and it is
likely that it could be combined with the information-theoretic protocol
considered in this paper and the parameter estimation step in order to arrive
at a full second-order analysis for both the parameter estimation and key
distillation steps in a quantum key distribution protocol. A benefit of this
approach is that quantum information-theoretic security would be retained. We
leave a full investigation of this approach for future work.

\section{Private communication over a compound wiretap channel with fixed
marginal}

\label{sec:priv-comm-comp-wiretap}We now present the one-shot
information-theoretic results that underlie the previous claims for key
distillation in quantum key distribution. We begin by considering private
communication over a compound wiretap channel with fixed marginal, and then in
Section~\ref{sec:key-dist-compound-source}, we move on to a key distillation
protocol for a compound wiretap source with fixed marginal.

\subsection{One-shot setting}

\label{sec:one-shot-qsdc}

In the setting of private communication, we suppose that Alice, Bob, and Eve
are connected by means of the following classical--quantum--quantum
(cqq)\ compound wiretap channel with fixed marginal:%
\begin{equation}
\mathcal{N}_{X\rightarrow BE}^{s}:x\rightarrow\rho_{BE}^{x,s},
\label{eq:cqq-ch}%
\end{equation}
where the classical input $x\in\mathcal{X}$, the alphabet $\mathcal{X}$ is
countable, the index $s\in\mathcal{S}$ represents the choice or state of the
channel, the set $\mathcal{S}$ can be uncountable, and $\rho_{BE}^{x,s}$ is a
density operator acting on the tensor-product separable Hilbert space
$\mathcal{H}_{B}\otimes\mathcal{H}_{E}$. We suppose that the channel has a
fixed marginal; i.e, the reduced Alice-Bob channel%
\begin{equation}
\mathcal{N}_{X\rightarrow B}^{s}:x\rightarrow\rho_{B}^{x,s}\coloneqq \operatorname{Tr}%
_{E}[\rho_{BE}^{x,s}]
\end{equation}
is known to Alice and Bob, so that%
\begin{equation}
\rho_{B}^{x,s}=\rho_{B}^{x}\quad\text{for all}\quad s\in\mathcal{S}.
\end{equation}
However, the reduced $s$-dependent Alice-Eve channel%
\begin{equation}
\mathcal{N}_{X\rightarrow E}^{s}:x\rightarrow\rho_{E}^{x,s}\coloneqq \operatorname{Tr}%
_{B}[\rho_{BE}^{x,s}]
\end{equation}
is not fully known to or characterized by Alice or Bob during their
communication. Thus, the channel $\mathcal{N}_{X\rightarrow BE}^{s}$ is a
particular member of a set $\mathcal{S}$\ of channels, for which%
\begin{equation}
\operatorname{Tr}_{E}[\rho_{BE}^{x,s}]=\rho_{B}^{x}\quad\text{for all}\quad
s\in\mathcal{S}. \label{eq:ch-assumption}%
\end{equation}
To be clear, we suppose that the set $\mathcal{S}$ is known to Alice and Bob,
but the particular channel indexed by $s\in\mathcal{S}$, which is actually
being employed during the communication protocol, is not known. The goal of a
private communication protocol is for Alice to send a message securely to Bob,
such that it is secure against the channel $\mathcal{N}_{X\rightarrow BE}^{s}$
for all $s\in\mathcal{S}$.

The action of the channel in \eqref{eq:cqq-ch}\ can be described in the fully
quantum picture as follows:%
\begin{equation}
\mathcal{N}_{X\rightarrow BE}^{s}(\omega_{X})\coloneqq \sum_{x}\langle x|_{X}%
\omega_{X}|x\rangle_{X}\rho_{BE}^{x,s}, \label{eq:fully-q-cqq-channel}%
\end{equation}
where $\{|x\rangle_{X}\}_{x\in\mathcal{X}}$ is a countable orthonormal basis
for an input separable Hilbert space $\mathcal{H}_{X}$.

In our model of communication, we suppose also that public shared randomness
is available for free to Alice and Bob. Since it is public, Eve gets a copy of
all of the shared randomness. In particular, therefore, this shared randomness is not the same as prior secret key shared by Alice and Bob.

Let $\overline{\Psi}_{R_{A}R_{B}R_{E}}$ denote the following public shared
randomness state:%
\begin{equation}
\overline{\Psi}_{R_{A}R_{B}R_{E}}\coloneqq \sum_{r=1}^{|\mathcal{R}|}p(r)|r\rangle\!
\langle r|_{R_{A}}\otimes|r\rangle\!\langle r|_{R_{B}}\otimes|r\rangle\!\langle
r|_{R_{E}},
\end{equation}
where $\mathcal{R}$ is the alphabet for the shared randomness, Alice possesses
the classical register $R_{A}$, Bob $R_{B}$, and Eve $R_{E}$. The systems
$R_{A}$, $R_{B}$, and $R_{E}$ have the same dimension. A private coding scheme
for the channel in \eqref{eq:cqq-ch} consists of the public shared randomness
state $\overline{\Psi}_{R_{A}R_{B}R_{E}}$, an encoding channel $\mathcal{E}%
_{R_{A}M\rightarrow X}$ taking the register $R_{A}$ and a message register $M$
to the channel input register $X$, and a decoding channel $\mathcal{D}%
_{R_{B}B\rightarrow\hat{M}}$. Two parameters $\varepsilon_{\operatorname{I}}$
and $\varepsilon_{\operatorname{II}}$ determine the decoding error probability
and the security of the protocol, respectively. Alice uses the coding scheme
and the channel in \eqref{eq:cqq-ch} to send a message $m\in\mathcal{M}$,
where $\mathcal{M}$ is the message alphabet, and the conditional probability
of decoding as $m^{\prime}\in\mathcal{M}$ is given by%
\begin{equation}
\Pr[\hat{M}=m^{\prime}|M=m]\coloneqq \langle m^{\prime}|_{\hat{M}}\mathcal{D}%
_{R_{B}B\rightarrow\hat{M}}(\rho_{BR_{B}}^{m})|m^{\prime}\rangle_{\hat{M}%
},\label{eq:decoding-channel-cc-induced}%
\end{equation}
where%
\begin{equation}
\rho_{BR_{B}}^{m}\coloneqq \operatorname{Tr}_{ER_{E}}[\rho_{BER_{B}R_{E}}^{m,s}],
\end{equation}%
\begin{multline}
\rho_{BER_{B}R_{E}}^{m,s}\coloneqq \\
(\mathcal{N}_{X\rightarrow BE}^{s}\circ\mathcal{E}_{R_{A}M\rightarrow
X})(|m\rangle\!\langle m|_{M}\otimes\overline{\Psi}_{R_{A}R_{B}R_{E}}).
\end{multline}
Note that the reduced state $\rho_{BR_{B}}^{m}$ has no dependence on $s$, due
to the assumption stated in \eqref{eq:ch-assumption}. To see
\eqref{eq:decoding-channel-cc-induced} in a different way, we can
alternatively consider Bob's positive operator-valued measure $\left\{
\Lambda_{R_{B}B}^{m}\right\}  _{m\in\mathcal{M}}$ to consist of the following
elements:%
\begin{equation}
\Lambda_{R_{B}B}^{m}\coloneqq (\mathcal{D}_{R_{B}B\rightarrow\hat{M}})^{\dag
}(|m\rangle\!\langle m|_{\hat{M}}),
\end{equation}
where $(\mathcal{D}_{R_{B}B\rightarrow\hat{M}})^{\dag}$ is the
Hilbert--Schmidt adjoint of the decoding channel $\mathcal{D}_{R_{B}%
B\rightarrow\hat{M}}$, so that we can write $\Pr[\hat{M}=m^{\prime}|M=m]$ in
terms of the Born rule as%
\begin{equation}
\Pr[\hat{M}=m^{\prime}|M=m]=\operatorname{Tr}[\Lambda_{R_{B}B}^{m^{\prime}%
}\rho_{BR_{B}}^{m}].
\end{equation}
The protocol is $\varepsilon_{\operatorname{I}}$-reliable, with $\varepsilon
_{\operatorname{I}}\in\left[  0,1\right]  $, if the following condition holds
for all $m\in\mathcal{M}$:%
\begin{equation}
\Pr[\hat{M}=m|M=m]\geq1-\varepsilon_{\operatorname{I}}.
\end{equation}
Equivalently, we require that%
\begin{equation}
\sup_{m\in\mathcal{M}}\left(  1-\Pr[\hat{M}=m|M=m]\right)  \leq\varepsilon
_{\operatorname{I}}.
\end{equation}

The protocol is $\varepsilon_{\operatorname{II}}$-secret, with $\varepsilon
_{\operatorname{II}}\in\left[  0,1\right]  $, if for all channel states
$s\in\mathcal{S}$, there exists a fixed state $\sigma_{ER_{E}}^{s}$ of Eve's
systems, such that for all messages $m\in\mathcal{M}$, the following
inequality holds%
\begin{equation}
\frac{1}{2}\left\Vert \rho_{ER_{E}}^{m,s}-\sigma_{ER_{E}}^{s}\right\Vert
_{1}\leq\varepsilon_{\operatorname{II}}.
\end{equation}
Equivalently, we require that%
\begin{equation}
\sup_{s\in\mathcal{S}}\inf_{\sigma_{ER_{E}}^{s}}\sup_{m\in\mathcal{M}}\frac
{1}{2}\left\Vert \rho_{ER_{E}}^{m,s}-\sigma_{ER_{E}}^{s}\right\Vert _{1}%
\leq\varepsilon_{\operatorname{II}}.
\end{equation}
Thus, for each $s\in\mathcal{S}$, the protocol allows for sending a message
securely, such that the eavesdropper cannot determine the message $m$, no
matter which channel $\mathcal{N}_{X\rightarrow BE}^{s}$ is selected from
$\mathcal{S}$.

The number of private bits communicated by the scheme is equal to $\log
_{2}\left\vert \mathcal{M}\right\vert $. Thus, a given protocol for private
communication over the channel in \eqref{eq:cqq-ch} is described by the three
parameters $\left\vert \mathcal{M}\right\vert $, $\varepsilon
_{\operatorname{I}}$, and $\varepsilon_{\operatorname{II}}$.

We remark that this setting reduces to the traditional one-shot,
shared-randomness-assisted private communication setting considered in quantum
information theoretic contexts in the case that $|\mathcal{S}|=1$. See
\cite{wilde2017position} for details of this special case.

\subsection{Distinguishability and information measures}

Before stating Theorem~\ref{thm:one-shot-priv-comm}\ regarding the number of
private messages that can be transmitted in the above setting, we briefly
review the basic distinguishability measures and information quantities needed
to understand the statement of Theorem~\ref{thm:one-shot-priv-comm}.

Let $\rho$ and $\sigma$ be states acting on a separable Hilbert space
$\mathcal{H}$. The hypothesis testing relative entropy $D_{H}^{\varepsilon
}(\rho\Vert\sigma)$\ is defined for $\varepsilon\in\lbrack0,1]$ as
\cite{BD10,BD11,WR12}%
\begin{equation}
D_{H}^{\varepsilon}(\rho\Vert\sigma)\coloneqq -\log_{2}\inf_{\Lambda\geq0}\left\{
\operatorname{Tr}[\Lambda\sigma]:\operatorname{Tr}[\Lambda\rho]\geq
1-\varepsilon,\Lambda\leq I\right\}  .
\end{equation}
Note that, without loss of generality, we can set the first constraint above
to be an equality \cite{KW17a}, so that%
\begin{equation}
D_{H}^{\varepsilon}(\rho\Vert\sigma)=-\log_{2}\inf_{\Lambda\geq0}\left\{
\operatorname{Tr}[\Lambda\sigma]:\operatorname{Tr}[\Lambda\rho]=1-\varepsilon
,\ \Lambda\leq I\right\}  .
\end{equation}
The max-relative entropy $D_{\max}(\rho\Vert\sigma)$ is defined as
\cite{Datta2009b}%
\begin{equation}
D_{\max}(\rho\Vert\sigma)\coloneqq \inf\left\{  \lambda:\rho\leq2^{\lambda}%
\sigma\right\}  .
\end{equation}
The following alternative characterization of $D_{\max}(\rho\Vert\sigma)$ is
well known \cite{Dat09}:%
\begin{multline}
D_{\max}(\rho\Vert\sigma)=\label{eq:alt-dmax}\\
\inf_{\lambda,\omega\geq0}\left\{  \lambda:\sigma=2^{-\lambda}\rho+\left(
1-2^{-\lambda}\right)  \omega,\ \operatorname{Tr}[\omega]=1\right\}  ,
\end{multline}
which allows for thinking of $\sigma$ as a convex combination of $\rho$ and
some other state $\omega$. The smooth max-relative entropy for $\varepsilon
\in\left[  0,1\right]  $ is defined as \cite{Datta2009b}%
\begin{multline}
D_{\max}^{\varepsilon}(\rho\Vert\sigma)\coloneqq \\
\inf\left\{  \lambda:\widetilde{\rho}\leq2^{\lambda}\sigma,\ \widetilde{\rho
}\geq0,\ \operatorname{Tr}[\widetilde{\rho}]=1,\ P(\widetilde{\rho},\rho
)\leq\varepsilon\right\}  ,
\end{multline}
where the sine distance $P(\widetilde{\rho},\rho)$ is defined as
\cite{R02,R03,R06,GLN04}%
\begin{align}
P(\widetilde{\rho},\rho)  &  \coloneqq \sqrt{1-F(\widetilde{\rho},\rho)},\\
F(\widetilde{\rho},\rho)  &  \coloneqq \left\Vert \sqrt{\widetilde{\rho}}\sqrt{\rho
}\right\Vert _{1}^{2},
\end{align}
the latter quantity being the quantum fidelity \cite{Uhl76}.

Let $\rho_{AE}$ be a state acting on a separable Hilbert space $\mathcal{H}%
_{A}\otimes\mathcal{H}_{E}$. Then a mutual-information like quantity is
defined for $\varepsilon\in\left[  0,1\right]  $ as%
\begin{equation}
D_{\max}^{\varepsilon}(\rho_{AE}\Vert\rho_{A}\otimes\rho_{E}%
).\label{eq:smooth-dmax-MI}%
\end{equation}
We also define the alternate smooth max-mutual information as
\cite{AJW17,wilde2017position}%
\begin{equation}
\widetilde{I}_{\max}^{\varepsilon}(E;A)_{\rho}\coloneqq \inf_{\widetilde{\rho}%
_{AE}:P(\widetilde{\rho}_{AE},\rho_{AE})\leq\varepsilon}D_{\max}%
(\widetilde{\rho}_{AE}\Vert\rho_{A}\otimes\widetilde{\rho}_{E}%
).\label{eq:smooth-alt-imax}%
\end{equation}
A simple relation between these two generalizations of mutual information is
given in Appendix~\ref{app:relation-smooth-max-MIs}. Note that we adopt the
particular notation in \eqref{eq:smooth-alt-imax}\ in order to maintain
consistency with the notation of \cite{BCR09}.

The hypothesis testing mutual information of a bipartite state $\rho_{AB}$
acting on a separable Hilbert space $\mathcal{H}_{A}\otimes\mathcal{H}_{B}$ is
defined for $\varepsilon\in\left[  0,1\right]  $ as \cite{WR12}%
\begin{equation}
I_{H}^{\varepsilon}(A;B)_{\rho}\coloneqq D_{H}^{\varepsilon}(\rho_{AB}\Vert\rho
_{A}\otimes\rho_{B}).
\end{equation}

\subsection{Private communication via\ position-based coding and convex
splitting}

The communication protocol that we develop here is the same as that considered
in \cite{wilde2017position}, but with the added observation that its security
holds universally for any channel selected from the set $\mathcal{S}$ and for
channels whose output states act on a separable Hilbert space. The protocol
from \cite{wilde2017position} built upon the ideas of position-based coding
and convex splitting, which were in turn introduced in \cite{AJW17} and
\cite{ADJ17}, respectively. Our main one-shot achievability theorem, for the
communication setting described in Section~\ref{sec:one-shot-qsdc}, is as follows:

\begin{theorem}
\label{thm:one-shot-priv-comm}Fix $\varepsilon_{\operatorname{I}}%
,\varepsilon_{\operatorname{II}}\in(0,1)$, $\eta_{\operatorname{I}}%
\in(0,\varepsilon_{\operatorname{I}})$, and $\eta_{\operatorname{II}}%
\in(0,{\varepsilon_{\operatorname{II}}})$. Then the following quantity is
an achievable number of private message bits that can be transmitted over the
compound wiretap channel in \eqref{eq:cqq-ch}, with decoding error probability
not larger than $\varepsilon_{\operatorname{I}}$ and security parameter not
larger than $\varepsilon_{\operatorname{II}}$:%
\begin{multline}
\sup_{p_{X}}\left[  I_{H}^{\varepsilon_{\operatorname{I}}-\eta
_{\operatorname{I}}}(X;B)_{\rho}-\sup_{s\in\mathcal{S}}\widetilde{I}_{\max
}^{{\varepsilon_{\operatorname{II}}}-\eta_{\operatorname{II}}}%
(E;X)_{\rho^{s}}\right] \\
-\log_{2}(4\varepsilon_{\operatorname{I}}/\eta_{\operatorname{I}}^{2}%
)-2\log_{2}(1/2\eta_{\operatorname{II}}),
\end{multline}
where the entropic quantities are evaluated with respect to the following
state:%
\begin{equation}
\rho_{XBE}^{s}\coloneqq \sum_{x\in\mathcal{X}}p_{X}(x)|x\rangle\!\langle x|_{X}%
\otimes\rho_{BE}^{x,s}.
\end{equation}

\end{theorem}

\begin{proof}
Consider the compound wiretap channel in \eqref{eq:cqq-ch}. Let $p_{X}$ be a
probability distribution over the channel input alphabet $\mathcal{X}$. Let
$\rho_{XX^{\prime}X^{\prime\prime}}$ denote the following public shared
randomness state:%
\begin{equation}
\rho_{XX^{\prime}X^{\prime\prime}}\coloneqq \sum_{x\in\mathcal{X}}p_{X}(x)|x\rangle\!
\langle x|_{X}\otimes|x\rangle\!\langle x|_{X^{\prime}}\otimes|x\rangle\!\langle
x|_{X^{\prime\prime}}.
\end{equation}
We suppose that Alice, Bob, and Eve share the state
\begin{equation}
\rho_{XX^{\prime}X^{\prime\prime}}^{\otimes\left\vert \mathcal{M}\right\vert
\left\vert \mathcal{R}\right\vert }%
\end{equation}
before communication begins and that the $X$ systems are indexed as $X_{m,r}$,
i.e., in lexicographic order by $m\in\mathcal{M}$ and $r\in\mathcal{R}$, where
$\mathcal{R}\coloneqq \{1,\ldots,\left\vert \mathcal{R}\right\vert \}$. Alice
possesses all of the $X$ systems, Bob the $X^{\prime}$ systems, and Eve the
$X^{\prime\prime}$ systems.

To send the message $m\in\mathcal{M}$, Alice's encoding consists of picking
$r\in\mathcal{R}$ uniformly at random from the set $\mathcal{R}$, and then she
sends the classical system $X_{m,r}$ through the channel in
\eqref{eq:fully-q-cqq-channel}. For fixed values of $m$, $r$, and $s$, the
global shared state at this point is given by%
\begin{multline}
\rho_{X^{\left\vert \mathcal{M}\right\vert \left\vert \mathcal{R}\right\vert
}X^{\prime\left\vert \mathcal{M}\right\vert \left\vert \mathcal{R}\right\vert
}X^{\prime\prime\left\vert \mathcal{M}\right\vert \left\vert \mathcal{R}%
\right\vert }BE}^{m,r,s}\coloneqq \rho_{X_{1,1}X_{1,1}^{\prime}X_{1,1}^{\prime\prime}%
}\otimes\cdots\\
\otimes\rho_{X_{m,r-1}X_{m,r-1}^{\prime}X_{m,r-1}^{\prime\prime}}\otimes
\rho_{X_{m,r}X_{m,r}^{\prime}X_{m,r}^{\prime\prime}BE}^{s}\otimes\\
\rho_{X_{m,r+1}X_{m,r+1}^{\prime}X_{m,r+1}^{\prime\prime}}\otimes\cdots
\otimes\rho_{X_{\left\vert \mathcal{M}\right\vert ,\left\vert \mathcal{R}%
\right\vert }X_{\left\vert \mathcal{M}\right\vert ,\left\vert \mathcal{R}%
\right\vert }^{\prime}X_{\left\vert \mathcal{M}\right\vert ,\left\vert
\mathcal{R}\right\vert }^{\prime\prime}},
\end{multline}
where
\begin{multline}
\rho_{X_{m,r}X_{m,r}^{\prime}X_{m,r}^{\prime\prime}BE}^{s} \coloneqq \\
\sum_{x\in\mathcal{X}}p_{X}(x)|x,x,x\rangle\!\langle x,x,x|_{X_{m,r}%
X_{m,r}^{\prime}X_{m,r}^{\prime\prime}}\otimes\rho_{BE}^{x,s}.
\end{multline}
The resulting state of Bob, for fixed values $m$ and $r$, is as follows:%
\begin{multline}
\rho_{X^{\prime\left\vert \mathcal{M}\right\vert \left\vert \mathcal{R}%
\right\vert }B}^{m,r}\coloneqq \rho_{X_{1,1}^{\prime}}\otimes\cdots\otimes
\rho_{X_{m,r-1}^{\prime}}\otimes\rho_{X_{m,r}^{\prime}B}%
\label{eq:Bob-state-qsdc}\\
\otimes\rho_{X_{m,r+1}^{\prime}}\otimes\cdots\otimes\rho_{X_{\left\vert
\mathcal{M}\right\vert ,\left\vert \mathcal{R}\right\vert }^{\prime}},
\end{multline}
where%
\begin{equation}
\rho_{X_{m,r}^{\prime}B}\coloneqq \sum_{x\in\mathcal{X}}p_{X}(x)|x\rangle\!\langle
x|_{X_{m,r}^{\prime}}\otimes\rho_{B}^{x}.
\end{equation}
By employing the sequential and position-based decoding scheme from
\cite{OMW19} (which has been shown therein to hold for states acting on
separable Hilbert spaces), it follows that Bob can decode both the message $m$
and $r$ with probability not smaller than $1-\varepsilon_{\operatorname{I}}$
as long as%
\begin{equation}
\log_{2}(\left\vert \mathcal{M}\right\vert \left\vert \mathcal{R}\right\vert
)=I_{H}^{\varepsilon_{\operatorname{I}}-\eta_{\operatorname{I}}}(X;B)_{\rho
}-\log_{2}(4\varepsilon_{\operatorname{I}}/\eta_{\operatorname{I}}^{2}),
\end{equation}
where $\eta_{\operatorname{I}}\in(0,\varepsilon_{\operatorname{I}})$ and%
\begin{equation}
\rho_{XB}\coloneqq \sum_{x\in\mathcal{X}}p_{X}(x)|x\rangle\!\langle x|_{X}\otimes
\rho_{B}^{x}.
\end{equation}

To Eve, who is unaware of the random choice of the local randomness variable
$r\in\mathcal{R}$, the state of her systems for a fixed value of the message
$m$ is as follows:%
\begin{multline}
\rho_{X^{\prime\prime\left\vert \mathcal{M}\right\vert \left\vert
\mathcal{R}\right\vert }E}^{m,s}\coloneqq \frac{1}{\left\vert \mathcal{R}\right\vert
}\sum_{r\in\mathcal{R}}\rho_{X_{1,1}^{\prime\prime}}\otimes\cdots\otimes
\rho_{X_{m,r-1}^{\prime\prime}}\otimes\rho_{X_{m,r}^{\prime\prime}E}%
^{s}\label{eq:Eve-state-qsdc}\\
\otimes\rho_{X_{m,r+1}^{\prime\prime}}\otimes\cdots\otimes\rho_{X_{\left\vert
\mathcal{M}\right\vert ,\left\vert \mathcal{R}\right\vert }^{\prime\prime}}%
\end{multline}
By the invariance of the trace distance with respect to tensor-product states,
i.e., $\left\Vert \sigma\otimes\tau-\omega\otimes\tau\right\Vert
_{1}=\left\Vert \sigma-\omega\right\Vert _{1}$, we find that%
\begin{multline}
\frac{1}{2}\left\Vert \rho_{X^{\prime\prime\left\vert \mathcal{M}\right\vert
\left\vert \mathcal{R}\right\vert }E}^{m,s}-\rho_{X^{\prime\prime\left\vert
\mathcal{M}\right\vert \left\vert \mathcal{R}\right\vert }}\otimes
\widetilde{\rho}_{E}^{s}\right\Vert _{1}\\
=\frac{1}{2}\left\Vert \rho_{X_{m,1}^{\prime\prime}\ldots X_{m,\left\vert
\mathcal{R}\right\vert }^{\prime\prime}E}^{m,s}-\rho_{X_{m,1}^{\prime\prime
}\ldots X_{m,\left\vert \mathcal{R}\right\vert }^{\prime\prime}}%
\otimes\widetilde{\rho}_{E}^{s}\right\Vert _{1},
\end{multline}
for any state $\widetilde{\rho}_{E}^{s}$. It follows from the smooth universal
convex-split lemma (see Appendix~\ref{sec:smooth-univ-convex-split}) that if%
\begin{equation}
\log_{2}\left\vert \mathcal{R}\right\vert =\sup_{s\in\mathcal{S}}\widetilde
{I}_{\max}^{{\varepsilon_{\operatorname{II}}}-\eta_{\operatorname{II}}%
}(E;X)_{\rho^{s}}+2\log_{2}(1/2\eta_{\operatorname{II}}),
\end{equation}
then the following bound holds for all $m\in\mathcal{M}$ and for all
$s\in\mathcal{S}$%
\begin{equation}
\frac{1}{2}\left\Vert \rho_{X^{\prime\prime\left\vert \mathcal{M}\right\vert
\left\vert \mathcal{R}\right\vert }E}^{m,s}-\rho_{X^{\prime\prime\left\vert
\mathcal{M}\right\vert \left\vert \mathcal{R}\right\vert }}\otimes
\widetilde{\rho}_{E}^{s}\right\Vert _{1}\leq\varepsilon_{\operatorname{II}},
\end{equation}
where the state $\widetilde{\rho}_{E}^{s}$ satisfies $P(\widetilde{\rho}%
_{E}^{s},\rho_{E}^{s})\leq{\varepsilon_{\operatorname{II}}}%
-\eta_{\operatorname{II}}$ and the alternate smooth max-mutual information
$\widetilde{I}_{\max}^{\sqrt{\varepsilon_{\operatorname{II}}}-\eta
_{\operatorname{II}}}(E;X)_{\rho^{s}}$ is evaluated with respect to the state%
\begin{equation}
\rho_{XE}^{s}\coloneqq \sum_{x\in\mathcal{X}}p_{X}(x)|x\rangle\!\langle x|_{X}%
\otimes\rho_{E}^{x,s}.
\end{equation}
The analysis of this step is similar to that given in \cite{wilde2017position}%
. We provide details of the smooth universal convex-split lemma, as applicable
to states acting on separable Hilbert spaces, in
Appendix~\ref{sec:smooth-univ-convex-split}. Thus, the number of private
message bits that can be established with this scheme is equal to%
\begin{multline}
\log_{2}\left\vert \mathcal{M}\right\vert =I_{H}^{\varepsilon
_{\operatorname{I}}-\eta_{\operatorname{I}}}(X;B)_{\rho}-\sup_{s\in
\mathcal{S}}\widetilde{I}_{\max}^{{\varepsilon_{\operatorname{II}}}%
-\eta_{\operatorname{II}}}(E;X)_{\rho^{s}}\\
-\log_{2}(4\varepsilon_{\operatorname{I}}/\eta_{\operatorname{I}}^{2}%
)-2\log_{2}(1/2\eta_{\operatorname{II}}).
\end{multline}
Since the above number of private message bits is achievable for a fixed
distribution $p_{X}$, it is then possible to optimize over all distributions
to conclude that it is possible to send the following number of private
message bits:%
\begin{multline}
\log_{2}\left\vert \mathcal{M}\right\vert
=\label{eq:one-shot-ach-private-comm}\\
\sup_{p_{X}}\left[  I_{H}^{\varepsilon_{\operatorname{I}}-\eta
_{\operatorname{I}}}(X;B)_{\rho}-\sup_{s\in\mathcal{S}}\widetilde{I}_{\max
}^{{\varepsilon_{\operatorname{II}}}-\eta_{\operatorname{II}}}%
(E;X)_{\rho^{s}}\right] \\
-\log_{2}(4\varepsilon_{\operatorname{I}}/\eta_{\operatorname{I}}^{2}%
)-2\log_{2}(1/2\eta_{\operatorname{II}}),
\end{multline}
concluding the proof.
\end{proof}

\begin{remark}
If desired, this private communication scheme can be derandomized along the
lines shown in \cite{wilde2017position}, in order to end up with a scheme that
does not require public shared randomness. After doing so, the protocol is no
longer guaranteed to be secure against an arbitrary state of the eavesdropper
in the uncertainty set $\mathcal{S}$. The resulting protocol is only
guaranteed to be secure against a fixed, known eavesdropper, similar to the
standard information-theoretic model of private communication from
\cite{D05,CWY04}.
\end{remark}

\section{Secret key distillation from a compound quantum wiretap source with
fixed marginal}

\label{sec:key-dist-compound-source}

\subsection{One-shot setting}

\label{sec:key-dist-compound-source-one-shot}

The setting of secret key distillation from a compound quantum wiretap source
with fixed marginal can be considered the static version of the dynamic
setting presented in Section~\ref{sec:one-shot-qsdc}, and so many aspects are
similar. In this setting, we suppose that Alice has system $X$, Bob system
$B$, and Eve system $E$ of the following classical--quantum--quantum compound
wiretap source state with fixed marginal:%
\begin{equation}
\rho_{XBE}^{s}\coloneqq \sum_{x}p_{X}(x)|x\rangle\!\langle x|_{X}\otimes\rho_{BE}%
^{x,s},\label{eq:cqq-state}%
\end{equation}
where $s\in\mathcal{S}$. In this model, the marginal state of Alice and Bob is
fixed, such that they know their reduced state $\rho_{XB}$. That is,%
\begin{equation}
\rho_{B}^{x,s}=\rho_{B}^{x}\quad\text{for all}\quad s\in\mathcal{S},
\end{equation}
where%
\begin{equation}
\rho_{B}^{x,s}\coloneqq \operatorname{Tr}_{E}[\rho_{BE}^{x,s}],
\end{equation}
but they do not know the full state. We suppose that forward public classical
communication from Alice to Bob is allowed for free. The goal is to use this
compound wiretap source state in \eqref{eq:cqq-state}, along with the free
public classical communication, in order to distill a key that is secure for
all $\rho_{XBE}^{s}$ and $s\in\mathcal{S}$.

A secret key distillation scheme consists of a classical encoding channel
$\mathcal{E}_{X\rightarrow KL}$ and a decoding channel $\mathcal{D}%
_{LB\rightarrow\hat{K}}$. In the above, $K$ is the classical key register of
Alice, and $L$ is a classical register communicated to Bob. Two parameters
$\varepsilon_{\operatorname{I}}$ and $\varepsilon_{\operatorname{II}}$
determine the error probability when decoding and the security of the
protocol, respectively. Alice uses the distillation scheme and the state in
\eqref{eq:cqq-state} to produce a uniformly random key value $k\in\mathcal{K}$
(the probability $\Pr[K=k]=1/\left\vert \mathcal{K}\right\vert $), and the
conditional probability of Bob decoding $k^{\prime}\in\mathcal{K}$ from his
systems is given by%
\begin{equation}
\Pr[\hat{K}=k^{\prime}|K=k]\coloneqq \langle k^{\prime}|_{\hat{K}}\mathcal{D}%
_{LB\rightarrow\hat{K}}(\rho_{KLB}^{k})|k^{\prime}\rangle_{\hat{K}},
\end{equation}
where%
\begin{align}
\rho_{KLB}^{k}  &  \coloneqq \operatorname{Tr}_{E}[\rho_{KLBE}^{k,s}],\\
\rho_{KLBE}^{k,s}  &  \coloneqq \frac{\langle k|_{K}\mathcal{E}_{X\rightarrow KL}%
(\rho_{XBE})|k\rangle_{K}}{\Pr[K=k]},\\
&  =\langle k|_{K}\mathcal{E}_{X\rightarrow KL}(\rho_{XBE})|k\rangle_{K}%
\cdot\left\vert \mathcal{K}\right\vert .
\end{align}
Alternatively, we can consider Bob's positive operator-valued measure
$\left\{  \Lambda_{LB}^{k}\right\}  _{k\in\mathcal{K}}$ to consist of the
following elements%
\begin{equation}
\Lambda_{LB}^{k}\coloneqq (\mathcal{D}_{LB\rightarrow\hat{K}})^{\dag}(|k\rangle\!\langle
k|_{\hat{K}}),
\end{equation}
where $(\mathcal{D}_{LB\rightarrow\hat{K}})^{\dag}$ is the Hilbert--Schmidt
adjoint of the decoding channel $\mathcal{D}_{LB\rightarrow\hat{K}}$, so that
we can write the probability $\Pr[ \hat{K}=k^{\prime}|K=k] $ in terms of the
Born rule as%
\begin{equation}
\Pr[\hat{K}=k^{\prime}|K=k]=\operatorname{Tr}[\Lambda_{LB}^{k^{\prime}}%
\rho_{KLB}^{k}].
\end{equation}
The protocol is $\varepsilon_{\operatorname{I}}$-reliable, with $\varepsilon
_{\operatorname{I}}\in\left[  0,1\right]  $, if the following condition holds
for all $k\in\mathcal{K}$:%
\begin{equation}
\Pr[\hat{K}=k|K=k]\geq1-\varepsilon_{\operatorname{I}}.
\label{eq:reliable-distill}%
\end{equation}
Equivalently, we require that
\begin{equation}
\sup_{k\in\mathcal{K}}\left(  1-\Pr[\hat{K}=k|K=k]\right)  \leq\varepsilon
_{\operatorname{I}}.
\end{equation}
The protocol is $\varepsilon_{\operatorname{II}}$-secret, with $\varepsilon
_{\operatorname{II}}\in\left[  0,1\right]  $, if for all source values
$s\in\mathcal{S}$, there exists fixed state $\sigma_{ER_{E}}^{s}$ of Eve's
systems, such that for all key values $k\in\mathcal{K}$, the following
inequality holds%
\begin{equation}
\frac{1}{2}\left\Vert \rho_{LE}^{k,s}-\sigma_{LE}^{s}\right\Vert _{1}%
\leq\varepsilon_{\operatorname{II}}.
\end{equation}
Equivalently, we require that%
\begin{equation}
\sup_{s\in\mathcal{S}}\inf_{\sigma_{LE}^{s}}\sup_{k\in\mathcal{K}}\frac{1}%
{2}\left\Vert \rho_{LE}^{k,s}-\sigma_{LE}^{s}\right\Vert _{1}\leq
\varepsilon_{\operatorname{II}}. \label{eq:security-distill}%
\end{equation}
Thus, for each $s\in\mathcal{S}$, the protocol allows for distilling a secure
key no matter which state $\rho_{XBE}^{s}$ is selected from~$\mathcal{S}$.

The number of key bits established by the scheme is equal to $\log
_{2}\left\vert \mathcal{K}\right\vert $. Thus, a given protocol for secret key
distillation using the source in \eqref{eq:cqq-state} is described by the
three parameters $\left\vert \mathcal{K}\right\vert $, $\varepsilon
_{\operatorname{I}}$, and $\varepsilon_{\operatorname{II}}$.

We remark here that the definition given above implies the usual
trace-distance based criterion \cite{KRBM07,TLGR12} for security that is commonly
employed in quantum key distribution. Let $K$ and $\hat{K}$ denote the
respective classical key registers of Alice and Bob (we identify both the
random variable and the system label with the same symbol). Then the final
state of the protocol, for fixed $s\in\mathcal{S}$, can be written as follows:%
\begin{equation}
\rho_{K\hat{K}LE}^{s}\coloneqq \sum_{k}\frac{1}{\left\vert \mathcal{K}\right\vert
}|k\rangle\!\langle k|_{K}\otimes\sum_{k^{\prime}}p(k^{\prime}|k)|k^{\prime
}\rangle\!\langle k^{\prime}|_{\hat{K}}\otimes\rho_{LE}^{k,s},
\end{equation}
where $p(k^{\prime}|k)\coloneqq \operatorname{Tr}[\Lambda_{LB}^{k^{\prime}}\rho
_{KLB}^{k}]$, while the ideal state, for the same $s$, is as follows:%
\begin{equation}
\overline{\Phi}_{K\hat{K}}\otimes\sigma_{LE}^{s},
\end{equation}
for some fixed state $\sigma_{LE}^{s}$ and where $\overline{\Phi}_{K\hat{K}%
}\coloneqq \sum_{k}\frac{1}{\left\vert \mathcal{K}\right\vert }|k\rangle\!\langle
k|_{K}\otimes|k\rangle\!\langle k|_{\hat{K}}$. The conditions in
\eqref{eq:reliable-distill} and \eqref{eq:security-distill} and the triangle
inequality for trace distance imply that%
\begin{equation}
\frac{1}{2}\left\Vert \rho_{K\hat{K}LE}^{s}-\overline{\Phi}_{K\hat{K}}%
\otimes\sigma_{LE}^{s}\right\Vert _{1}\leq\varepsilon_{\operatorname{I}%
}+\varepsilon_{\operatorname{II}}.
\end{equation}
Equivalently, we require
\begin{equation}\label{eq:security_trace_dist}
	\sup_{s\in\mathcal{S}}\inf_{\sigma_{LE}^s}\frac{1}{2}\left\Vert\rho_{K\hat{K}LE}^s-\overline{\Phi}_{K\hat{K}}\otimes\sigma_{LE}^s\right\Vert_1\leq\varepsilon_{\operatorname{I}}+\varepsilon_{\operatorname{II}},
\end{equation}
which is the standard trace-distance based criterion considered in the context
of secret key distillation \cite{TLGR12}. As pointed out in \cite{TLGR12,PR14}, however, this criterion is not composable. The criterion that is known to be composable is
\begin{equation}\label{eq:security_trace_dist_composable}
	\sup_{s\in\mathcal{S}}\frac{1}{2}\left\Vert\rho_{K\hat{K}LE}^s-\overline{\Phi}_{K\hat{K}}\otimes\rho_{LE}^s\right\Vert_1\leq\varepsilon,
\end{equation}
where $\rho_{LE}^s\coloneqq\Tr_{K\hat{K}}[\rho_{K\hat{K}LE}^s]$. As shown in \cite[Appendix~B]{PR14}, if the criterion in \eqref{eq:security_trace_dist} is satisfied, then \eqref{eq:security_trace_dist_composable} is satisfied with $\varepsilon=2(\varepsilon_{\operatorname{I}}+\varepsilon_{\operatorname{II}})$.

We end this section by remarking that this setting reduces to the traditional
one-shot, secret key distillation setting from a known source, considered in
quantum information-theoretic contexts \cite{RR12}, in the case
that~$|\mathcal{S}|=1$.

\subsection{Secret key distillation via position-based coding and convex
splitting}

\label{sec:secret-key-dist}

Our main one-shot achievability theorem, corresponding to the key distillation
setting discussed in the previous section, is as follows:

\begin{theorem}
\label{thm:one-shot-key-dist}Fix $\varepsilon_{\operatorname{I}}%
,\varepsilon_{\operatorname{II}}\in(0,1)$, $\eta_{\operatorname{I}}%
\in(0,\varepsilon_{\operatorname{I}})$, and $\eta_{\operatorname{II}}%
\in(0,{\varepsilon_{\operatorname{II}}})$. Then the following quantity is
an achievable number of secret key bits that can be distilled from the
compound wiretap source in \eqref{eq:cqq-state}, with decoding error
probability not larger than $\varepsilon_{\operatorname{I}}$ and security
parameter not larger than $\varepsilon_{\operatorname{II}}$:%
\begin{multline}
I_{H}^{\varepsilon_{\operatorname{I}}-\eta_{\operatorname{I}}}(X;B)_{\rho
}-\sup_{s\in\mathcal{S}}\widetilde{I}_{\max}^{{\varepsilon
_{\operatorname{II}}}-\eta_{\operatorname{II}}}(E;X)_{\rho^{s}}\\
-\log_{2}(4\varepsilon_{\operatorname{I}}/\eta_{\operatorname{I}}^{2}%
)-2\log_{2}(1/2\eta_{\operatorname{II}}),
\end{multline}
where the entropic quantities are evaluated with respect to the state in \eqref{eq:cqq-state}.
\end{theorem}

\begin{proof}
The secret key distillation scheme consists of the following steps. Alice,
Bob, and Eve begin with the classical--quantum--quantum state in
\eqref{eq:cqq-state}. Alice picks a value $k\in\mathcal{K}$ uniformly at
random, and she picks a value $r\in\mathcal{R}$ uniformly at random, where
$\mathcal{R}\coloneqq \{1,\ldots,\left\vert \mathcal{R}\right\vert \}$. She then
labels her $X$ system of the state in \eqref{eq:cqq-state} by the pair
$(k,r)$, as $X_{k,r}$. She prepares $\left\vert \mathcal{K}\right\vert
\left\vert \mathcal{R}\right\vert -1$ independent instances of the classical
state%
\begin{equation}
\rho_{X}=\sum_{x\in\mathcal{X}}p_{X}(x)|x\rangle\!\langle x|
\end{equation}
and labels the systems as $X_{1,1}$, \ldots, $X_{k,r-1}$, $X_{k,r+1}$, \ldots,
$X_{\left\vert \mathcal{K}\right\vert \left\vert \mathcal{R}\right\vert }$.
Alice sends the classical registers $X_{1,1}$, \ldots, $X_{\left\vert
\mathcal{K}\right\vert \left\vert \mathcal{R}\right\vert }$ in lexicographic
order over a public classical communication channel, so that Bob and Eve
receive copies of them. For fixed values of $k$, $r$, and $s$, the global
shared state at this point is given by%
\begin{multline}
\rho_{X^{\left\vert \mathcal{K}\right\vert \left\vert \mathcal{R}\right\vert
}X^{\prime\left\vert \mathcal{K}\right\vert \left\vert \mathcal{R}\right\vert
}X^{\prime\prime\left\vert \mathcal{K}\right\vert \left\vert \mathcal{R}%
\right\vert }BE}^{k,r,s}\coloneqq \rho_{X_{1,1}X_{1,1}^{\prime}X_{1,1}^{\prime\prime}%
}\otimes\cdots\\
\otimes\rho_{X_{k,r-1}X_{k,r-1}^{\prime}X_{k,r-1}^{\prime\prime}}\otimes
\rho_{X_{k,r}X_{k,r}^{\prime}X_{k,r}^{\prime\prime}BE}^{s}\\
\otimes\rho_{X_{k,r+1}X_{k,r+1}^{\prime}X_{k,r+1}^{\prime\prime}}\otimes
\cdots\otimes\rho_{X_{\left\vert \mathcal{K}\right\vert ,\left\vert
\mathcal{R}\right\vert }X_{\left\vert \mathcal{K}\right\vert ,\left\vert
\mathcal{R}\right\vert }^{\prime}X_{\left\vert \mathcal{K}\right\vert
,\left\vert \mathcal{R}\right\vert }^{\prime\prime}}.
\end{multline}
Thus, the reduced state of Bob, for fixed $k$ and $r$, is as follows:%
\begin{multline}
\rho_{X^{\prime\left\vert \mathcal{K}\right\vert \left\vert \mathcal{R}%
\right\vert }B}^{k,r}\coloneqq \rho_{X_{1,1}^{\prime}}\otimes\cdots\otimes
\rho_{X_{k,r-1}^{\prime}}\otimes\rho_{X_{k,r}^{\prime}B}\\
\otimes\rho_{X_{k,r+1}^{\prime}}\otimes\cdots\otimes\rho_{X_{\left\vert
\mathcal{K}\right\vert ,\left\vert \mathcal{R}\right\vert }^{\prime}}.
\end{multline}
The reduced state of Eve, for a fixed value of $k$ and $s$, is as follows:%
\begin{multline}
\rho_{X^{\prime\prime\left\vert \mathcal{K}\right\vert \left\vert
\mathcal{R}\right\vert }E}^{k,s}\coloneqq \frac{1}{\left\vert \mathcal{R}\right\vert
}\sum_{r\in\mathcal{R}}\rho_{X_{1,1}^{\prime}}\otimes\cdots\otimes
\rho_{X_{k,r-1}^{\prime}}\otimes\rho_{X_{k,r}^{\prime}E}^{s}\\
\otimes\rho_{X_{k,r+1}^{\prime}}\otimes\cdots\otimes\rho_{X_{\left\vert
\mathcal{K}\right\vert ,\left\vert \mathcal{R}\right\vert }^{\prime}}.
\end{multline}
These reduced states are exactly the same as they are in
\eqref{eq:Bob-state-qsdc} and \eqref{eq:Eve-state-qsdc}, and so the same
analysis applies. Bob decodes both the key value $k$ and the local randomness
value $r$ with success probability not smaller than $1-\varepsilon
_{\operatorname{I}}$, as long as%
\begin{equation}
\log_{2}(\left\vert \mathcal{K}\right\vert \left\vert \mathcal{R}\right\vert
)=I_{H}^{\varepsilon_{\operatorname{I}}-\eta_{\operatorname{I}}}(X;B)_{\rho
}-\log_{2}(4\varepsilon_{\operatorname{I}}/\eta_{\operatorname{I}}^{2}),
\end{equation}
while the following security condition holds for all $k\in\mathcal{K}$ and
$s\in\mathcal{S}$:%
\begin{equation}
\frac{1}{2}\left\Vert \rho_{X^{\prime\prime\left\vert \mathcal{K}\right\vert
\left\vert \mathcal{R}\right\vert }E}^{k,s}-\rho_{X^{\prime\prime\left\vert
\mathcal{K}\right\vert \left\vert \mathcal{R}\right\vert }}\otimes
\widetilde{\rho}_{E}^{s}\right\Vert _{1}\leq\varepsilon_{\operatorname{II}},
\end{equation}
for some state $\widetilde{\rho}_{E}^{s}$, as long as%
\begin{equation}
\log_{2}\left\vert \mathcal{R}\right\vert =\sup_{s\in\mathcal{S}}\widetilde
{I}_{\max}^{{\varepsilon_{\operatorname{II}}}-\eta_{\operatorname{II}}%
}(E;X)_{\rho^{s}}+2\log_{2}(1/2\eta_{\operatorname{II}}).
\end{equation}
Thus, the number of key bits that can be established with this scheme is equal
to%
\begin{multline}
\log_{2}\left\vert \mathcal{K}\right\vert =I_{H}^{\varepsilon
_{\operatorname{I}}-\eta_{\operatorname{I}}}(X;B)_{\rho}-\sup_{s\in
\mathcal{S}}\widetilde{I}_{\max}^{{\varepsilon_{\operatorname{II}}}%
-\eta_{\operatorname{II}}}(E;X)_{\rho^{s}}\\
-\log_{2}(4\varepsilon_{\operatorname{I}}/\eta_{\operatorname{I}}^{2}%
)-2\log_{2}(1/2\eta_{\operatorname{II}}),
\end{multline}
as claimed.
\end{proof}

\section{Second-order asymptotics of private communication and secret key
distillation}

\label{sec:second-order-priv-comm-key-dist}

In this section, we show how to apply the one-shot results from
Sections~\ref{sec:priv-comm-comp-wiretap} and
\ref{sec:key-dist-compound-source} to the case of independent and identically
distributed (i.i.d.) resources. First, let us suppose that Alice, Bob, and Eve
are connected by means of a cqq compound wiretap channel of the form
in~\eqref{eq:cqq-ch}:%
\begin{equation}
\mathcal{N}_{X\rightarrow BE}^{s}:x\rightarrow\rho_{BE}^{x,s},
\end{equation}
with all the same assumptions discussed previously in
Section~\ref{sec:one-shot-qsdc}. However, now we allow them to use the channel
multiple times, and we suppose that the particular value of $s$ is fixed but
unknown for all channel uses. Thus, the resource they are employing for
private communication is the tensor-power channel $(\mathcal{N}_{X\rightarrow
BE}^{s})^{\otimes n}$, where $n$ is a large positive integer. This setting is
directly related to a collective attack in quantum key distribution, as
discussed in Section~\ref{sec:app-QKD}.

By applying the result of Theorem~\ref{thm:one-shot-priv-comm} to this
setting, invoking Lemma~\ref{lem:smooth-max-MI-relations} and the second-order
asymptotic expansions discussed in \eqref{eq:hypo-test-rel-ent-2nd-order} and
Corollary~\ref{cor-2nd-order-smooth-dmax}, with $\eta_{\operatorname{I}}%
=\eta_{\operatorname{II}}=1/\sqrt{n}$, we find that it is possible to send
private message bits at the following rate, for sufficiently large$~n$, with
decoding error probability not larger than $\varepsilon_{\operatorname{I}}$
and security parameter not larger than $\varepsilon_{\operatorname{II}}$:%
\begin{multline}
\frac{\log_{2}\left\vert \mathcal{M}\right\vert }{n}=\sup_{p_{X}}I(X;B)_{\rho
}+\sqrt{\frac{1}{n}V(X;B)_{\rho}}\Phi^{-1}(\varepsilon_{\operatorname{I}%
})\label{eq:ach-rate-second-order}\\
-\sup_{s\in\mathcal{S}}\left[  I(X;E)_{\rho^{s}}-\sqrt{\frac{1}{n}%
V(X;E)_{\rho^{s}}}\Phi^{-1}({\varepsilon_{\operatorname{II}}^2})\right] \\
+O\!\left(  \frac{\log n}{n}\right)  .
\end{multline}
In the above, the information quantities are evaluated with respect to the
following classical--quantum state:%
\begin{equation}
\rho_{XBE}^{s}\coloneqq \sum_{x}p_{X}(x)|x\rangle\!\langle x|_{X}\otimes\rho_{BE}^{x,s},
\end{equation}
and%
\begin{align}
I(X;B)_{\rho}  &  \coloneqq \sum_{x}p_{X}(x)D(\rho_{B}^{x}\Vert\rho_{B}),\\
I(X;E)_{\rho^{s}}  &  \coloneqq \sum_{x}p_{X}(x)D(\rho_{E}^{x,s}\Vert\rho_{E}^{s})\\
\rho_{B}  &  \coloneqq \sum_{x}p_{X}(x)\rho_{B}^{x},\\
\rho_{E}^{s}  &  \coloneqq \sum_{x}p_{X}(x)\rho_{E}^{x,s}.
\end{align}
The Holevo information variances $V(X;B)_{\rho}$ and $V(X;E)_{\rho^{s}}$ are
given by%
\begin{multline}
V(X;B)_{\rho}\coloneqq \\
\sum_{x\in\mathcal{X}}p_{X}(x)\left[  V(\rho_{B}^{x}\Vert\rho_{B})+\left[
D(\rho_{B}^{x}\Vert\rho_{B})\right]  ^{2}\right]  -\left[  I(X;B)_{\rho
}\right]  ^{2},
\end{multline}%
\begin{multline}
V(X;E)_{\rho^{s}}\coloneqq \\
\sum_{x\in\mathcal{X}}p_{X}(x)\left[  V(\rho_{E}^{x,s}\Vert\rho_{E}%
^{s})+\left[  D(\rho_{E}^{x,s}\Vert\rho_{E}^{s})\right]  ^{2}\right]  -\left[
I(X;E)_{\rho^{s}}\right]  ^{2},
\end{multline}
with these formulas considered in more detail in
Appendix~\ref{app:mut-info-var-and-T}. The term $O\!\left(  \frac{\log n}%
{n}\right)  $ hides constants involving $\varepsilon_{\operatorname{I}}$,
$\varepsilon_{\operatorname{II}}$, $T(X;B)_{\rho}$, and $T(X;E)_{\rho^{s}}$,
with the latter two quantities defined in
Appendix~\ref{app:mut-info-var-and-T}. By noting that all of the quantities
$I(X;B)_{\rho}$, $I(X;E)_{\rho^{s}}$, $V(X;B)_{\rho}$, $V(X;E)_{\rho^{s}}$,
$T(X;B)_{\rho}$, and $T(X;E)_{\rho^{s}}$ involve an expectation with respect
to the distribution $p_{X}$, it follows from an approximation argument that
the same formula in \eqref{eq:ach-rate-second-order} is an achievable rate for
private communication when $p_{X}$ is a probability distribution over a
continuous alphabet. In this case, all of the expectations for the various
quantities are evaluated by integration.

In the case that the states involved in the above formulas act on separable
Hilbert spaces, then we should define the various quantities in more detail
because the formulas in \eqref{eq:rel-ent-q} and
\eqref{eq:rel-ent-q-var}\ need to be interpreted in a particular way (that is,
the trace cannot be taken with respect to an arbitrary orthonormal basis).
Suppose that $\rho$ and $\sigma$ have the following spectral decompositions:%
\begin{equation}
\rho=\sum_{x}\lambda_{x}|\psi_{x}\rangle\!\langle\psi_{x}|,\quad\sigma=\sum
_{y}\mu_{y}|\phi_{y}\rangle\!\langle\phi_{y}|.
\label{eq:rho-sig-spectral-decomps}%
\end{equation}
Then the quantum relative entropy \cite{F70,Lindblad1973} and the relative
entropy variance \cite{TH12,li12,KW17a} are defined as%
\begin{align}
D(\rho\Vert\sigma)  &  \coloneqq \sum_{x,y}\left\vert \langle\phi_{y}|\psi_{x}%
\rangle\right\vert ^{2}\lambda_{x}\log_{2}\!\left(  \frac{\lambda_{x}}{\mu
_{y}}\right)  ,\label{eq:rel-ent-sep}\\
V(\rho\Vert\sigma)  &  \coloneqq \sum_{x,y}\left\vert \langle\phi_{y}|\psi_{x}%
\rangle\right\vert ^{2}\lambda_{x}\left[  \log_{2}\!\left(  \frac{\lambda_{x}%
}{\mu_{y}}\right)  -D(\rho\Vert\sigma)\right]  ^{2}.
\label{eq:rel-ent-var-sep}%
\end{align}

We can also consider the application to secret key distillation. Suppose now
that Alice, Bob, and Eve share $n$~copies of the following cqq state:%
\begin{equation}
\rho_{XBE}^{s}\coloneqq \sum_{x}p_{X}(x)|x\rangle\!\langle x|_{X}\otimes\rho_{BE}^{x,s},
\label{eq:cqq-state-iid}%
\end{equation}
with all of the same assumptions discussed previously in
Section~\ref{sec:secret-key-dist}. Then by the same reasoning as given above,
but this time applying Theorem~\ref{thm:one-shot-key-dist}, the following
formula represents an achievable rate for secret key distillation, for
sufficiently large $n$ and with decoding error probability not larger than
$\varepsilon_{\operatorname{I}}$ and security parameter not larger than
$\varepsilon_{\operatorname{II}}$:%
\begin{multline}
\frac{\log_{2}\left\vert \mathcal{K}\right\vert }{n}=I(X;B)_{\rho}+\sqrt
{\frac{1}{n}V(X;B)_{\rho}}\Phi^{-1}(\varepsilon_{\operatorname{I}%
})\label{eq:second-order-key-rate}\\
-\sup_{s\in\mathcal{S}}\left[  I(X;E)_{\rho^{s}}-\sqrt{\frac{1}{n}%
V(X;E)_{\rho^{s}}}\Phi^{-1}({\varepsilon_{\operatorname{II}}^2})\right] \\
+O\!\left(  \frac{\log n}{n}\right)  ,
\end{multline}
where all information quantities are evaluated with respect to the state in \eqref{eq:cqq-state-iid}.

\section{Discussion of collective attacks in QKD and the compound wiretap
channel and source}

\label{sec:col-att-QKD-comp-wiretap}As mentioned in the introduction, from
what we can gather by combing through the literature, it seems that there is a
disconnect between the community of researchers working on security proofs
against collective attacks in quantum key distribution and those working on
key distillation from a compound wiretap source. Thus, one contribution of
this paper is to connect these two research directions, with
Theorem~\ref{thm:one-shot-key-dist} and its applications in
Sections~\ref{sec:app-QKD} and \ref{sec:second-order-priv-comm-key-dist}%
\ providing a direct link between them.

The main purpose of this section is to trace some of the historical
developments and early roots of these disparate communities. This might help
with unifying them going forward. To begin with, let us note that there are
several reviews of quantum key distribution in which collective attacks are
discussed \cite{SBCDLP09,Lut14,DLQY16,XMZLP19}. Furthermore, there is a recent
review of the compound wiretap channel that traces its development in
classical information theory \cite{SBP15}.

The compound channel for classical communication was introduced in
\cite{blackwell1959,Wolfowitz1959}. The basic idea here is that the actual
communication channel is chosen from an uncertainty set $\{p_{Y|X}^{s}\}_{s}$,
it is the same over a large blocklength (used in an i.i.d.~way), and the goal
is to be able to communicate regardless of which channel is selected. The main
result of \cite{blackwell1959,Wolfowitz1959} is that Shannon's formula for
capacity becomes modified to%
\begin{equation}
\max_{p_{X}}\min_{s}I(X;Y)_{p^{s}},
\end{equation}
where $p_{XY}^{s}=p_{X}p_{Y|X}^{s}$.

The classical wiretap channel $p_{YZ|X}$ and its private capacity were
introduced in \cite{W75}\ and studied further by \cite{CK78,Csi96,MW00}. In
this model, the sender inputs the random variable $X$, the legitimate receiver
obtains $Y$, and the wiretapper (eavesdropper) $Z$. The conditional
probability distribution $p_{YZ|X}$ is known to all parties involved. The
culmination of these papers was to identify the formula%
\begin{equation}
\max_{U\rightarrow X\rightarrow YZ}\left[ I(U;Y)-I(U;Z)\right]
\label{eq:classical-wiretap-formula}%
\end{equation}
as the private capacity of an arbitrary classical wiretap channel. The
optimization is over Markov chains $U\rightarrow X\rightarrow YZ$,
where $U$ is known as an auxiliary random variable. Interestingly, this
formula provides the insight that noise at the encoder (the channel from $X$
to $U$) can increase private capacity, with the reasoning being that, even
though it decreases both informations $I(U;Y)$ and $I(U;Z)$, it can happen
that it decreases the wiretapper's information $I(U;Z)$ by more, so that there
can be a noise benefit. The main reason why the classical wiretap channel
model has not been embraced by the classical cryptography community is that
the model assumes that the channel to the eavesdropper is known, and this is
too much to assume in practice. Interestingly, it would be some years before
the compound wiretap channel and source were defined.

The source model for secret key agreement was introduced by \cite{AC93}. In
this model, two legitimate parties have access to respective random variables
$X$ and $Y$ and an eavesdropper has access to a random variable $Z$ described
by the joint probability distribution $p_{XYZ}$, and the legitimate parties
are allowed public classical communication. All parties know the distribution
$p_{XYZ}$. If public classical communication is allowed only in one round from
$X$ to $Y$, then the secret key agreement capacity is given by the formula
\cite{AC93}%
\begin{equation}
\max_{p_{UV|X}}\left[  I(V;Y|U)-I(V;Z|U)\right]  .
\end{equation}
In independent work, it was shown that public communication can enhance the
secret key agreement capacity of a wiretap channel \cite{M93}, and an upper
bound on the secret key agreement capacity was established in the case that
arbitrary public classical communication is allowed.

The notion of a collective attack in quantum key distribution was proposed by
\cite{BM97,Biham2002}. Since then, it has been studied intensively in quantum
key distribution \cite{SBCDLP09,Lut14,DLQY16,XMZLP19}. Quantum de Finetti
theorems and their variants were proved \cite{Renner2005,CKR09}, which
demonstrate that general coherent (arbitrary) attacks are no better than
collective attacks in the limit of large blocklength. Quantum de Finetti
theorems thus establish the significance of focusing on collective attacks in
the context of quantum key distribution.

The quantum wiretap channel was proposed and studied in \cite{D05,CWY04}. In
this model, the wiretap channel is given by $x\rightarrow\rho_{BE}^{x}$, with
the sender having access to the input, the legitimate receiver to the quantum
system $B$, and the eavesdropper to the quantum system $E$. It is assumed that
the full channel is known to both the sender and legitimate receiver. An
important result from \cite{D05,CWY04} is that
\begin{equation}
\max_{U\rightarrow X} \left[  I(U;B)-I(U;E) \right]
\end{equation}
is an achievable rate at which they can communicate privately, paralleling the
formula in \eqref{eq:classical-wiretap-formula} for the classical case. It is
not known whether this achievable rate is optimal in general. That is, it is
an open problem to determine the private capacity of the classical--quantum
wiretap channel in the general case.

The quantum wiretap source was proposed and studied in \cite{Devetak2005}. The
model is that Alice, Bob, and Eve share a state of the form $\sum
_{x}p(x)|x\rangle\!\langle x|_{X}\otimes\rho_{BE}^{x}$, and the goal is to use
many copies of this state along with public classical communication in order
to distill a secret key. The authors of \cite{Devetak2005} found that the rate%
\begin{equation}
I(X;B)-I(X;E) \label{eq:DW-formula-hist}%
\end{equation}
is achievable for key distillation using many copies of the aforementioned
state. It is important to stress that, in this model, it is assumed that Alice
and Bob have full knowledge of the state, including the state of the
eavesdropper. For this reason, it is not justified to apply the formula in
\eqref{eq:DW-formula-hist}\ generally when analyzing the security of quantum
key distribution against collective attacks. In the special case that the
quantum key distribution protocol involves preparations and measurements that
are tomographically complete, the legitimate parties determine the state of
the eavesdropper up to an information-theoretically irrelevant isometry, and
it is then justified to apply the formula in \eqref{eq:DW-formula-hist} for
security against collective attacks.

The classical compound wiretap channel was proposed and studied in
\cite{Liang2009} and studied further in \cite{Bjelakovic2013} (see also
\cite{SBP15} for a review). The model is that the actual wiretap channel is
chosen from an uncertainty set $\{p_{YZ|X}^{s}\}_{s}$, it is the same over a
large blocklength (used in an i.i.d.~way), and the goal is to be able to
communicate privately regardless of which channel is selected. One critical
result of \cite{Liang2009} is that the following formula is an achievable rate
for private communication in this setting, when the uncertainty set is finite:%
\begin{equation}
\max_{U\rightarrow X\rightarrow Y_{s}Z_{s}}\left[  \min_{s}%
I(U;Y_{s})-\max_{s}I(U;Z_{s})\right]  .
\end{equation}
The classical compound wiretap channel model is more acceptable from a
cryptographic perspective than is the standard wiretap channel model, because
it allows for uncertainty in the channel to the wiretapper. However, it is
still  difficult to argue that anything about the channel to the wiretapper
would be known in a fully classical context.

The classical compound wiretap source was considered in \cite{BW13} and
further studied in \cite{TBS17}. The special case of a classical compound
source with one fixed marginal was considered previously in \cite{Bl10}. The
model from \cite{BW13,TBS17} is that two legitimate parties and an
eavesdropper share a source chosen from an uncertainty set $\left\{
p_{XYZ}^{s}\right\}  _{s}$, and they are allowed public classical
communication to assist in the task of distilling a secret key. In
\cite{TBS17}, a formula for the achievable distillable secret key was given,
and we refer to \cite{TBS17} for the details. In the special case that the
marginals $X$ and $Y$ are fixed, the results of \cite{Bl10,TBS17} imply that
the rate%
\begin{equation}
I(X;Y)-\max_{s}I(X;Z)\
\end{equation}
is achievable for key distillation.

The compound quantum wiretap source was proposed and studied in \cite{BJ16}.
The model studied there is that the legitimate parties and the eavesdropper
share a state $\rho_{XBE}^{s}=\sum_{x}p^{s}(x)|x\rangle\!\langle x|_{X}%
\otimes\rho_{BE}^{x,s}$ selected from an uncertainty set $\left\{  \rho
_{XBE}^{s}\right\}  _{s}$, and the goal is to use one round of forward public
classical communication from $X$ to $B$ in order to distill a secret key. In
\cite{BJ16}, a formula for the achievable distillable secret key was given,
and we refer to \cite{BJ16} for the details. In the special case that the
marginals $X$ and $B$ are fixed, the results of \cite{BJ16} imply that the
rate%
\begin{equation}
I(X;B)-\max_{s}I(X;E)\label{eq:distillable-key-compound-Boche}%
\end{equation}
is achievable. An even further special case is when the $B$ system is
classical. Thus, the results of \cite{BJ16} can be used to analyze the
security of quantum key distribution against collective attacks. This model is
acceptable from a cryptographic setting because one can determine the
uncertainty set for the state of the eavesdropper during the parameter
estimation round of a quantum key distribution protocol. This is possible
under the assumption that the eavesdropper is constrained by the laws of
quantum mechanics and due to the structure of quantum mechanics. One of the
contributions of our paper is to extend these results (with fixed marginal on
$X$ and $B$) to a second-order coding rate and have the formula apply to
infinite-dimensional systems of interest in continuous-variable quantum key distribution.

The compound quantum wiretap channel was presented in \cite{DH10,BCCD14}. In
\cite{DH10}, the model is that the wiretap channel is selected from an
uncertainty set $x\rightarrow\rho_{BE}^{x,s}$, indexed by $s$, the channel is
the same over the whole blocklength (used in an i.i.d.~way), and the goal is
to communicate privately to the legitimate receiver $B$ regardless of which
channel is selected from the set. It is unclear whether the claims of
\cite{DH10} hold as generally as stated therein, but in the case that public
shared randomness is available to the legitimate parties, then the claims of
\cite{DH10} appear to hold. The main issue with the claims of \cite{DH10} is
that the code is derandomized, and when this is done, it is no longer the case
that the code is secure against all of the possible states of the adversaries,
as claimed therein. One would need to apply a union bound argument, as is done
in \cite[Section~IV]{BCCD14}. However, this issue does not arise if the code
is not derandomized (in the case that the legitimate parties have public
shared randomness). The coding scheme of \cite{DH10} is universal in the sense
that all that is required for communication is a lower bound on the Holevo
information with the legitimate receiver and an upper bound on the Holevo
information with the wiretapper. The paper \cite{BCCD14} considered the same
model, but with the uncertainty set being known and finite. They found that
the following rate is achievable for private communication in this setting:%
\begin{equation}
\ \max_{U\rightarrow X\rightarrow B_{s}E_{s}}\left[  \min_{s}%
I(U;B_{s})-\max_{s}I(U;E_{s})\right]  ,
\end{equation}
representing a quantum generalization of the main result of~\cite{Liang2009}.

Going forward from here, it would be interesting to develop the compound
wiretap channel and source in more detail, in order to extend the scope of the
results and given the applicability to key distillation in quantum key
distribution. This is one contribution of the present paper, since we have
established second-order coding rates for compound wiretap channels and
sources for infinite-dimensional quantum systems when there is a fixed
marginal. Furthermore, the quantum key distribution application is well
motivated from a cryptographic perspective, given that it is possible to
constrain the possible states (uncertainty set) of an eavesdropper who
operates according to the laws of quantum mechanics. This is a critical
difference with the classical compound wiretap channel, in which it is not
possible to do so.

\subsection{Remarks on the Devetak--Winter formula and security against
collective attacks}

We now comment briefly on the use of the Devetak--Winter argument in the
context of collective attacks in quantum key distribution. The formula in
\eqref{eq:distillable-key-compound-Boche} has been consistently employed by
the quantum key distribution community as an achievable rate for distillable
key against collective attacks. This formula is indeed correct. However, as
discussed above, the argument of \cite{Devetak2005} did not justify this
formula. Instead, it argued for the achievability of
\eqref{eq:distillable-key-compound-Boche} when there is a known collective
attack, which is applicable in the case that the preparation and measurement
procedure in a QKD protocol is tomographically complete. Specifically, one can
analyze the proof of \cite[Theorem~2.1]{Devetak2005} to see that the key
distillation scheme constructed ends up depending on the state shared by all
three parties and security is only guaranteed in this case. From what we can
tell, the formula in \eqref{eq:distillable-key-compound-Boche} was first
proven by \cite{RGK05}\ for the qubit case and then in \cite[Corollary~6.5.2]%
{Renner2005} for general finite-dimensional states. The latter case has also
been considered in \cite{BJ16}. For infinite-dimensional states, the present
paper has established a proof that the formula in
\eqref{eq:distillable-key-compound-Boche} is achievable.

\section{Other contributions}

\label{sec:other-contribs}In this section, we briefly list some other
consequences of our work. First, the task of entanglement-assisted private
communication over a broadcast channel was considered recently in
\cite{Qi_2018}. The technique behind the proof of
Theorem~\ref{thm:one-shot-priv-comm}\ applies to this setting, allowing for
protection against a quantum wiretap channel with fixed marginal to the
decoding set. The result also applies to the infinite-dimensional case
(separable Hilbert spaces).

We can also combine the techniques from Theorem~\ref{thm:one-shot-priv-comm}
with those in \cite[Theorem~1]{AJW19} and \cite[Corollary~2]{Sen18}\ to find
achievable rates for one-shot private communication over a compound wiretap
channel with a finite uncertainty set, when the sender and receiver possess
public shared randomness. These results can then be extended to second-order
coding rates, and they also hold for the infinite-dimensional case (separable
Hilbert spaces).

Our results lead to achievable second-order coding rates for classical
communication over single-mode bosonic channels, when assisted by randomness,
and thus extend the findings of \cite{WRG16,OMW19}. Our results also lead to
achievable second-order coding rates for private classical communication over
the same channels, when assisted by public shared randomness, thus
generalizing the findings of \cite{wilde2017position}.

\section{Conclusion}

\label{sec:conclusion}

In this paper, we have provided a second-order analysis of quantum key distribution as a bridge between the asymptotic and non-asymptotic regimes. The technical contributions
that allowed for this advance are a coding theorem for one-shot key
distillation from a compound quantum wiretap source with fixed marginal, as
well as the establishment of the second-order asymptotics for the smooth
max-relative entropy. Another contribution is a coding theorem for private
communication over a compound quantum wiretap channel with fixed marginal. We
also showed how to optimize the second-order coding rate for several exemplary
QKD\ protocols, including six-state, BB84, and continuous-variable QKD. In
Section~\ref{sec:other-contribs}, we briefly mentioned several other immediate
applications of our technical results.

Going forward from here, an important open problem is to provide a full
second-order analysis of both the parameter estimation and key distillation
steps of a quantum key distribution protocol. At the moment, we have
exclusively analyzed second-order coding rates for the key distillation step,
under the assumption that the parameter estimation step provides reliable
estimates. It seems plausible to incorporate the latest developments from
\cite{Hay19}, in combination with the methods of this paper, in order to have
a complete second-order analysis of both parameter estimation and key distillation.

Another practical concern is to reduce the amount of public shared randomness
that the protocol from Theorem~\ref{thm:one-shot-priv-comm}\ employs.
Likewise, it would be ideal to reduce the amount of public classical
communication used by the protocol from Theorem~\ref{thm:one-shot-key-dist}.
Even though these resources are typically considered essentially free in a
private communication setting, it would still be ideal to minimize their
consumption. We note here that this large usage of a free resource is typical
of the methods of position-based coding \cite{AJW17}\ and convex splitting
\cite{ADJ17}, which seems to be the cost for obtaining such simple formulas in
the one-shot case. Recent work in other domains has shown how to reduce the
amount of free resource that these kinds of protocols consume \cite{Ans18}.

We think it would also be interesting to consider secret key distillation in
the setting of private communication per unit cost from \cite{DPW19}. It seems
likely that\ the key distillation protocol given in
Theorem~\ref{thm:one-shot-key-dist}\ could be helpful for this task.

We acknowledge discussions with Nilanjana Datta, Anthony Leverrier, Stefano Mancini, and Samad Khabbazi Oskouei. MMW, SK,
and EK acknowledge support from the National Science Foundation under grant
no.~1714215. SG acknowledges support from the National Science Foundation
(NSF) RAISE-EQuIP: \textquotedblleft Engineering Quantum Integrated Platforms
for Quantum Communication\textquotedblright\ Program, under University of
Arizona Grant Number 1842559.

\bibliographystyle{alpha}
\bibliography{Ref}

\newcommand{\etalchar}[1]{$^{#1}$}
\begin{thebibliography}{TMMPE17}

\bibitem[ABJT18]{ABJT18}
Anurag Anshu, Mario Berta, Rahul Jain, and Marco Tomamichel.
\newblock Partially smoothed information measures.
\newblock July 2018.
\newblock arXiv:1807.05630v2.

\bibitem[ABJT19]{ABJT19}
Anurag Anshu, Mario Berta, Rahul Jain, and Marco Tomamichel.
\newblock A minimax approach to one-shot entropy inequalities.
\newblock June 2019.
\newblock arXiv:1906.00333v1.

\bibitem[AC93]{AC93}
Rudolf {Ahlswede} and Imre {Csiszar}.
\newblock Common randomness in information theory and cryptography. i. secret
  sharing.
\newblock {\em IEEE Transactions on Information Theory}, 39(4):1121--1132, July
  1993.

\bibitem[ADH08]{ADH08}
Remigiusz Augusiak, Maciej Demianowicz, and {Pawe\l{}} Horodecki.
\newblock Universal observable detecting all two-qubit entanglement and
  determinant-based separability tests.
\newblock {\em Physical Review A}, 77(3):030301, March 2008.
\newblock arXiv:quant-ph/0604109.

\bibitem[ADJ17]{ADJ17}
Anurag Anshu, Vamsi~Krishna Devabathini, and Rahul Jain.
\newblock Quantum communication using coherent rejection sampling.
\newblock {\em Physical Review Letters}, 119(12):120506, September 2017.
\newblock arXiv:1410.3031.

\bibitem[AJW19a]{AJW19}
Anurag {Anshu}, Rahul {Jain}, and Naqueeb~A. {Warsi}.
\newblock A hypothesis testing approach for communication over
  entanglement-assisted compound quantum channel.
\newblock {\em IEEE Transactions on Information Theory}, 65(4):2623--2636,
  April 2019.
\newblock arXiv:1706.08286.

\bibitem[AJW19b]{AJW17}
Anurag Anshu, Rahul Jain, and Naqueeb~Ahmad Warsi.
\newblock One shot entanglement assisted classical and quantum communication
  over noisy quantum channels: A hypothesis testing and convex split approach.
\newblock {\em IEEE Transactions on Information Theory}, 65(2):1287--1306,
  February 2019.
\newblock arXiv:1702.01940.

\bibitem[AMKB11]{AMKB11}
Silvestre Abruzzo, Markus Mertz, Hermann Kampermann, and Dagmar Bruss.
\newblock {Finite-key analysis of the six-state protocol with photon number
  resolution detectors}.
\newblock In Roberto Zamboni, François Kajzar, Attila~A. Szep, Mark~T.
  Gruneisen, Miloslav Dusek, John~G. Rarity, Colin Lewis, and Douglas Burgess,
  editors, {\em Optics and Photonics for Counterterrorism and Crime Fighting
  VII; Optical Materials in Defence Systems Technology VIII; and
  Quantum-Physics-based Information Security}, volume 8189, pages 352--362.
  International Society for Optics and Photonics, SPIE, 2011.
\newblock arXiv:1111.2798.

\bibitem[Ans18]{Ans18}
Anurag Anshu.
\newblock {\em One-Shot Protocols for Communication over Quantum Networks:
  Achievability and Limitations}.
\newblock PhD thesis, National University of Singapore, May 2018.
\newblock \url{https://scholarbank.nus.edu.sg/handle/10635/146923}.

\bibitem[BB84]{bb84}
Charles~H. Bennett and Gilles Brassard.
\newblock Quantum cryptography: Public key distribution and coin tossing.
\newblock In {\em Proceedings of IEEE International Conference on Computers
  Systems and Signal Processing}, pages 175--179, Bangalore, India, December
  1984.

\bibitem[BBB{\etalchar{+}}02]{Biham2002}
Eli Biham, Michel Boyer, Gilles Brassard, Jeroen van~de Graaf, and Tal Mor.
\newblock Security of quantum key distribution against all collective attacks.
\newblock {\em Algorithmica}, 34(4):372--388, November 2002.
\newblock arXiv:quant-ph/9801022.

\bibitem[BBS13]{Bjelakovic2013}
Igor Bjelakovi{\'{c}}, Holger Boche, and Jochen Sommerfeld.
\newblock Secrecy results for compound wiretap channels.
\newblock {\em Problems of Information Transmission}, 49(1):73--98, January
  2013.
\newblock arXiv:1106.2013.

\bibitem[BBT59]{blackwell1959}
David Blackwell, Leo Breiman, and A.~J. Thomasian.
\newblock The capacity of a class of channels.
\newblock {\em The Annals of Mathematical Statistics}, 30(4):1229--1241, 12
  1959.

\bibitem[BCCD14]{BCCD14}
Holger Boche, Minglai Cai, Ning Cai, and Christian Deppe.
\newblock Secrecy capacities of compound quantum wiretap channels and
  applications.
\newblock {\em Physical Review A}, 89(5):052320, May 2014.
\newblock arXiv:1302.3412.

\bibitem[BCR11]{BCR09}
Mario Berta, Matthias Christandl, and Renato Renner.
\newblock The quantum reverse {Shannon} theorem based on one-shot information
  theory.
\newblock {\em Communications in Mathematical Physics}, 306(3):579--615, August
  2011.
\newblock arXiv:0912.3805.

\bibitem[BD10]{BD10}
Francesco Buscemi and Nilanjana Datta.
\newblock The quantum capacity of channels with arbitrarily correlated noise.
\newblock {\em IEEE Transactions on Information Theory}, 56(3):1447--1460,
  March 2010.
\newblock arXiv:0902.0158.

\bibitem[BD11]{BD11}
Fernando G. S.~L. Brandao and Nilanjana Datta.
\newblock One-shot rates for entanglement manipulation under non-entangling
  maps.
\newblock {\em IEEE Transactions on Information Theory}, 57(3):1754--1760,
  March 2011.
\newblock arXiv:0905.2673.

\bibitem[BDSW96]{BDSW96}
Charles~H. Bennett, David~P. DiVincenzo, John~A. Smolin, and William~K.
  Wootters.
\newblock Mixed-state entanglement and quantum error correction.
\newblock {\em Physical Review A}, 54(5):3824--3851, November 1996.
\newblock arXiv:quant-ph/9604024.

\bibitem[BJ16]{BJ16}
Holger Boche and Gisbert Jan{\ss}en.
\newblock Distillation of secret-key from a class of compound memoryless
  quantum sources.
\newblock {\em Journal of Mathematical Physics}, 57(8):082201, August 2016.
\newblock arXiv:1604.05530.

\bibitem[{Blo}10]{Bl10}
Matthieu {Bloch}.
\newblock Channel intrinsic randomness.
\newblock In {\em 2010 IEEE International Symposium on Information Theory},
  pages 2607--2611, June 2010.

\bibitem[BM97]{BM97}
Eli Biham and Tal Mor.
\newblock Security of quantum cryptography against collective attacks.
\newblock {\em Physical Review Letters}, 78(11):2256--2259, March 1997.
\newblock arXiv:quant-ph/9605007.

\bibitem[BPG99]{BG99}
H.~Bechmann-Pasquinucci and Nicolas Gisin.
\newblock Incoherent and coherent eavesdropping in the six-state protocol of
  quantum cryptography.
\newblock {\em Physical Review A}, 59(6):4238--4248, June 1999.

\bibitem[{Bru}98]{B98}
Dagmar {Bru\ss{}}.
\newblock Optimal eavesdropping in quantum cryptography with six states.
\newblock {\em Physical Review Letters}, 81(14):3018--3021, October 1998.
\newblock arXiv:quant-ph/9805019.

\bibitem[BW13]{BW13}
Holger {Boche} and Rafael~F. {Wyrembelski}.
\newblock Secret key generation using compound sources - optimal key-rates and
  communication costs.
\newblock In {\em SCC 2013; 9th International ITG Conference on Systems,
  Communication and Coding}, pages 1--6, January 2013.

\bibitem[BW18]{Bradler2018}
Kamil Br\'adler and Christian Weedbrook.
\newblock Security proof of continuous-variable quantum key distribution using
  three coherent states.
\newblock {\em Physical Review A}, 97(2):022310, February 2018.
\newblock arXiv:1709.01758.

\bibitem[CCD12]{CCD12}
Minglai Cai, Ning Cai, and Christian Deppe.
\newblock Capacities of classical compound quantum wiretap and classical
  quantum compound wiretap channels.
\newblock In {\em Proceedings of the 2012 IEEE International Symposium on
  Information Theory}, pages 726--730, July 2012.
\newblock arXiv:1202.0773.

\bibitem[CG06]{CG06}
Filippo Caruso and Vittorio Giovannetti.
\newblock Degradability of bosonic {Gaussian} channels.
\newblock {\em Physical Review A}, 74(6):062307, December 2006.
\newblock arXiv:quant-ph/0603257.

\bibitem[CGH06]{CGH06}
Filippo Caruso, Vittorio Giovannetti, and Alexander~S. Holevo.
\newblock One-mode bosonic {Gaussian} channels: a full weak-degradability
  classification.
\newblock {\em New Journal of Physics}, 8(12):310, December 2006.
\newblock arXiv:quant-ph/0609013.

\bibitem[Che05]{PhysRevA.71.062320}
Xiao-yu Chen.
\newblock Gaussian relative entropy of entanglement.
\newblock {\em Physical Review A}, 71(6):062320, June 2005.
\newblock arXiv:quant-ph/0402109.

\bibitem[CK78]{CK78}
Imre Csisz\'{a}r and Janos K\"orner.
\newblock Broadcast channels with confidential messages.
\newblock {\em IEEE Transactions on Information Theory}, 24(3):339--348, May
  1978.

\bibitem[CKR09]{CKR09}
Matthias Christandl, Robert K\"{o}nig, and Renato Renner.
\newblock Postselection technique for quantum channels with applications to
  quantum cryptography.
\newblock {\em Physical Review Letters}, 102(2):020504, January 2009.
\newblock arXiv:0809.3019.

\bibitem[CLL04]{CLL04}
Marcos Curty, Maciej Lewenstein, and Norbert L\"utkenhaus.
\newblock Entanglement as a precondition for secure quantum key distribution.
\newblock {\em Physical Review Letters}, 92(21):217903, May 2004.
\newblock arXiv:quant-ph/0307151.

\bibitem[CML16]{CML16}
Patrick~J. Coles, Eric~M. Metodiev, and Norbert L\"utkenhaus.
\newblock Numerical approach for unstructured quantum key distribution.
\newblock {\em Nature Communications}, 7:11712, May 2016.
\newblock arXiv:1510.01294.

\bibitem[CN97]{CN97}
Isaac~L. Chuang and M.~A. Nielsen.
\newblock Prescription for experimental determination of the dynamics of a
  quantum black box.
\newblock {\em Journal of Modern Optics}, 44(11-12):2455--2467, 1997.
\newblock arXiv:quant-ph/9610001.

\bibitem[Csi96]{Csi96}
Imre Csisz{\'a}r.
\newblock Almost independence and secrecy capacity.
\newblock {\em Problems of Information Transmission}, 32(1):40--47, 1996.

\bibitem[CWY04]{CWY04}
Ning Cai, Andreas Winter, and Raymond~W. Yeung.
\newblock Quantum privacy and quantum wiretap channels.
\newblock {\em Problems of Information Transmission}, 40(4):318--336, October
  2004.

\bibitem[Dat09a]{Dat09}
Nilanjana Datta.
\newblock Max-relative entropy of entanglement, alias log robustness.
\newblock {\em International Journal of Quantum Information}, 7(02):475--491,
  January 2009.
\newblock arXiv:0807.2536.

\bibitem[Dat09b]{Datta2009b}
Nilanjana Datta.
\newblock Min- and max-relative entropies and a new entanglement monotone.
\newblock {\em IEEE Transactions on Information Theory}, 55(6):2816--2826, June
  2009.
\newblock arXiv:0803.2770.

\bibitem[Dev05]{D05}
Igor Devetak.
\newblock The private classical capacity and quantum capacity of a quantum
  channel.
\newblock {\em IEEE Transactions on Information Theory}, 51(1):44--55, January
  2005.
\newblock arXiv:quant-ph/0304127.

\bibitem[DF19]{DF19}
Frederic Dupuis and Omar Fawzi.
\newblock Entropy accumulation with improved second-order term.
\newblock {\em IEEE Transactions on Information Theory (to appear)}, 2019.
\newblock arXiv:1805.11652.

\bibitem[DH10]{DH10}
Nilanjana Datta and Min-Hsiu Hsieh.
\newblock Universal coding for transmission of private information.
\newblock {\em Journal of Mathematical Physics}, 51(12):122202, December 2010.
\newblock arXiv:1007.2629.

\bibitem[DHCB05]{DHCB05}
W.~{D\"ur}, M.~Hein, J.~I. Cirac, and H.-J. Briegel.
\newblock Standard forms of noisy quantum operations via depolarization.
\newblock {\em Physical Review A}, 72(5):052326, November 2005.
\newblock arXiv:quant-ph/0507134.

\bibitem[DL15]{DL15}
Nilanjana Datta and Felix Leditzky.
\newblock Second-order asymptotics for source coding, dense coding, and
  pure-state entanglement conversions.
\newblock {\em IEEE Transactions on Information Theory}, 61(1):582--608,
  January 2015.
\newblock arXiv:1403.2543.

\bibitem[DLQY16]{DLQY16}
Eleni Diamanti, Hoi-Kwong Lo, Bing Qi, and Zhiliang Yuan.
\newblock Practical challenges in quantum key distribution.
\newblock {\em npj Quantum Information}, 2:16025, November 2016.
\newblock arXiv:1606.05853.

\bibitem[DPR16]{DPR15}
Nilanjana Datta, Yan Pautrat, and Cambyse Rouz\'{e}.
\newblock Second-order asymptotics for quantum hypothesis testing in settings
  beyond i.i.d. - quantum lattice systems and more.
\newblock {\em Journal of Mathematical Physics}, 57(6):062207, June 2016.
\newblock arXiv:1510.04682.

\bibitem[DPW19]{DPW19}
Dawei {Ding}, Dimitri~S. {Pavlichin}, and Mark~M. {Wilde}.
\newblock Quantum channel capacities per unit cost.
\newblock {\em IEEE Transactions on Information Theory}, 65(1):418--435,
  January 2019.
\newblock arXiv:1705.08878.

\bibitem[DTW16]{DTW14}
Nilanjana Datta, Marco Tomamichel, and Mark~M. Wilde.
\newblock On the second-order asymptotics for entanglement-assisted
  communication.
\newblock {\em Quantum Information Processing}, 15(6):2569--2591, June 2016.
\newblock arXiv:1405.1797.

\bibitem[DW05]{Devetak2005}
Igor Devetak and Andreas Winter.
\newblock Distillation of secret key and entanglement from quantum states.
\newblock {\em Proceedings of the Royal Society A}, 461(2053):207--235, January
  2005.
\newblock arXiv:quant-ph/0306078.

\bibitem[Eke91]{E91}
Artur~K. Ekert.
\newblock Quantum cryptography based on {Bell's} theorem.
\newblock {\em Physical Review Letters}, 67(6):661--663, August 1991.

\bibitem[Fal70]{F70}
Harold Falk.
\newblock Inequalities of {J. W. Gibbs}.
\newblock {\em American Journal of Physics}, 38(7):858--869, July 1970.

\bibitem[FAR11]{FAR11}
Fabian Furrer, Johan {{A}berg}, and Renato Renner.
\newblock Min- and max-entropy in infinite dimensions.
\newblock {\em Communications in Mathematical Physics}, 306(1):165--186, August
  2011.
\newblock arXiv:1004.1386.

\bibitem[FFB{\etalchar{+}}12]{PhysRevLett.109.100502}
Fabian Furrer, Torsten Franz, Mario Berta, Anthony Leverrier, Volkher~B.
  Scholz, Marco Tomamichel, and Reinhard~F. Werner.
\newblock Continuous variable quantum key distribution: Finite-key analysis of
  composable security against coherent attacks.
\newblock {\em Physical Review Letters}, 109(10):100502, September 2012.
\newblock arXiv:1112.2179.

\bibitem[Fur14]{PhysRevA.90.042325}
Fabian Furrer.
\newblock Reverse-reconciliation continuous-variable quantum key distribution
  based on the uncertainty principle.
\newblock {\em Physical Review A}, 90(4):042325, October 2014.
\newblock arXiv:1405.5965.

\bibitem[GG02a]{GG02}
Fr\'ed\'eric Grosshans and Philippe Grangier.
\newblock Continuous variable quantum cryptography using coherent states.
\newblock {\em Physical Review Letters}, 88(5):057902, January 2002.
\newblock arXiv:quant-ph/0109084.

\bibitem[GG02b]{Grangier2002}
Fr\'ed\'eric Grosshans and Philippe Grangier.
\newblock Reverse reconciliation protocols for quantum cryptography with
  continuous variables.
\newblock {\em Proceedings of the 6th International Conference on Quantum
  Communications, Measurement, and Computing}, 2002.
\newblock arXiv:quant-ph/0204127.

\bibitem[GGDL19]{GGDL19}
Shouvik Ghorai, Philippe Grangier, Eleni Diamanti, and Anthony Leverrier.
\newblock Asymptotic security of continuous-variable quantum key distribution
  with a discrete modulation.
\newblock {\em Physical Review X}, 9(2):021059, June 2019.
\newblock arXiv:1902.01317.

\bibitem[GK04]{GK04}
Christopher Gerry and Peter Knight.
\newblock {\em Introductory Quantum Optics}.
\newblock Cambridge University Press, November 2004.

\bibitem[GLN05]{GLN04}
Alexei Gilchrist, Nathan~K. Langford, and Michael~A. Nielsen.
\newblock Distance measures to compare real and ideal quantum processes.
\newblock {\em Physical Review A}, 71(6):062310, June 2005.
\newblock arXiv:quant-ph/0408063.

\bibitem[GLS16]{GLS16}
Marco~G. Genoni, Ludovico Lami, and Alessio Serafini.
\newblock Conditional and unconditional gaussian quantum dynamics.
\newblock {\em Contemporary Physics}, 57(3):331--349, January 2016.
\newblock arXiv:1607.02619.

\bibitem[GPC06]{PC2006}
Ra\'ul Garc\'{\i}a-Patr\'on and Nicolas~J. Cerf.
\newblock Unconditional optimality of {G}aussian attacks against
  continuous-variable quantum key distribution.
\newblock {\em Physical Review Letters}, 97(19):190503, November 2006.
\newblock arXiv:quant-ph/0608032.

\bibitem[Guh04]{Guha04}
Saikat Guha.
\newblock Classical capacity of the free-space quantum-optical channel.
\newblock Master's thesis, Massachusetts Institute of Technology, January 2004.

\bibitem[Hay06]{Hay06}
Masahito Hayashi.
\newblock Practical evaluation of security for quantum key distribution.
\newblock {\em Physical Review A}, 74(2):022307, August 2006.
\newblock arXiv:quant-ph/0602113.

\bibitem[Hay09]{Hay09}
Masahito Hayashi.
\newblock Information spectrum approach to second-order coding rate in channel
  coding.
\newblock {\em IEEE Transactions on Information Theory}, 55(11):4947--4966,
  November 2009.
\newblock arXiv:0801.2242.

\bibitem[{Hay}19]{Hay19}
Masahito {Hayashi}.
\newblock Universal channel coding for general output alphabet.
\newblock {\em IEEE Transactions on Information Theory}, 65(1):302--321,
  January 2019.
\newblock arXiv:1502.02218.

\bibitem[Hol73]{Holevo73}
Alexander~S. Holevo.
\newblock Bounds for the quantity of information transmitted by a quantum
  communication channel.
\newblock {\em Problems of Information Transmission}, 9:177--183, 1973.

\bibitem[Hol10]{H10}
Alexander~S. Holevo.
\newblock The entropy gain of infinite-dimensional quantum channels.
\newblock {\em Doklady Mathematics}, 82(2):730--731, October 2010.
\newblock arXiv:1003.5765.

\bibitem[Hol12]{H13book}
Alexander~S. Holevo.
\newblock {\em Quantum systems, channels, information: a mathematical
  introduction}, volume~16.
\newblock Walter de Gruyter, 2012.

\bibitem[HT12]{Hayashi_2012}
Masahito Hayashi and Toyohiro Tsurumaru.
\newblock Concise and tight security analysis of the {Bennett-Brassard} 1984
  protocol with finite key lengths.
\newblock {\em New Journal of Physics}, 14(9):093014, September 2012.
\newblock arXiv:1107.0589.

\bibitem[HZ12]{HZ12}
Teiko Heinosaari and M\'ario Ziman.
\newblock {\em The Mathematical Language of Quantum Theory: From Uncertainty to
  Entanglement}.
\newblock Cambridge University Press, 2012.

\bibitem[Iss18]{Isserlis18}
Leon Isserlis.
\newblock On a formula for the product-moment coefficient of any order of a
  normal frequency distribution in any number of variables.
\newblock {\em Biometrika}, 12(1/2):134--139, 1918.

\bibitem[JN12]{JN12}
Rahul {Jain} and Ashwin {Nayak}.
\newblock Short proofs of the quantum substate theorem.
\newblock {\em IEEE Transactions on Information Theory}, 58(6):3664--3669, June
  2012.
\newblock arXiv:1103.6067.

\bibitem[KGW19]{KGW19}
Eneet Kaur, Saikat Guha, and Mark~M. Wilde.
\newblock Asymptotic security of discrete-modulation protocols for
  continuous-variable quantum key distribution.
\newblock January 2019.
\newblock arXiv:1901.10099.

\bibitem[Kha16]{Kh16}
Sumeet Khatri.
\newblock Symmetric extendability of quantum states and the extreme limits of
  quantum key distribution.
\newblock Master's thesis, University of Waterloo, 2016.
\newblock \url{https://uwspace.uwaterloo.ca/handle/10012/10993}.

\bibitem[Koa06]{Koashi_2006}
Masato Koashi.
\newblock Unconditional security of quantum key distribution and the
  uncertainty principle.
\newblock {\em Journal of Physics: Conference Series}, 36:98--102, April 2006.

\bibitem[KR08]{KR08}
Oliver Kern and Joseph~M. Renes.
\newblock Improved one-way rates for {BB84} and 6-state protocols.
\newblock {\em Quantum Information and Computation}, 8(8):756--772, September
  2008.
\newblock arXiv:0712.1494.

\bibitem[KRBM07]{KRBM07}
Robert {K\"onig}, Renato Renner, Andor Bariska, and Ueli Maurer.
\newblock Small accessible quantum information does not imply security.
\newblock {\em Physical Review Letters}, 98(14):140502, April 2007.
\newblock arXiv:quant-ph/0512021.

\bibitem[Kru06]{K06}
Ole Krueger.
\newblock {\em Quantum Information Theory with {Gaussian} Systems}.
\newblock PhD thesis, Technische Universit\"at Braunschweig, April 2006.
\newblock Available at
  \url{https://publikationsserver.tu-braunschweig.de/receive/dbbs_mods_00020741}.

\bibitem[KS10]{KS10}
V.~Yu. Korolev and Irina~G. Shevtsova.
\newblock On the upper bound for the absolute constant in the {Berry-Esseen}
  inequality.
\newblock {\em Theory of Probability \& Its Applications}, 54(4):638--658,
  November 2010.

\bibitem[KW17]{KW17a}
Eneet Kaur and Mark~M. Wilde.
\newblock Upper bounds on secret key agreement over lossy thermal bosonic
  channels.
\newblock {\em Physical Review A}, 96(6):062318, December 2017.
\newblock arXiv:1706.04590.

\bibitem[LC99]{LC99}
Hoi-Kwong Lo and Hoi~Fung Chau.
\newblock Unconditional security of quantum key distribution over arbitrarily
  long distances.
\newblock {\em Science}, 283(5410):2050--2056, March 1999.
\newblock arXiv:quant-ph/9803006.

\bibitem[LCT14]{LCT14}
Hoi-Kwong Lo, Marcos Curty, and Kiyoshi Tamaki.
\newblock Secure quantum key distribution.
\newblock {\em Nature Photonics}, 8:595--604, July 2014.
\newblock arXiv:1505.05303.

\bibitem[Lev15]{Lev15}
Anthony Leverrier.
\newblock Composable security proof for continuous-variable quantum key
  distribution with coherent states.
\newblock {\em Physical Review Letters}, 114(7):070501, February 2015.
\newblock arXiv:1408.5689.

\bibitem[Lev17]{PhysRevLett.118.200501}
Anthony Leverrier.
\newblock Security of continuous-variable quantum key distribution via a
  {Gaussian} de {Finetti} reduction.
\newblock {\em Physical Review Letters}, 118(20):200501, May 2017.
\newblock arXiv:1701.03393.

\bibitem[LGG10]{LGG10}
Anthony Leverrier, Fr\'ed\'eric Grosshans, and Philippe Grangier.
\newblock Finite-size analysis of a continuous-variable quantum key
  distribution.
\newblock {\em Physical Review A}, 81(6):062343, June 2010.
\newblock arXiv:1005.0339.

\bibitem[Li14]{li12}
Ke~Li.
\newblock Second order asymptotics for quantum hypothesis testing.
\newblock {\em Annals of Statistics}, 42(1):171--189, February 2014.
\newblock arXiv:1208.1400.

\bibitem[Lin73]{Lindblad1973}
G{\"o}ran Lindblad.
\newblock Entropy, information and quantum measurements.
\newblock {\em Communications in Mathematical Physics}, 33(4):305--322,
  December 1973.

\bibitem[LKPS09]{Liang2009}
Yingbin Liang, Gerhard Kramer, H.~Vincent Poor, and Shlomo~Shamai (Shitz).
\newblock Compound wiretap channels.
\newblock {\em EURASIP Journal on Wireless Communications and Networking},
  2009(1):142374, 2009.

\bibitem[Lo01]{L01}
Hoi-Kwong Lo.
\newblock Proof of unconditional security of six-state quantum key distribution
  scheme.
\newblock {\em Quantum Information and Computation}, 1(2):81--94, August 2001.
\newblock arXiv:quant-ph/0102138.

\bibitem[LUL19]{LUL19}
Jie Lin, Twesh Upadhyaya, and Norbert L\"utkenhaus.
\newblock Asymptotic security analysis of discrete-modulated
  continuous-variable quantum key distribution.
\newblock May 2019.
\newblock arXiv:1905.10896.

\bibitem[L{\"u}t14]{Lut14}
Norbert L{\"u}tkenhaus.
\newblock {\em Quantum Information and Coherence}, chapter Quantum Key
  Distribution, pages 107--146.
\newblock Springer International Publishing, 2014.

\bibitem[LW19]{LW19}
Zi-Wen Liu and Andreas Winter.
\newblock Resource theories of quantum channels and the universal role of
  resource erasure.
\newblock April 2019.
\newblock arXiv:1904.04201v1.

\bibitem[Mau93]{M93}
Ueli~M. Maurer.
\newblock Secret key agreement by public discussion from common information.
\newblock {\em IEEE Transactions on Information Theory}, 39(3):733--742, May
  1993.

\bibitem[May01]{M01}
Dominic Mayers.
\newblock Unconditional security in quantum cryptography.
\newblock {\em Journal of the ACM}, 48(3):351--406, May 2001.
\newblock arXiv:quant-ph/9802025.

\bibitem[MW00]{MW00}
Ueli Maurer and Stefan Wolf.
\newblock Information-theoretic key agreement: From weak to strong secrecy for
  free.
\newblock In Bart Preneel, editor, {\em Advances in Cryptology --- EUROCRYPT
  2000}, pages 351--368, Berlin, Heidelberg, 2000. Springer Berlin Heidelberg.

\bibitem[Myh10]{Myhr10}
Geir~Ove Myhr.
\newblock {\em Symmetric extension of bipartite quantum states and its use in
  quantum key distribution with two-way postprocessing}.
\newblock PhD thesis, Friedrich-Alexander-Universit\"at Erlangen-N\"urnberg,
  November 2010.
\newblock arXiv:1103.0766.

\bibitem[NGA06]{NGA06}
Miguel Navascu\'es, Fr\'ed\'eric Grosshans, and Antonio Ac\'{\i}n.
\newblock Optimality of {G}aussian attacks in continuous-variable quantum
  cryptography.
\newblock {\em Physical Review Letters}, 97(19):190502, November 2006.
\newblock arXiv:quant-ph/0608034.

\bibitem[Nota]{Note1}
In some QKD protocols, the entire joint distribution of Alice's and Bob's
  classical data $p_{XY}(x,y)$ need not be estimated in order to distill a
  secure key. A well known example is CV-QKD, where only two scalar parameters
  need to be estimated in order to derive a key-rate lower bound under a
  collective-attack assumption, even in a finite key-length regime \cite
  {Lev15}. It is possible that, for such protocols, the second-order correction
  to the key distillation rate that we present in this paper could be extended
  to a formulation that does not require the estimation of $p_{XY}(x,y)$.

\bibitem[Notb]{Note2}
When estimating $p_{Y|X}$, we can bin the outcomes of the joint distribution
  $p_{XY}$ to get a discrete distribution $p_{X_m Y_m}$. Then we can estimate
  this finite distribution $p_{X_m Y_m}$ instead, which is reasonable since
  there are a finite number of parameters, and plug the values into the
  classical mutual information $I(X_m;Y_m)$ and classical mutual information
  variance $V(X_m;Y_m)$ for a second-order coding rate for information
  reconciliation. Then the limits $I(X_m;Y_m) \to I(X;Y)$ and $V(X_m;Y_m)\to
  V(X;Y)$ hold as the bin size gets smaller.

\bibitem[OMW19]{OMW19}
Samad~Khabbazi Oskouei, Stefano Mancini, and Mark~M. Wilde.
\newblock Union bound for quantum information processing.
\newblock {\em Proceedings of the Royal Society A}, 475(2221):20180612, January
  2019.
\newblock arXiv:1804.08144.

\bibitem[PCZ97]{PCZ97}
J.~F. Poyatos, J.~Ignacio Cirac, and Peter Zoller.
\newblock Complete characterization of a quantum process: The two-bit quantum
  gate.
\newblock {\em Physical Review Letters}, 78(2):390--393, January 1997.
\newblock arXiv:quant-ph/9611013.

\bibitem[{Pir}17]{PLOB15}
{Pirandola {\textit{et al.}}}
\newblock Fundamental limits of repeaterless quantum communications.
\newblock {\em Nat. Comm.}, 8(15043), 2017.

\bibitem[{Pol}13]{Poly13}
Yury {Polyanskiy}.
\newblock On dispersion of compound {DMC}s.
\newblock In {\em 2013 51st Annual Allerton Conference on Communication,
  Control, and Computing (Allerton)}, pages 26--32, October 2013.

\bibitem[PPV10]{polyanskiy10}
Yury Polyanskiy, H.~Vincent Poor, and Sergio Verd\'{u}.
\newblock Channel coding rate in the finite blocklength regime.
\newblock {\em IEEE Transactions on Information Theory}, 56(5):2307--2359, May
  2010.

\bibitem[PR14]{PR14}
Christopher Portmann and Renato Renner.
\newblock Cryptographic security of quantum key distribution.
\newblock 2014.
\newblock arXiv:1409.3525.

\bibitem[QSW18]{Qi_2018}
Haoyu Qi, Kunal Sharma, and Mark~M. Wilde.
\newblock Entanglement-assisted private communication over quantum broadcast
  channels.
\newblock {\em Journal of Physics A: Mathematical and Theoretical},
  51(37):374001, August 2018.
\newblock arXiv:1803.03976.

\bibitem[Ras02]{R02}
Alexey~E. Rastegin.
\newblock Relative error of state-dependent cloning.
\newblock {\em Physical Review A}, 66(4):042304, October 2002.

\bibitem[Ras03]{R03}
Alexey~E. Rastegin.
\newblock A lower bound on the relative error of mixed-state cloning and
  related operations.
\newblock {\em Journal of Optics B: Quantum and Semiclassical Optics},
  5(6):S647, December 2003.
\newblock arXiv:quant-ph/0208159.

\bibitem[Ras06]{R06}
Alexey~E. Rastegin.
\newblock Sine distance for quantum states.
\newblock February 2006.
\newblock arXiv:quant-ph/0602112.

\bibitem[RC09]{RC09}
Renato Renner and J.~I. Cirac.
\newblock de {Finetti} representation theorem for infinite-dimensional quantum
  systems and applications to quantum cryptography.
\newblock {\em Physical Review Letters}, 102(11):110504, March 2009.
\newblock arXiv:0809.2243.

\bibitem[Ren05]{Renner2005}
Renato Renner.
\newblock {\em Security of Quantum Key Distribution}.
\newblock PhD thesis, ETH Zurich, September 2005.
\newblock arXiv:quant-ph/0512258.

\bibitem[RGK05]{RGK05}
Renato Renner, Nicolas Gisin, and Barbara Kraus.
\newblock Information-theoretic security proof for quantum-key-distribution
  protocols.
\newblock {\em Physical Review A}, 72(1):012332, July 2005.
\newblock arXiv:quant-ph/0502064.

\bibitem[RR12]{RR12}
Joseph~M. {Renes} and Renato {Renner}.
\newblock One-shot classical data compression with quantum side information and
  the distillation of common randomness or secret keys.
\newblock {\em IEEE Transactions on Information Theory}, 58(3):1985--1991,
  March 2012.
\newblock arXiv:1008.0452.

\bibitem[RS78]{RS78}
M.~Reed and B.~Simon.
\newblock {\em Methods of Modern Mathematical Physics}, volume I: Functional
  Analysis.
\newblock Academic Press, New York, 1978.

\bibitem[Sav12]{Savov12}
Ivan Savov.
\newblock {\em Network information theory for classical-quantum channels}.
\newblock PhD thesis, McGill University, School of Computer Science, July 2012.

\bibitem[SBP15]{SBP15}
Rafael~F. {Schaefer}, Holger {Boche}, and H.~Vincent Poor.
\newblock Secure communication under channel uncertainty and adversarial
  attacks.
\newblock {\em Proceedings of the IEEE}, 103(10):1796--1813, October 2015.

\bibitem[SBPC{\etalchar{+}}09]{SBCDLP09}
Valerio Scarani, Helle Bechmann-Pasquinucci, Nicolas~J. Cerf, Miloslav
  Du\ifmmode~\check{s}\else \v{s}\fi{}ek, Norbert L\"utkenhaus, and Momtchil
  Peev.
\newblock The security of practical quantum key distribution.
\newblock {\em Reviews of Modern Physics}, 81(3):1301--1350, September 2009.
\newblock arXiv:0802.4155.

\bibitem[Sen18]{Sen18}
Pranab Sen.
\newblock A one-shot quantum joint typicality lemma.
\newblock June 2018.
\newblock arXiv:1806.07278v2.

\bibitem[Ser17]{Ser17}
Alessio Serafini.
\newblock {\em Quantum Continuous Variables: A Primer of Theoretical Methods}.
\newblock CRC Press, Taylor \& Francis Group, 2017.

\bibitem[Sha48]{bell1948shannon}
Claude~E. Shannon.
\newblock A mathematical theory of communication.
\newblock {\em Bell System Technical Journal}, 27:379--423, 1948.

\bibitem[She11]{S11}
Irina Shevtsova.
\newblock On the absolute constants in the {Berry-Esseen} type inequalities for
  identically distributed summands.
\newblock November 2011.
\newblock arXiv:1111.6554.

\bibitem[SP00]{SP00}
Peter~W. Shor and John Preskill.
\newblock Simple proof of security of the {BB84} quantum key distribution
  protocol.
\newblock {\em Physical Review Letters}, 85(2):441--444, July 2000.
\newblock arXiv:quant-ph/0003004.

\bibitem[SRS08]{smith:170502}
Graeme Smith, Joseph~M. Renes, and John~A. Smolin.
\newblock Structured codes improve the {Bennett-Brassard}-84 quantum key rate.
\newblock {\em Physical Review Letters}, 100(17):170502, April 2008.
\newblock arXiv:quant-ph/0607018.

\bibitem[SS08]{SS08}
Graeme Smith and John~A. Smolin.
\newblock Additive extensions of a quantum channel.
\newblock In {\em 2008 IEEE Information Theory Workshop}, pages 368--372, May
  2008.
\newblock arXiv:0712.2471.

\bibitem[Sti55]{S55}
William~F. Stinespring.
\newblock Positive functions on {C*}-algebras.
\newblock {\em Proceedings of the American Mathematical Society}, 6:211--216,
  1955.

\bibitem[Str62]{S62}
Volker Strassen.
\newblock Asymptotische absch\"atzungen in {Shannons} informationstheorie.
\newblock In {\em Transactions of the Third Prague Conference on Information
  Theory}, pages 689--723, Prague, 1962.
\newblock \url{http://www.math.cornell.edu/~pmlut/strassen.pdf}.

\bibitem[TBR16]{TBR15}
Marco Tomamichel, Mario Berta, and Joseph~M. Renes.
\newblock Quantum coding with finite resources.
\newblock {\em Nature Communications}, 7:11419, May 2016.
\newblock arXiv:1504.04617.

\bibitem[TBS17]{TBS17}
Nima {Tavangaran}, Holger {Boche}, and Rafael~F. {Schaefer}.
\newblock Secret-key generation using compound sources and one-way public
  communication.
\newblock {\em IEEE Transactions on Information Forensics and Security},
  12(1):227--241, January 2017.
\newblock arXiv:1601.07513.

\bibitem[TH13]{TH12}
Marco Tomamichel and Masahito Hayashi.
\newblock A hierarchy of information quantities for finite block length
  analysis of quantum tasks.
\newblock {\em IEEE Transactions on Information Theory}, 59(11):7693--7710,
  November 2013.
\newblock arXiv:1208.1478.

\bibitem[TL17]{Tomamichel2017largelyself}
Marco Tomamichel and Anthony Leverrier.
\newblock A largely self-contained and complete security proof for quantum key
  distribution.
\newblock {\em {Quantum}}, 1:14, July 2017.
\newblock arXiv:1506.08458.

\bibitem[TLGR12]{TLGR12}
Marco Tomamichel, Charles Ci~Wen Lim, Nicolas Gisin, and Renato Renner.
\newblock Tight finite-key analysis for quantum cryptography.
\newblock {\em Nature Communications}, 3:634, January 2012.
\newblock arXiv:1103.4130.

\bibitem[TMMPE17]{Tomamichel2017}
Marco Tomamichel, Jesus Martinez-Mateo, Christoph Pacher, and David Elkouss.
\newblock Fundamental finite key limits for one-way information reconciliation
  in quantum key distribution.
\newblock {\em Quantum Information Processing}, 16(11):280, October 2017.
\newblock arXiv:1401.5194.

\bibitem[Tom15]{T15book}
Marco Tomamichel.
\newblock {\em Quantum Information Processing with Finite Resources:
  Mathematical Foundations}, volume~5.
\newblock Springer, 2015.
\newblock arXiv:1504.00233.

\bibitem[TT15a]{TT15}
Vincent Y.~F. {Tan} and Marco {Tomamichel}.
\newblock The third-order term in the normal approximation for the awgn
  channel.
\newblock {\em IEEE Transactions on Information Theory}, 61(5):2430--2438, May
  2015.
\newblock arXiv:1311.2337.

\bibitem[TT15b]{TT13}
Marco Tomamichel and Vincent Y.~F. Tan.
\newblock Second-order asymptotics for the classical capacity of image-additive
  quantum channels.
\newblock {\em Communications in Mathematical Physics}, 338(1):103--137, August
  2015.
\newblock arXiv:1308.6503.

\bibitem[Uhl76]{Uhl76}
Armin Uhlmann.
\newblock The ``transition probability'' in the state space of a *-algebra.
\newblock {\em Reports on Mathematical Physics}, 9(2):273--279, April 1976.

\bibitem[Ume62]{U62}
Hisaharu Umegaki.
\newblock Conditional expectations in an operator algebra {IV} (entropy and
  information).
\newblock {\em Kodai Mathematical Seminar Reports}, 14(2):59--85, 1962.

\bibitem[WGC06]{Wolf06}
Michael~M. Wolf, Geza Giedke, and J.~Ignacio Cirac.
\newblock Extremality of {G}aussian quantum states.
\newblock {\em Physical Review Letters}, 96(8):080502, March 2006.
\newblock arXiv:quant-ph/0509154.

\bibitem[Wil36]{W36}
John Williamson.
\newblock On the algebraic problem concerning the normal forms of linear
  dynamical systems.
\newblock {\em American Journal of Mathematics}, 58(1):141--163, January 1936.

\bibitem[Wil17a]{wilde2017position}
Mark~M. Wilde.
\newblock Position-based coding and convex splitting for private communication
  over quantum channels.
\newblock {\em Quantum Information Processing}, 16(10):264, October 2017.
\newblock arXiv:1703.01733.

\bibitem[Wil17b]{Wil17}
Mark~M. Wilde.
\newblock {\em Quantum Information Theory}.
\newblock Cambridge University Press, second edition, February 2017.
\newblock arXiv:1106.1445.

\bibitem[WLC18]{Winick2018reliablenumerical}
Adam Winick, Norbert L{\"{u}}tkenhaus, and Patrick~J. Coles.
\newblock Reliable numerical key rates for quantum key distribution.
\newblock {\em {Quantum}}, 2:77, July 2018.
\newblock arXiv:1710.05511.

\bibitem[Wol59]{Wolfowitz1959}
J.~Wolfowitz.
\newblock Simultaneous channels.
\newblock {\em Archive for Rational Mechanics and Analysis}, 4(1):371--386,
  January 1959.

\bibitem[WR12]{WR12}
Ligong Wang and Renato Renner.
\newblock One-shot classical-quantum capacity and hypothesis testing.
\newblock {\em Physical Review Letters}, 108(20):200501, 2012.
\newblock arXiv:1007.5456.

\bibitem[WRG16]{WRG16}
Mark~M. Wilde, Joseph~M. Renes, and Saikat Guha.
\newblock Second-order coding rates for pure-loss bosonic channels.
\newblock {\em Quantum Information Processing}, 15(3):1289--1308, March 2016.
\newblock arXiv:1408.5328.

\bibitem[WTB17]{WTB16}
Mark~M. Wilde, Marco Tomamichel, and Mario Berta.
\newblock Converse bounds for private communication over quantum channels.
\newblock {\em IEEE Transactions on Information Theory}, 63(3):1792--1817,
  March 2017.
\newblock arXiv:1602.08898.

\bibitem[WTLB17]{BLTW16}
Mark~M. Wilde, Marco Tomamichel, Seth Lloyd, and Mario Berta.
\newblock Gaussian hypothesis testing and quantum illumination.
\newblock {\em Physical Review Letters}, 119(12):120501, September 2017.
\newblock arXiv:1608.06991.

\bibitem[WW18]{WW18}
Xin Wang and Mark~M. Wilde.
\newblock Exact entanglement cost of quantum states and channels under
  {PPT}-preserving operations.
\newblock September 2018.
\newblock arXiv:1809.09592v1.

\bibitem[WW19]{WW19states}
Xin Wang and Mark~M. Wilde.
\newblock Resource theory of asymmetric distinguishability.
\newblock May 2019.
\newblock arXiv:1905.11629.

\bibitem[Wyn75]{W75}
Aaron~D. Wyner.
\newblock The wire-tap channel.
\newblock {\em Bell System Technical Journal}, 54(8):1355--1387, October 1975.

\bibitem[XMZ{\etalchar{+}}19]{XMZLP19}
Feihu Xu, Xiongfeng Ma, Qiang Zhang, Hoi-Kwong Lo, and Jian-Wei Pan.
\newblock Quantum cryptography with realistic devices.
\newblock March 2019.
\newblock arXiv:1903.09051v1.

\bibitem[ZHRL09]{Zhao2009}
Yi-Bo Zhao, Matthias Heid, Johannes Rigas, and Norbert L\"utkenhaus.
\newblock Asymptotic security of binary modulated continuous-variable quantum
  key distribution under collective attacks.
\newblock {\em Physical Review A}, 79(1):012307, January 2009.
\newblock arXiv:0807.3751.

\end{thebibliography}

\appendix

\section{Mutual and Holevo informations and variances for six-state and BB84
protocols}

\label{app:mut-hol-info-6-bb84}

When Alice and Bob discard their basis information in the six-state or BB84
protocols, then their classical data can be described via the following
quantum channel:
\begin{equation}
\mathcal{N}_{A\rightarrow B}(\rho_{A})=p_{1}\rho_{A}+p_{2}Z\rho_{A}%
Z+p_{3}X\rho_{A}X+p_{4}Y\rho Y, \label{eq:pauli-channel-app}%
\end{equation}
where $p_{1}, p_{2}, p_{3}, p_{4} \geq0$, $p_{1}+p_{2}+p_{3}+p_{4}=1$, and
there are further constraints on $p_{1}, p_{2}, p_{3}, p_{4}$.

For the six-state protocol,
\begin{align}
p_{1} &  =1-\frac{3Q}{2},\label{eq-app:six-state-pauli-1}\\
p_{2} &  =p_{3}=p_{4}=\frac{Q}{2},\label{eq-app:six-state-pauli-2}\\
Q &  =\frac{1}{3}(Q_{x}+Q_{y}+Q_{z}).\label{eq-app:six-state-pauli-3}%
\end{align}
For these values of $p_{1},p_{2},p_{3},p_{4}$, we obtain
\begin{align}
p_{XY}(0,0) &  \coloneqq\frac{1}{2}\operatorname{Tr}[|0\rangle\!\langle
0|_{B}\mathcal{N}_{A\rightarrow B}(|0\rangle\!\langle0|_{A})]\\
&  =\frac{1}{2}(1-Q),\\
p_{XY}(0,1) &  \coloneqq\frac{1}{2}\operatorname{Tr}[|1\rangle\!\langle
1|_{B}\mathcal{N}_{A\rightarrow B}(|0\rangle\!\langle0|_{A})]\\
&  =\frac{1}{2}Q,\\
p_{XY}(1,0) &  \coloneqq\frac{1}{2}\operatorname{Tr}[|0\rangle\!\langle
0|_{B}\mathcal{N}_{A\rightarrow B}(|1\rangle\!\langle1|_{A})]\\
&  =\frac{1}{2}Q,\\
p_{XY}(1,1) &  \coloneqq\frac{1}{2}\operatorname{Tr}[|1\rangle\!\langle
1|_{B}\mathcal{N}_{A\rightarrow B}(|1\rangle\!\langle1|_{A})]\\
&  =\frac{1}{2}(1-Q),
\end{align}
which is consistent with the probability distribution $p_{X_{2}Y_{2}%
}^{\text{6-state}|\text{sift}}$ in \eqref{eq-6state_pXY_discard_1}--\eqref{eq-6state_pXY_discard_4}.

For the BB84 protocol,
\begin{align}
p_{1} &  =1-2Q+s,\label{eq-app:BB84-pauli-1}\\
p_{2} &  =Q-s,\\
p_{3} &  =Q-s,\\
p_{4} &  =s,\\
Q &  =\frac{1}{2}(Q_{x}+Q_{z}),\\
s &  =Q-\frac{Q_{y}}{2},\label{eq-app:BB84-pauli-last}%
\end{align}
where $s\in\lbrack0,Q]$ is a parameter to be optimized. Indeed, for these
values of $p_{1},p_{2},p_{3},p_{4}$,
\begin{align}
p_{XY}(0,0) &  \coloneqq\frac{1}{2}\operatorname{Tr}[|0\rangle\!\langle
0|_{B}\mathcal{N}_{A\rightarrow B}(|0\rangle\!\langle0|_{A})]\\
&  =\frac{1}{2}(1-Q),\\
p_{XY}(0,1) &  \coloneqq\frac{1}{2}\operatorname{Tr}[|1\rangle\!\langle
1|_{B}\mathcal{N}_{A\rightarrow B}(|0\rangle\!\langle0|_{A})]\\
&  =\frac{1}{2}Q,\\
p_{XY}(1,0) &  \coloneqq\frac{1}{2}\operatorname{Tr}[|0\rangle\!\langle
0|_{B}\mathcal{N}_{A\rightarrow B}(|1\rangle\!\langle1|_{A})]\\
&  =\frac{1}{2}Q,\\
p_{XY}(1,1) &  \coloneqq\frac{1}{2}\operatorname{Tr}[|1\rangle\!\langle
1|_{B}\mathcal{N}_{A\rightarrow B}(|1\rangle\!\langle1|_{A})]\\
&  =\frac{1}{2}(1-Q),
\end{align}
which is consistent with the probability distribution $p_{X_{2}Y_{2}%
}^{\text{BB84}|\text{sift}}$ in \eqref{eq-BB84_pXY_discard_1}--\eqref{eq-BB84_pXY_discard_4}.

In general,
\begin{align}
p_{XY}(0,0) &  \coloneqq\frac{1}{2}\operatorname{Tr}[|0\rangle\!\langle
0|_{B}\mathcal{N}_{A\rightarrow B}(|0\rangle\!\langle0|_{A})]\\
&  =\frac{1}{2}(p_{1}+p_{2}),\\
p_{XY}(0,1) &  \coloneqq\frac{1}{2}\operatorname{Tr}[|1\rangle\!\langle
1|_{B}\mathcal{N}_{A\rightarrow B}(|0\rangle\!\langle0|_{A})]\\
&  =\frac{1}{2}(p_{3}+p_{4}),\\
p_{XY}(1,0) &  \coloneqq\frac{1}{2}\operatorname{Tr}[|0\rangle\!\langle
0|_{B}\mathcal{N}_{A\rightarrow B}(|1\rangle\!\langle1|_{A})]\\
&  =\frac{1}{2}(p_{3}+p_{4}),\\
p_{XY}(1,1) &  \coloneqq\frac{1}{2}\operatorname{Tr}[|1\rangle\!\langle
1|_{B}\mathcal{N}_{A\rightarrow B}(|1\rangle\!\langle1|_{A})]\\
&  =\frac{1}{2}(p_{1}+p_{2}).
\end{align}
Then,
\begin{multline}
\rho_{XY}=\frac{1}{2}(p_{1}+p_{2})|0,0\rangle\!\langle0,0|+\frac{1}{2}%
(p_{3}+p_{4})|0,1\rangle\!\langle0,1|\\
+\frac{1}{2}(p_{3}+p_{4})|1,0\rangle\!\langle1,0|\\
+\frac{1}{2}(p_{1}+p_{2})|1,1\rangle\!\langle1,1|.
\end{multline}
It is then straightforward to show that
\begin{align}
I(X;Y)_{\rho} &  =1+(p_{1}+p_{2})\log_{2}(p_{1}+p_{2})\nonumber\\
&  \qquad+(p_{3}+p_{4})\log_{2}(p_{3}+p_{4})\\
&  =1-h_{2}(p_{1}+p_{2}),
\end{align}
where $h_{2}$ is the binary entropy function.

Let us now calculate $V(X;Y)_{\rho}$. We have%Recall that
\begin{equation}
V(\rho\Vert\sigma)\coloneqq\operatorname{Tr}[\rho(\log_{2}\rho-\log_{2}%
\sigma)^{2}]-D(\rho\Vert\sigma)^{2},
\end{equation}
and%
\begin{align}
V(X;Y)_{\rho}  &  \coloneqq V(\rho_{XY}\Vert\rho_{X}\otimes\rho_{Y})\\
&  =\operatorname{Tr}[\rho_{XY}(\log_{2}\rho_{XY}-\log_{2}(\rho_{X}\otimes
\rho_{Y}))^{2}]\nonumber\\
&  \qquad-I(X;Y)_{\rho}^{2}.
\end{align}
One can show that
\begin{multline}
\operatorname{Tr}[\rho_{XY}(\log_{2}\rho_{XY}-\log_{2}(\rho_{X}\otimes\rho
_{Y}))^{2}]\\
=(p_{1}+p_{2})\left(  1+\log_{2}(p_{1}+p_{2})\right)  ^{2}\\
+(p_{3}+p_{4})\left(  1+\log_{2}(p_{3}+p_{4})\right)  ^{2},
\end{multline}
from which it follows that
\begin{equation}
V(X;Y)_{\rho}=(p_{1}+p_{2})(p_{3}+p_{4})\left(  \log_{2}\!\left(  \frac
{p_{1}+p_{2}}{p_{3}+p_{4}}\right)  \right)  ^{2}.
\end{equation}

We now calculate $I(X;E)$ and $V(X;E)$. We consider the following purification
of the Choi state of the channel in \eqref{eq:pauli-channel-app}:
\begin{multline}
|\psi\rangle_{ABE}\coloneqq\sqrt{p_{1}}|\Phi^{+}\rangle_{AB}|0,0\rangle
_{E}+\sqrt{p_{2}}|\Phi^{-}\rangle_{AB}|1,1\rangle_{E}\\
+\sqrt{p_{3}}|\Psi^{+}\rangle_{AB}|0,1\rangle_{E}+\sqrt{p_{4}}|\Psi^{-}%
\rangle_{AB}|1,0\rangle_{E}.
\end{multline}
Now,
\begin{align}
&  (\langle0|_{A}\otimes I_{BE})|\psi\rangle_{ABE}\nonumber\\
&  =|0\rangle_{B}\frac{1}{\sqrt{2}}(\sqrt{p_{1}}|0,0\rangle_{E}+\sqrt{p_{2}%
}|1,1\rangle_{E})\nonumber\\
&  \quad+|1\rangle_{B}\frac{1}{\sqrt{2}}(\sqrt{p_{3}}|0,1\rangle_{E}%
+\sqrt{p_{4}}|1,0\rangle_{E}),\\
&  (\langle1|_{A}\otimes I_{BE})|\psi\rangle_{ABE}\nonumber\\
&  =|1\rangle_{B}\frac{1}{\sqrt{2}}(\sqrt{p_{1}}|0,0\rangle_{E}-\sqrt{p_{2}%
}|1,1\rangle_{E})\nonumber\\
&  \quad+|0\rangle_{B}\frac{1}{\sqrt{2}}(\sqrt{p_{3}}|0,1\rangle_{E}%
-\sqrt{p_{4}}|1,0\rangle_{E}).
\end{align}
Using this, we obtain
\begin{widetext}
	\begin{align}
		\rho_{XE}&=\frac{1}{2}\ket{0}\bra{0}_X\otimes \rho_{E|X=0}+\frac{1}{2}\ket{1}\bra{1}_X\otimes\rho_{E|X=1}\\
		&=\begin{pmatrix} \frac{p_1}{2} & 0 & 0 & \frac{\sqrt{p_1p_2}}{2} & 0 & 0 & 0 & 0 \\
						0 & \frac{p_3}{2} & \frac{\sqrt{p_3p_4}}{2} & 0 & 0 & 0 & 0 & 0 \\
						0 & \frac{\sqrt{p_3p_4}}{2} & \frac{p_4}{2} & 0 & 0 & 0 & 0 & 0 \\
						\frac{\sqrt{p_1p_2}}{2} & 0 & 0 & \frac{p_2}{2} & 0 & 0 & 0 & 0 \\
						0 & 0 & 0 & 0 & \frac{p_1}{2} & 0 & 0 & -\frac{\sqrt{p_1p_2}}{2} \\
						0 & 0 & 0 & 0 & 0 & \frac{p_3}{2} & -\frac{\sqrt{p_3p_4}}{2} & 0 \\
						0 & 0 & 0 & 0 & 0 & -\frac{\sqrt{p_3p_4}}{2} & \frac{p_4}{2} & 0 \\
						0 & 0 & 0 & 0 & -\frac{\sqrt{p_1p_2}}{2} & 0 & 0 & \frac{p_2}{2} \end{pmatrix}.
	\end{align}
	\end{widetext}The state $\rho_{XE}$ can be diagonalized as follows:
\begin{multline}
\rho_{XE}=\left(  \frac{p_{1}+p_{2}}{2}\right)  |0\rangle\!\langle0|_{X}%
\otimes|\Phi_{p_{1},p_{2}}^{+}\rangle\!\langle\Phi_{p_{1},p_{2}}^{+}|\\
+\left(  \frac{p_{1}+p_{2}}{2}\right)  |1\rangle\!\langle1|_{X}\otimes
|\Phi_{p_{1},p_{2}}^{-}\rangle\!\langle\Phi_{p_{1},p_{2}}^{-}|\\
+\left(  \frac{p_{3}+p_{4}}{2}\right)  |0\rangle\!\langle0|_{X}\otimes
|\Psi_{p_{3},p_{4}}^{+}\rangle\!\langle\Psi_{p_{3},p_{4}}^{+}|\\
+\left(  \frac{p_{3}+p_{4}}{2}\right)  |1\rangle\!\langle1|_{X}\otimes
|\Psi_{p_{3},p_{4}}^{-}\rangle\!\langle\Psi_{p_{3},p_{4}}^{-}|,
\end{multline}
where
\begin{align}
|\Phi_{p_{1},p_{2}}^{\pm}\rangle &  \coloneqq\frac{1}{\sqrt{\frac{p_{1}+p_{2}%
}{2}}}\left(  \sqrt{\frac{p_{1}}{2}}|0,0\rangle\pm\sqrt{\frac{p_{2}}{2}%
}|1,1\rangle\right)  ,\\
|\Psi_{p_{3},p_{4}}^{\pm}\rangle &  \coloneqq\frac{1}{\sqrt{\frac{p_{3}+p_{4}%
}{2}}}\left(  \sqrt{\frac{p_{3}}{2}}|0,1\rangle\pm\sqrt{\frac{p_{4}}{2}%
}|1,0\rangle\right)  .
\end{align}
Then, by observing that%
\begin{multline}
\rho_{E}\coloneqq\operatorname{Tr}_{X}[\rho_{XE}]=p_{1}|0,0\rangle
\langle0,0|+p_{3}|0,1\rangle\!\langle0,1|\\
+p_{4}|1,0\rangle\!\langle1,0|+p_{2}|1,1\rangle\!\langle1,1|,
\end{multline}
we find that
\begin{equation}
I(X;E)=H(\vec{p})-h_{2}(p_{1}+p_{2}),
\end{equation}
where
\begin{equation}
H(\vec{p})\coloneqq\sum_{i=1}^{4}-p_{i}\log_{2}p_{i}.
\end{equation}

Finally, for $V(X;E)$, it is straightforward to show that
\begin{multline}
\operatorname{Tr}[\rho_{XE}(\log_{2}\rho_{XE}-\log_{2}(\rho_{X}\otimes\rho
_{E}))^{2}]\\
=p_{1}\left(  \log_{2}\!\left(  \frac{p_{1}}{p_{1}+p_{2}}\right)  \right)
^{2}+p_{2}\left(  \log_{2}\!\left(  \frac{p_{2}}{p_{1}+p_{2}}\right)  \right)
^{2}\\
+p_{3}\left(  \log_{2}\!\left(  \frac{p_{3}}{p_{3}+p_{4}}\right)  \right)
^{2}+p_{4}\left(  \log_{2}\!\left(  \frac{p_{4}}{p_{3}+p_{4}}\right)  \right)
^{2},
\end{multline}
from which it follows that
\begin{multline}
V(X;E)=p_{1}\left(  \log_{2}\!\left(  \frac{p_{1}}{p_{1}+p_{2}}\right)
\right)  ^{2}\\
+p_{2}\left(  \log_{2}\!\left(  \frac{p_{2}}{p_{1}+p_{2}}\right)  \right)
^{2}+p_{3}\left(  \log_{2}\!\left(  \frac{p_{3}}{p_{3}+p_{4}}\right)  \right)
^{2}\\
+p_{4}\left(  \log_{2}\!\left(  \frac{p_{4}}{p_{3}+p_{4}}\right)  \right)
^{2}-(H(\vec{p})-h_{2}(p_{1}+p_{2}))^{2}.
\end{multline}

For the six-state protocol, by using the formulas above and
\eqref{eq-app:six-state-pauli-1}--\eqref{eq-app:six-state-pauli-3}, we obtain
\begin{align}
I(X;Y)_{\rho} &  =1-h_{2}(Q),\\
V(X;Y)_{\rho} &  =Q(1-Q)\left(  \log_{2}\!\left(  \frac{1-Q}{Q}\right)
\right)  ^{2},\\
I(X;E)_{\rho} &  =-\left(  1-\frac{3Q}{2}\right)  \log_{2}\!\left(
1-\frac{3Q}{2}\right)  \nonumber\\
&  \qquad-\frac{3Q}{2}\log_{2}\!\left(  \frac{Q}{2}\right)  -h_{2}(Q),\\
V(X;E)_{\rho} &  =Q+\left(  1-\frac{3Q}{2}\right)  \left(  \log_{2}\!\left(
\frac{1-\frac{3Q}{2}}{1-Q}\right)  \right)  ^{2}\nonumber\\
&  \qquad+\frac{Q}{2}\left(  \log_{2}\!\left(  \frac{\frac{Q}{2}}{1-Q}\right)
\right)  ^{2}-I(X;E)_{\rho}^{2}.
\end{align}
%The key rate is then
%\begin{multline}
%K_{\text{6-state}}(Q,n,\varepsilon_{\text{I}},\varepsilon_{\text{II}%
%})=1+\left(  1-\frac{3Q}{2}\right)  \log_{2}\!\left(  1-\frac{3Q}{2}\right)
%\\
%+\frac{3Q}{2}\log_{2}\!\left(  \frac{Q}{2}\right)  \\
%\quad+\sqrt{\frac{Q(1-Q)}{n}}\log_{2}\!\left(  \frac{1-Q}{Q}\right)  \Phi
%^{-1}(\varepsilon_{\text{I}})\\
%+\sqrt{\frac{V(X;E)_{\rho}}{n}}\Phi^{-1}(\varepsilon_{\text{II}}).
%\end{multline}

For the BB84 protocol, by using the formulas above and
\eqref{eq-app:BB84-pauli-1}--\eqref{eq-app:BB84-pauli-last}, we obtain
\begin{align}
I(X;Y)_{\rho} &  =1-h_{2}(Q),\\
V(X;Y)_{\rho} &  =Q(1-Q)\left(  \log_{2}\!\left(  \frac{1-Q}{Q}\right)
\right)  ^{2},\\
I(X;E)_{\rho} &  =H(\{1-2Q+s,Q-s,Q-s,s\})\nonumber\\
&  \qquad-h_{2}(Q),\\
V(X;E)_{\rho} &  =(1-2Q+s)\left(  \log_{2}\!\left(  \frac{1-2Q+s}{1-Q}\right)
\right)  ^{2}\nonumber\\
&  \qquad+(Q-s)\left(  \log_{2}\!\left(  \frac{Q-s}{1-Q}\right)  \right)
^{2}\nonumber\\
&  \qquad+(Q-s)\left(  \log_{2}\!\left(  \frac{Q-s}{Q}\right)  \right)
^{2}\nonumber\\
&  \qquad+s\left(  \log_{2}\!\left(  \frac{s}{Q}\right)  \right)
^{2}-I(X;E)_{\rho}^{2}.
\end{align}

\section{Relation between smooth max-mutual informations}

\label{app:relation-smooth-max-MIs}

Here we prove the following lemma relating two different variants of smooth
max-mutual information:

\begin{lemma}
\label{lem:smooth-max-MI-relations}Let $\rho_{AE}$ be a bipartite state acting
on a separable Hilbert space $\mathcal{H}_{A}\otimes\mathcal{H}_{E}$, and let
$\varepsilon,\delta>0$ be such that $\varepsilon+\delta<1$. Then%
\begin{equation}
\widetilde{I}_{\max}^{\varepsilon+\delta}(E;A)_{\rho}\leq D_{\max
}^{\varepsilon}(\rho_{AE}\Vert\rho_{A}\otimes\rho_{E})+\log_{2}\!\left(
\frac{8}{\delta^{2}}\right)  ,
\end{equation}
where $\widetilde{I}_{\max}^{\varepsilon+\delta}(E;A)_{\rho}$ is defined in
\eqref{eq:smooth-alt-imax} and $D_{\max}^{\varepsilon}\!\left(  \rho_{AE}%
\Vert\rho_{A}\otimes\rho_{E}\right)  $ in \eqref{eq:smooth-dmax-MI}.
\end{lemma}

\begin{proof}
The proof follows the proof of \cite[Theorem~2]{ABJT18}, but has some
differences because it is not necessary in our case to ensure partial
smoothing. Let $\widetilde{\rho}_{AE}$ be a state satisfying $P(\widetilde
{\rho}_{AE},\rho_{AE})\leq\varepsilon$. Let $\gamma=\delta^{2}/8$, and set
$\Pi_{A}^{\gamma}$ to be the projection onto the positive eigenspace of
$\frac{1}{\gamma}\widetilde{\rho}_{A}-\rho_{A}$. Then it follows that%
\begin{multline}
\Pi_{A}^{\gamma}\left(  \frac{1}{\gamma}\widetilde{\rho}_{A}-\rho_{A}\right)
\Pi_{A}^{\gamma}\geq0\\
\Rightarrow\Pi_{A}^{\gamma}\rho_{A}\Pi_{A}^{\gamma}\leq\frac{1}{\gamma}\Pi
_{A}^{\gamma}\widetilde{\rho}_{A}\Pi_{A}^{\gamma}=\frac{8}{\delta^{2}}\Pi
_{A}^{\gamma}\widetilde{\rho}_{A}\Pi_{A}^{\gamma}%
,\label{eq:op-ineq-rho-tilde-rho}%
\end{multline}
and%
\begin{multline}
\left(  I-\Pi_{A}^{\gamma}\right)  \left(  \frac{1}{\gamma}\widetilde{\rho
}_{A}-\rho_{A}\right)  \left(  I-\Pi_{A}^{\gamma}\right)  \leq0\\
\Rightarrow\operatorname{Tr}[\left(  I-\Pi_{A}^{\gamma}\right)  \widetilde
{\rho}_{A}]\leq\gamma\operatorname{Tr}[\left(  I-\Pi_{A}^{\gamma}\right)
\rho_{A}]\leq\gamma=\frac{\delta^{2}}{8},\label{eq:trace-ineq-delta-gamma}%
\end{multline}
where the last inequality follows because $\operatorname{Tr}[\left(  I-\Pi
_{A}^{\gamma}\right)  \rho_{A}]\leq1$. The inequality in
\eqref{eq:trace-ineq-delta-gamma} can be rewritten as%
\begin{equation}
\operatorname{Tr}[\Pi_{A}^{\gamma}\widetilde{\rho}_{A}]\geq1-\frac{\delta^{2}%
}{8}. \label{eq:gamma-project-onto-smooth-rho}%
\end{equation}
Let us define the following states:%
\begin{multline}
\overline{\rho}_{AEX}\coloneqq \Pi_{A}^{\gamma}\widetilde{\rho}_{AE}\Pi_{A}^{\gamma
}\otimes|0\rangle\!\langle0|_{X}\\
+\left(  I-\Pi_{A}^{\gamma}\right)  \widetilde{\rho}_{AE}\left(  I-\Pi
_{A}^{\gamma}\right)  \otimes|1\rangle\!\langle1|_{X},
\end{multline}%
\begin{multline}
\widehat{\rho}_{AEX}\coloneqq \\
\left(  \Pi_{A}^{\gamma}\widetilde{\rho}_{AE}\Pi_{A}^{\gamma}+\widetilde{\rho
}_{A}^{1/2}\left(  I-\Pi_{A}^{\gamma}\right)  \widetilde{\rho}_{A}%
^{1/2}\otimes\rho_{E}\right)  \otimes|0\rangle\!\langle0|_{X},
\end{multline}
so that%
\begin{align}
\widehat{\rho}_{AE}  &  =\operatorname{Tr}_{X}[\widehat{\rho}_{AEX}]\\
&  =\Pi_{A}^{\gamma}\widetilde{\rho}_{AE}\Pi_{A}^{\gamma}+\widetilde{\rho}%
_{A}^{1/2}\left(  I-\Pi_{A}^{\gamma}\right)  \widetilde{\rho}_{A}^{1/2}%
\otimes\rho_{E}.
\end{align}
Then, using the inequality $\widetilde{\rho}_{AE}\leq\mu\rho_{A}\otimes
\rho_{E}$, with%
\begin{equation}
\mu\coloneqq 2^{D_{\max}(\widetilde{\rho}_{AE}\Vert\rho_{A}\otimes\rho_{E})},
\end{equation}
and the fact that $\mu\frac{8}{\delta^{2}}\geq1$ (which holds because
$D_{\max}(\widetilde{\rho}_{AE}\Vert\rho_{A}\otimes\rho_{E})\geq0$ and
$8\geq\delta^{2}$), we find that%
\begin{align}
\widehat{\rho}_{AE}  &  \leq\mu\Pi_{A}^{\gamma}\rho_{A}\Pi_{A}^{\gamma}%
\otimes\rho_{E}+\widetilde{\rho}_{A}^{1/2}\left(  I-\Pi_{A}^{\gamma}\right)
\widetilde{\rho}_{A}^{1/2}\otimes\rho_{E}\nonumber\\
&  \leq\mu\frac{8}{\delta^{2}}\Pi_{A}^{\gamma}\widetilde{\rho}_{A}\Pi
_{A}^{\gamma}\otimes\rho_{E}+\widetilde{\rho}_{A}^{1/2}\left(  I-\Pi
_{A}^{\gamma}\right)  \widetilde{\rho}_{A}^{1/2}\otimes\rho_{E}\nonumber\\
&  \leq\mu\frac{8}{\delta^{2}}\left[  \Pi_{A}^{\gamma}\widetilde{\rho}_{A}%
\Pi_{A}^{\gamma}\otimes\rho_{E}+\widetilde{\rho}_{A}^{1/2}\left(  I-\Pi
_{A}^{\gamma}\right)  \widetilde{\rho}_{A}^{1/2}\otimes\rho_{E}\right]
\nonumber\\
&  =\mu\frac{8}{\delta^{2}}\left[  \Pi_{A}^{\gamma}\widetilde{\rho}_{A}\Pi
_{A}^{\gamma}+\widetilde{\rho}_{A}^{1/2}\left(  I-\Pi_{A}^{\gamma}\right)
\widetilde{\rho}_{A}^{1/2}\right]  \otimes\rho_{E}\nonumber\\
&  =\mu\frac{8}{\delta^{2}}\widehat{\rho}_{A}\otimes\rho_{E}.
\end{align}
The second inequality above follows from \eqref{eq:op-ineq-rho-tilde-rho}.
Applying the definition of $D_{\max}(\widehat{\rho}_{AE}\Vert\widehat{\rho
}_{A}\otimes\rho_{E})$, we conclude that%
\begin{multline}
D_{\max}(\widehat{\rho}_{AE}\Vert\widehat{\rho}_{A}\otimes\rho_{E}%
)\leq\label{eq:dmax-inequality-smooth-max-relate}\\
D_{\max}(\widetilde{\rho}_{AE}\Vert\rho_{A}\otimes\rho_{E})+\log_{2}\!\left(
\frac{8}{\delta^{2}}\right)  .
\end{multline}

We can conclude the statement of the lemma if $P(\widehat{\rho}_{AE},\rho
_{AE})\leq\varepsilon+\delta$, and so it is our aim to show this now. Consider
that%
\begin{equation}
P(\widehat{\rho}_{AEX},\overline{\rho}_{AEX})=\sqrt{1-F(\widehat{\rho}%
_{AEX},\overline{\rho}_{AEX})}.
\end{equation}
The following chain of inequalities holds%
\begin{align}
&  \sqrt{F(\widehat{\rho}_{AEX},\overline{\rho}_{AEX})}\nonumber\\
&  =\operatorname{Tr}\left[  \left(  \sqrt{\Pi_{A}^{\gamma}\widetilde{\rho
}_{AE}\Pi_{A}^{\gamma}}\widehat{\rho}_{AE}\sqrt{\Pi_{A}^{\gamma}%
\widetilde{\rho}_{AE}\Pi_{A}^{\gamma}}\right)  ^{1/2}\right] \nonumber\\
&  \geq\operatorname{Tr}\left[  \left(  \sqrt{\Pi_{A}^{\gamma}\widetilde{\rho
}_{AE}\Pi_{A}^{\gamma}}\left(  \Pi_{A}^{\gamma}\widetilde{\rho}_{AE}\Pi
_{A}^{\gamma}\right)  \sqrt{\Pi_{A}^{\gamma}\widetilde{\rho}_{AE}\Pi
_{A}^{\gamma}}\right)  ^{1/2}\right] \nonumber\\
&  =\operatorname{Tr}\left[  \Pi_{A}^{\gamma}\widetilde{\rho}_{AE}\Pi
_{A}^{\gamma}\right] \nonumber\\
&  =\operatorname{Tr}[\Pi_{A}^{\gamma}\widetilde{\rho}_{A}]\nonumber\\
&  \geq1-\frac{\delta^{2}}{8},
\end{align}
where the inequality follows from operator monotonicity of the square root and
the fact that%
\begin{align}
\widehat{\rho}_{AE}  &  =\Pi_{A}^{\gamma}\widetilde{\rho}_{AE}\Pi_{A}^{\gamma
}+\widetilde{\rho}_{A}^{1/2}\left(  I-\Pi_{A}^{\gamma}\right)  \widetilde
{\rho}_{A}^{1/2}\otimes\rho_{E}\\
&  \geq\Pi_{A}^{\gamma}\widetilde{\rho}_{AE}\Pi_{A}^{\gamma}%
\end{align}
From the above and \eqref{eq:gamma-project-onto-smooth-rho}, we conclude that
$F(\widehat{\rho}_{AEX},\overline{\rho}_{AEX})\geq1-\frac{\delta^{2}}{4}$,
which implies that%
\begin{equation}
P(\widehat{\rho}_{AEX},\overline{\rho}_{AEX})\leq\frac{\delta}{2}.
\label{eq:purified-bar-hat-delta-2}%
\end{equation}
Now consider that%
\begin{align}
&  P(\overline{\rho}_{AEX},\rho_{AE}\otimes|0\rangle\!\langle0|_{X})\nonumber\\
&  \leq P(\overline{\rho}_{AEX},\widetilde{\rho}_{AE}\otimes|0\rangle
\langle0|_{X})\nonumber\\
&  \qquad+P(\widetilde{\rho}_{AE}\otimes|0\rangle\!\langle0|_{X},\rho
_{AE}\otimes|0\rangle\!\langle0|_{X})\nonumber\\
&  =\sqrt{1-F(\overline{\rho}_{AEX},\widetilde{\rho}_{AE}\otimes
|0\rangle\!\langle0|_{X})}+P(\widetilde{\rho}_{AE},\rho_{AE})\nonumber\\
&  =\sqrt{1-\left\Vert \sqrt{\Pi_{A}^{\gamma}\widetilde{\rho}_{AE}\Pi
_{A}^{\gamma}}\sqrt{\widetilde{\rho}_{AE}}\right\Vert _{1}^{2}}+P(\widetilde
{\rho}_{AE},\rho_{AE})\nonumber\\
&  =\sqrt{1-\left\Vert \sqrt{\Pi_{A}^{\gamma}\widetilde{\rho}_{AE}\Pi
_{A}^{\gamma}}\sqrt{\Pi_{A}^{\gamma}\widetilde{\rho}_{AE}\Pi_{A}^{\gamma}%
}\right\Vert _{1}^{2}}+P(\widetilde{\rho}_{AE},\rho_{AE})\nonumber\\
&  =\sqrt{1-\left(  \operatorname{Tr}[\Pi_{A}^{\gamma}\widetilde{\rho}%
_{AE}]\right)  ^{2}}+P(\widetilde{\rho}_{AE},\rho_{AE})\nonumber\\
&  \leq\frac{\delta}{2}+\varepsilon,
\end{align}
where we applied the triangle inequality of the sine distance
\cite{R02,R03,R06,GLN04} for the first inequality and the fact that
$\left\Vert \sqrt{\Pi\omega\Pi}\sqrt{\tau}\right\Vert _{1} = \left\Vert
\sqrt{\Pi\omega\Pi}\sqrt{\Pi\tau\Pi}\right\Vert _{1}$ for a projector $\Pi$
and states $\omega$ and $\tau$. Combining this with
\eqref{eq:purified-bar-hat-delta-2}, we find that%
\begin{align}
&  P(\widehat{\rho}_{AE},\rho_{AE})\nonumber\\
&  =P(\widehat{\rho}_{AEX},\rho_{AE}\otimes|0\rangle\!\langle0|_{X})\\
&  \leq P(\widehat{\rho}_{AEX},\overline{\rho}_{AEX})+P(\overline{\rho}%
_{AEX},\rho_{AE}\otimes|0\rangle\!\langle0|_{X})\\
&  =\varepsilon+\delta.
\end{align}
Since we have found a state $\widehat{\rho}_{AE}$ satisfying $P(\widehat{\rho
}_{AE},\rho_{AE})\leq\varepsilon+\delta$ and
\eqref{eq:dmax-inequality-smooth-max-relate}, we conclude that%
\begin{equation}
\widetilde{I}_{\max}^{\varepsilon+\delta}(E;A)_{\rho}\leq D_{\max}%
(\widetilde{\rho}_{AE}\Vert\rho_{A}\otimes\rho_{E})+\log_{2}\!\left(  \frac
{8}{\delta^{2}}\right)  .
\end{equation}
Since this inequality has been shown for all states $\widetilde{\rho}_{AE}$
satisfying $P(\widetilde{\rho}_{AE},\rho_{AE})\leq\varepsilon$, we conclude
the statement of the lemma.
\end{proof}

\section{Smooth universal convex split lemma for states acting on a separable
Hilbert space}

\label{sec:smooth-univ-convex-split}

Here we prove a smooth universal convex split lemma for states acting on a
separable Hilbert space:

\begin{lemma}
[Smooth universal convex split]Let $\mathcal{S}$ be a set, and let $\rho
_{AE}^{s}$ be a state acting on a separable Hilbert space $\mathcal{H}%
_{A}\otimes\mathcal{H}_{E}$, such that $\operatorname{Tr}_{E}[\rho_{AE}%
^{s}]=\rho_{A}$ for all $s\in\mathcal{S}$. Let $\tau_{A_{1}\cdots A_{R}E}^{s}$
denote the following state:%
\begin{multline}
\tau_{A_{1}\cdots A_{R}E}^{s}\coloneqq \\
\frac{1}{R}\sum_{r=1}^{R}\rho_{A_{1}}\otimes\cdots\otimes\rho_{A_{r-1}}%
\otimes\rho_{A_{r}E}^{s}\otimes\rho_{A_{r+1}}\otimes\cdots\otimes\rho_{A_{R}}.
\end{multline}
Let $\varepsilon\in(0,1)$ and $\eta\in(0,\sqrt{\varepsilon})$. If%
\begin{equation}
\log_{2}R\geq\sup_{s\in\mathcal{S}}\widetilde{I}_{\max}^{\sqrt{\varepsilon
}-\eta}(E;A)_{\rho^{s}}+2\log_{2}\!\left(  \frac{1}{2\eta}\right)  ,
\end{equation}
then for all $s\in\mathcal{S}$, there exists a state $\widetilde{\rho}_{E}%
^{s}$ satisfying%
\begin{equation}
\frac{1}{2}\left\Vert \tau_{A_{1}\cdots A_{R}E}^{s}-\rho_{A_{1}}\otimes
\cdots\otimes\rho_{A_{R}}\otimes\widetilde{\rho}_{E}^{s}\right\Vert _{1}%
\leq\sqrt{\varepsilon},
\end{equation}
and $P(\widetilde{\rho}_{E}^{s},\rho_{E}^{s})\leq\sqrt{\varepsilon}-\eta$.
\end{lemma}

\begin{proof}
Fix $s\in\mathcal{S}$. Let $\widetilde{\rho}_{AE}^{s}$ be an arbitrary state
satisfying $P(\widetilde{\rho}_{AE}^{s},\rho_{AE}^{s})\leq\sqrt{\varepsilon
}-\eta$ and is such that%
\begin{equation}
\rho_{A}^{s}\otimes\widetilde{\rho}_{E}^{s}=p\widetilde{\rho}_{AE}^{s}+\left(
1-p\right)  \omega_{AE}^{s}, \label{eq:dmax-condition-cs}%
\end{equation}
for some $p\in(0,1)$ and $\omega_{AE}^{s}$ some state. We define the following
state, which we think of as an approximation to $\tau_{A_{1}\cdots A_{R}E}%
^{s}$:%
\begin{multline}
\widetilde{\tau}_{A_{1}\cdots A_{R}E}^{s}\coloneqq \\
\frac{1}{R}\sum_{r=1}^{R}\rho_{A_{1}}\otimes\cdots\otimes\rho_{A_{r-1}}%
\otimes\widetilde{\rho}_{A_{r}E}^{s}\otimes\rho_{A_{r+1}}\otimes\cdots
\otimes\rho_{A_{R}}.
\end{multline}
It is a good approximation if $\sqrt{\varepsilon}-\eta$ is small, because%
\begin{align}
&  \!\!\!\!\sqrt{F}(\tau_{A_{1}\cdots A_{R}E}^{s},\widetilde{\tau}%
_{A_{1}\cdots A_{R}E}^{s})\nonumber\\
&  \geq\frac{1}{R}\sum_{r=1}^{R}\sqrt{F}(\rho_{A}^{\otimes r-1}\otimes
\rho_{A_{r}E}^{s}\otimes\rho_{A}^{\otimes R-r} ,\nonumber\\
&  \quad\qquad\qquad\qquad\rho_{A}^{\otimes r-1}\otimes\widetilde{\rho}%
_{A_{r}E}^{s}\otimes\rho_{A}^{\otimes R-r})\\
&  =\frac{1}{R}\sum_{r=1}^{R}\sqrt{F}(\rho_{A_{r}E}^{s},\widetilde{\rho
}_{A_{r}E}^{s})\\
&  =\sqrt{F}(\rho_{A_{r}E}^{s},\widetilde{\rho}_{A_{r}E}^{s}),
\end{align}
which in turn implies that%
\begin{equation}
\sqrt{F}(\tau_{A_{1}\cdots A_{R}E}^{s},\widetilde{\tau}_{A_{1}\cdots A_{R}%
E}^{s})\geq\sqrt{F}(\rho_{A_{r}E}^{s},\widetilde{\rho}_{A_{r}E}^{s}).
\label{eq:cs-state-approx}%
\end{equation}
So the inequality in \eqref{eq:cs-state-approx}, the definition of the sine
distance, and the fact that $P(\widetilde{\rho}_{AE}^{s},\rho_{AE}^{s}%
)\leq\sqrt{\varepsilon}-\eta$, imply that%
\begin{equation}
P(\tau_{A_{1}\cdots A_{R}E}^{s},\widetilde{\tau}_{A_{1}\cdots A_{R}E}^{s}%
)\leq\sqrt{\varepsilon}-\eta,
\end{equation}
and in turn that%
\begin{equation}
\frac{1}{2}\left\Vert \tau_{A_{1}\cdots A_{R}E}^{s}-\widetilde{\tau}%
_{A_{1}\cdots A_{R}E}^{s}\right\Vert _{1}\leq\sqrt{\varepsilon}-\eta.
\label{eq:smoothed-state-cs}%
\end{equation}
Now, following \cite[Appendix~A]{LW19} closely, let us define the following
states:%
\begin{align}
\beta_{AE}^{s}  &  \coloneqq \rho_{A}^{s}\otimes\widetilde{\rho}_{E}^{s},\\
\alpha_{AE}^{s}  &  \coloneqq \widetilde{\rho}_{AE}^{s},
\end{align}%
\begin{multline}
\widetilde{\tau}_{A^{R}E^{R}}^{s}\coloneqq \frac{1}{R}\sum_{r=1}^{R}\beta_{A_{1}E_{1}%
}\otimes\cdots\otimes\beta_{A_{r-1}E_{r-1}}\otimes\alpha_{A_{r}E_{r}}%
^{s}\otimes\\
\beta_{A_{r+1}E_{r+1}}^{s}\otimes\cdots\otimes\beta_{A_{R}E_{R}},
\end{multline}
and observe that%
\begin{align}
\operatorname{Tr}_{E_{2}^{R}}[(\beta_{AE}^{s})^{\otimes R}]  &  =\rho_{A_{1}%
}\otimes\cdots\otimes\rho_{A_{R}}\otimes\widetilde{\rho}_{E}^{s},\\
\operatorname{Tr}_{E_{2}^{R}}[\widetilde{\tau}_{A^{R}E^{R}}^{s}]  &
=\widetilde{\tau}_{A_{1}\cdots A_{R}E}^{s}.
\end{align}
Thus, it follows from data processing of normalized trace distance that%
\begin{multline}
\frac{1}{2}\left\Vert \widetilde{\tau}_{A_{1}\cdots A_{R}E}^{s}-\rho_{A_{1}%
}\otimes\cdots\otimes\rho_{A_{R}}\otimes\widetilde{\rho}_{E}^{s}\right\Vert
_{1}\label{eq:c-split-partial-trace-mono}\\
\leq\frac{1}{2}\left\Vert \widetilde{\tau}_{A^{R}E^{R}}^{s}-(\beta_{AE}%
^{s})^{\otimes R}\right\Vert _{1}.
\end{multline}
Now applying \cite[Lemma~15]{LW19}, we find that%
\begin{equation}
\frac{1}{2}\left\Vert \widetilde{\tau}_{A_{1}\cdots A_{R}E}^{s}-\rho_{A_{1}%
}\otimes\cdots\otimes\rho_{A_{R}}\otimes\widetilde{\rho}_{E}^{s}\right\Vert
_{1}\leq\eta
\end{equation}
if%
\begin{equation}
\log_{2}R\geq\log_{2}(1/p)+2\log_{2}\!\left(  \frac{1}{2\eta}\right)  .
\end{equation}
For the same choice of $R$, it follows from
\eqref{eq:c-split-partial-trace-mono}\ that%
\begin{equation}
\frac{1}{2}\left\Vert \widetilde{\tau}_{A_{1}\cdots A_{R}E}^{s}-\rho_{A_{1}%
}\otimes\cdots\otimes\rho_{A_{R}}\otimes\widetilde{\rho}_{E}^{s}\right\Vert
_{1}\leq\eta. \label{eq:cs-result}%
\end{equation}
Applying the triangle inequality to \eqref{eq:smoothed-state-cs} and
\eqref{eq:cs-result}, we find that%
\begin{equation}
\frac{1}{2}\left\Vert \tau_{A_{1}\cdots A_{R}E}^{s}-\rho_{A_{1}}\otimes
\cdots\otimes\rho_{A_{R}}\otimes\widetilde{\rho}_{E}^{s}\right\Vert _{1}%
\leq\sqrt{\varepsilon}.
\end{equation}

The whole argument above holds for an arbitrary state $\widetilde{\rho}%
_{AE}^{s}$ satisfying $P(\widetilde{\rho}_{AE}^{s},\rho_{AE}^{s})\leq
\sqrt{\varepsilon}-\eta$ and \eqref{eq:dmax-condition-cs}, and so taking an
infimum of $\log_{2}(1/p)$ over all states satisfying these conditions, and
applying the definition in \eqref{eq:smooth-alt-imax}, as well as
\eqref{eq:alt-dmax}, we find that%
\begin{equation}
\frac{1}{2}\left\Vert \tau_{A_{1}\cdots A_{R}E}^{s}-\rho_{A_{1}}\otimes
\cdots\otimes\rho_{A_{R}}\otimes\widetilde{\rho}_{E}^{s}\right\Vert _{1}%
\leq\sqrt{\varepsilon}%
\end{equation}
if%
\begin{equation}
\log_{2}R\geq\widetilde{I}_{\max}^{\sqrt{\varepsilon}-\eta}(E;A)_{\rho^{s}%
}+2\log_{2}\!\left(  \frac{1}{2\eta}\right)  .
\end{equation}
Now note that the whole argument holds for all $s\in\mathcal{S}$, due to the
uniform structure of the convex split method, which consists of bringing in a
tensor-power state and performing a random permutation. We then find that%
\begin{equation}
\frac{1}{2}\left\Vert \tau_{A_{1}\cdots A_{R}E}^{s}-\rho_{A_{1}}\otimes
\cdots\otimes\rho_{A_{R}}\otimes\widetilde{\rho}_{E}^{s}\right\Vert _{1}%
\leq\sqrt{\varepsilon}%
\end{equation}
for all $s\in\mathcal{S}$ if%
\begin{equation}
\log_{2}R\geq\sup_{s\in\mathcal{S}}\widetilde{I}_{\max}^{\sqrt{\varepsilon
}-\eta}(E;A)_{\rho^{s}}+2\log_{2}\!\left(  \frac{1}{2\eta}\right)  .
\end{equation}
This concludes the proof.
\end{proof}

\section{Duality of smooth max-relative entropy}

\label{app:duality-smooth-dmax}This appendix generalizes a duality relation
for smooth max-relative entropy from the finite-dimensional case \cite{JN12}
to the case of states acting on a separable Hilbert space. Recall that the
smooth max-relative entropy is defined as%
\begin{equation}
D_{\max}^{\varepsilon}(\rho\Vert\sigma)=\log_{2}\inf_{\widetilde{\rho
}:P(\widetilde{\rho},\rho)\leq\varepsilon}\inf_{\lambda\geq0}\left\{
\lambda:\widetilde{\rho}\leq\lambda\sigma\right\}  .
\end{equation}
We can also define the dual smooth max-relative entropy as%
\begin{equation}
\widehat{D}_{\max}^{\varepsilon}(\rho\Vert\sigma)\coloneqq \log_{2}\sup_{M\geq0}%
\inf_{\widetilde{\rho}:P(\widetilde{\rho},\rho)\leq\varepsilon}\left\{
\operatorname{Tr}[M\widetilde{\rho}]:\operatorname{Tr}[M\sigma]\leq1\right\}
,
\end{equation}
where the supremum is with respect to all compact bounded operators $M\geq0$,
as these are dual to the trace-class operators \cite{RS78}. In this section,
for both of the above quantities, we expand the definitions to allow for
subnormalized states $\rho$ and $\sigma$ (i.e., such that $\operatorname{Tr}%
[\rho],\operatorname{Tr}[\sigma]\leq1$) and the sine distance becomes as
follows \cite{T15book}:%
\begin{equation}
P(\tau,\omega)\coloneqq \sqrt{1-F(\tau\oplus\left(  1-\operatorname{Tr}[\tau]\right)
,\omega\oplus\left(  1-\operatorname{Tr}[\omega]\right)  )}.
\end{equation}
In what follows, we call subnormalized states \textquotedblleft
substates\textquotedblright\ for short. Note that for any projection $\Pi$ and
substates $\tau$ and $\omega$, the following inequality holds \cite{T15book}%
\begin{equation}
P(\Pi\tau\Pi,\Pi\omega\Pi)\leq P(\tau,\omega).
\end{equation}

\begin{lemma}
\label{lem:smooth-dmax-primal-dual-equal}Let $\rho$ be a state and $\sigma$ 
a trace-class positive semi-definite operator acting on a separable Hilbert space $\mathcal{H}$. Let $\varepsilon\in(0,1)$
and suppose that $D_{\max}^{\varepsilon}(\rho\Vert\sigma)<\infty$. Then%
\begin{equation}
D_{\max}^{\varepsilon}(\rho\Vert\sigma)=\widehat{D}_{\max}^{\varepsilon}%
(\rho\Vert\sigma).
\end{equation}

\end{lemma}

\begin{proof}
We prove this in several steps, with the proof bearing some similarities to
the approach from \cite[Appendix~B]{FAR11} (see also \cite[Appendix~A]{WW18}
in this context). First, for a state $\widetilde{\rho}$, consider
from weak duality that%
\begin{equation}
\inf_{\lambda\geq0}\left\{  \lambda:\widetilde{\rho}\leq\lambda\sigma\right\}
\geq\sup_{M\geq0}\left\{  \operatorname{Tr}[M\widetilde{\rho}%
]:\operatorname{Tr}[M\sigma]\leq1\right\}  , \label{eq:simple-weak-duality}%
\end{equation}
where the optimization on the right-hand side is over the compact bounded
operators, given that these are dual to the trace-class operators \cite{RS78}.
In more detail, let $\lambda\geq0$ satisfy $\widetilde{\rho}\leq\lambda\sigma
$, and let $M$ be an arbitrary compact bounded operator satisfying $M\geq0$
and $\operatorname{Tr}[M\sigma]\leq1$. Then it follows that $\lambda
\geq\lambda\operatorname{Tr}[M\sigma]\geq\operatorname{Tr}[M\widetilde{\rho}%
]$, where we used $\widetilde{\rho}\leq\lambda\sigma$ and $M \geq0$. Since
this inequality holds for all $\lambda$ and $M$ satisfying the given
conditions, we conclude the weak duality inequality in
\eqref{eq:simple-weak-duality}. Then we have that%
\begin{align}
&  D_{\max}^{\varepsilon}(\rho\Vert\sigma)\nonumber\\
&  =\log_{2}\inf_{\widetilde{\rho}:P(\widetilde{\rho},\rho)\leq\varepsilon
}\inf_{\lambda\geq0}\left\{  \lambda:\widetilde{\rho}\leq\lambda\sigma\right\}
\\
&  \geq\log_{2}\inf_{\widetilde{\rho}:P(\widetilde{\rho},\rho)\leq\varepsilon
}\sup_{M\geq0}\left\{  \operatorname{Tr}[M\widetilde{\rho}]:\operatorname{Tr}%
[M\sigma]\leq1\right\} \\
&  \geq\log_{2}\sup_{M\geq0}\inf_{\widetilde{\rho}:P(\widetilde{\rho}%
,\rho)\leq\varepsilon}\left\{  \operatorname{Tr}[M\widetilde{\rho
}]:\operatorname{Tr}[M\sigma]\leq1\right\} \\
&  =\widehat{D}_{\max}^{\varepsilon}(\rho\Vert\sigma).
\end{align}
So the following weak-duality inequality holds%
\begin{equation}
D_{\max}^{\varepsilon}(\rho\Vert\sigma)\geq\widehat{D}_{\max}^{\varepsilon
}(\rho\Vert\sigma). \label{eq:weak-duality-starting-point}%
\end{equation}

Let $\{\Pi^{k}\}_{k}$ be a sequence of projectors onto finite-dimensional
subspaces strongly converging to the identity operator, and suppose that
$\Pi^{k}\leq\Pi^{k^{\prime}}$ for $k\leq k^{\prime}$. Then define%
\begin{equation}
\rho^{k}\coloneqq \Pi^{k}\rho\Pi^{k},\qquad\sigma^{k}\coloneqq \Pi^{k}\sigma\Pi^{k}.
\end{equation}

We now prove that%
\begin{equation}
\lim_{k\rightarrow\infty}D_{\max}^{\varepsilon}(\rho^{k}\Vert\sigma
^{k})=D_{\max}^{\varepsilon}(\rho\Vert\sigma).
\label{eq:smooth-dmax-primal-limit}%
\end{equation}
Let $\widetilde{\rho}$ be an arbitrary state satisfying $P(\widetilde{\rho
},\rho)\leq\varepsilon$ and let $\lambda$ satisfy $\widetilde{\rho}\leq
\lambda\sigma$. Then it follows that%
\begin{align}
\widetilde{\rho}^{k}  &  \coloneqq \Pi^{k}\widetilde{\rho}\Pi^{k}\leq\lambda\Pi
^{k}\sigma\Pi^{k}=\lambda\sigma^{k},\\
P(\widetilde{\rho}^{k},\rho^{k})  &  \leq\varepsilon,
\end{align}
so that%
\begin{equation}
D_{\max}^{\varepsilon}(\rho^{k}\Vert\sigma^{k})\leq\log_{2}\lambda.
\end{equation}
Since $\lambda$ and $\widetilde{\rho}$ are arbitrary, we conclude that%
\begin{equation}
D_{\max}^{\varepsilon}(\rho^{k}\Vert\sigma^{k})\leq D_{\max}^{\varepsilon
}(\rho\Vert\sigma),
\end{equation}
and since this inequality holds for all $k$, the following inequality holds%
\begin{equation}
\limsup_{k\rightarrow\infty}D_{\max}^{\varepsilon}(\rho^{k}\Vert\sigma
^{k})\leq D_{\max}^{\varepsilon}(\rho\Vert\sigma).
\label{eq:up-bnd-seq-smooth-D-max}%
\end{equation}

To show the opposite inequality, we first prove that $D_{\max}^{\varepsilon
}(\rho^{k}\Vert\sigma^{k})$ is monotone non-decreasing with $k$. (The proof is
in fact similar to what we have just shown.) Fix $k$ and $k^{\prime}$ such
that $k\leq k^{\prime}$. Let $\widetilde{\rho}^{k^{\prime}}$ be a substate
satisfying $P(\widetilde{\rho}^{k^{\prime}},\rho^{k^{\prime}})\leq\varepsilon$
and let $\lambda^{k^{\prime}}$ be such that $\widetilde{\rho}^{k^{\prime}}%
\leq\lambda^{k^{\prime}}\sigma^{k^{\prime}}$. Then it follows that%
\begin{equation}
\widetilde{\rho}^{k^{\prime},k}\coloneqq \Pi^{k}\widetilde{\rho}^{k^{\prime}}\Pi
^{k}\leq\lambda^{k^{\prime}}\Pi^{k}\sigma^{k^{\prime}}\Pi^{k}=\lambda
^{k^{\prime}}\sigma^{k},
\end{equation}
and%
\begin{equation}
P(\widetilde{\rho}^{k^{\prime},k},\Pi^{k}\rho^{k^{\prime}}\Pi^{k}%
)=P(\widetilde{\rho}^{k^{\prime},k},\rho^{k})\leq\varepsilon,
\end{equation}
so that%
\begin{equation}
D_{\max}^{\varepsilon}(\rho^{k}\Vert\sigma^{k})\leq\log_{2}\lambda^{k^{\prime
}}.
\end{equation}
Since $\lambda^{k^{\prime}}$ and $\widetilde{\rho}^{k^{\prime}}$ are
arbitrary, it follows that%
\begin{equation}
D_{\max}^{\varepsilon}(\rho^{k}\Vert\sigma^{k})\leq D_{\max}^{\varepsilon
}(\rho^{k^{\prime}}\Vert\sigma^{k^{\prime}}).
\end{equation}

Now let $\widetilde{\rho}^{k}$ and $\lambda^{k}$ be the optimal state and
value for $D_{\max}^{\varepsilon}(\rho^{k}\Vert\sigma^{k})$, so that
$\lambda^{k}=D_{\max}^{\varepsilon}(\rho^{k}\Vert\sigma^{k})$. By what we have
just shown, the sequence $\lambda^{k}$ is monotone non-decreasing, and so the
limit $\lambda\coloneqq \lim_{k\rightarrow\infty}\lambda^{k}$ is well defined. Also,
we know from \eqref{eq:up-bnd-seq-smooth-D-max} that $\lambda\leq D_{\max
}^{\varepsilon}(\rho\Vert\sigma)<\infty$. Due to the fact that $\left\Vert
\widetilde{\rho}^{k}\right\Vert _{1}=\operatorname{Tr}[\widetilde{\rho}%
^{k}]\leq\lambda^{k}\operatorname{Tr}[\sigma^{k}]\leq\lambda\operatorname{Tr}%
[\sigma]$, it follows that $\widetilde{\rho}^{k}$ is a bounded sequence in the
trace-class operators. Since the trace-class operators form the dual space of
the compact bounded operators \cite{RS78}, we apply the Banach--Alaoglu
theorem \cite{RS78} to find a subsequence $\left\{  \widetilde{\rho}%
^{k}\right\}  _{k\in S}$ having a weak* limit $\widetilde{\rho}$, which means
that $\operatorname{Tr}[K\widetilde{\rho}^{k}]\rightarrow\operatorname{Tr}%
[K\widetilde{\rho}]$ for $k\in S$ and for all compact bounded operators $K$.
It holds that $\widetilde{\rho}^{k}$ is positive semi-definite. Then it
follows that $\lambda^{k}\sigma^{k}-\widetilde{\rho}^{k}$ converges in the
weak operator topology to $\lambda\sigma-\widetilde{\rho}$. Then we conclude
that $\lambda\sigma-\widetilde{\rho}\geq0$, and finally that $D_{\max
}^{\varepsilon}(\rho\Vert\sigma)\leq\lambda$. So we conclude \eqref{eq:smooth-dmax-primal-limit}.

By strong duality, and the fact that $\rho^{k}$ and $\sigma^{k}$ are
finite-dimensional, it follows from \cite{JN12} that%
\begin{equation}
D_{\max}^{\varepsilon}(\rho^{k}\Vert\sigma^{k})=\widehat{D}_{\max
}^{\varepsilon}(\rho^{k}\Vert\sigma^{k})
\end{equation}
for all $k$. So these equalities and \eqref{eq:smooth-dmax-primal-limit} imply
that%
\begin{equation}
\lim_{k\rightarrow\infty}\widehat{D}_{\max}^{\varepsilon}(\rho^{k}\Vert
\sigma^{k})=D_{\max}^{\varepsilon}(\rho\Vert\sigma).
\label{eq:strong-duality-conseq}%
\end{equation}

Now our goal is to prove that%
\begin{equation}
\lim_{k\rightarrow\infty}\widehat{D}_{\max}^{\varepsilon}(\rho^{k}\Vert
\sigma^{k})\leq\widehat{D}_{\max}^{\varepsilon}(\rho\Vert\sigma),
\label{eq:dual-smooth-max-lim}%
\end{equation}
and we adopt a similar approach as above for doing so. Let $M^{k}$ be an
arbitrary compact bounded operator satisfying $\operatorname{Tr}[M^{k}%
\sigma^{k}]\leq1$ and let $\widetilde{\rho}$ be an arbitrary substate
satisfying $P(\widetilde{\rho},\rho)\leq\varepsilon$. Define $\widetilde{\rho
}^{k}=\Pi^{k}\widetilde{\rho}\Pi^{k}$. Then it follows that%
\begin{align}
\operatorname{Tr}[M^{k}\sigma^{k}]  &  =\operatorname{Tr}[\Pi^{k}M^{k}\Pi
^{k}\sigma]\leq1,\\
P(\widetilde{\rho}^{k},\rho^{k})  &  \leq\varepsilon,
\end{align}
and%
\begin{align}
\inf_{\widehat{\rho}^{k}:P(\widehat{\rho}^{k},\rho^{k})\leq\varepsilon
}\operatorname{Tr}[M^{k}\widehat{\rho}^{k}]  &  \leq\operatorname{Tr}%
[M^{k}\widetilde{\rho}^{k}]\\
&  =\operatorname{Tr}[M^{k}\Pi^{k}\widetilde{\rho}\Pi^{k}]\\
&  =\operatorname{Tr}[\Pi^{k}M^{k}\Pi^{k}\widetilde{\rho}].
\end{align}
Since the inequality holds for all $\widetilde{\rho}$ satisfying
$P(\widetilde{\rho},\rho)\leq\varepsilon$, we conclude that%
\begin{align}
&  \inf_{\widehat{\rho}^{k}:P(\widehat{\rho}^{k},\rho^{k})\leq\varepsilon
}\operatorname{Tr}[M^{k}\widehat{\rho}^{k}]\nonumber\\
&  \leq\inf_{\widetilde{\rho}:P(\widetilde{\rho},\rho)\leq\varepsilon
}\operatorname{Tr}[\Pi^{k}M^{k}\Pi^{k}\widetilde{\rho}]\\
&  \leq\sup_{M\geq0}\inf_{\widetilde{\rho}:P(\widetilde{\rho},\rho
)\leq\varepsilon}\left\{  \operatorname{Tr}[M\widetilde{\rho}%
]:\operatorname{Tr}[M\sigma]\leq1\right\} \\
&  =\widehat{D}_{\max}^{\varepsilon}(\rho\Vert\sigma),
\end{align}
where the second inequality follows because $\Pi^{k}M^{k}\Pi^{k}$ is a
particular compact bounded operator satisfying $\operatorname{Tr}[\Pi^{k}%
M^{k}\Pi^{k}\sigma]\leq1$. Since the inequality above has been shown for an
arbitrary compact bounded operator $M^{k}$\ satisfying $\operatorname{Tr}%
[M^{k}\sigma^{k}]\leq1$, we conclude that%
\begin{multline}
\sup_{M^{k}\geq0}\inf_{\widehat{\rho}^{k}:P(\widehat{\rho}^{k},\rho^{k}%
)\leq\varepsilon}\left\{  \operatorname{Tr}[M^{k}\widehat{\rho}^{k}%
]:\operatorname{Tr}[M^{k}\sigma^{k}]\leq1\right\} \\
=\widehat{D}_{\max}^{\varepsilon}(\rho^{k}\Vert\sigma^{k})\leq\widehat
{D}_{\max}^{\varepsilon}(\rho\Vert\sigma).
\end{multline}
Since this inequality has been shown for arbitrary $k$, we conclude that%
\begin{equation}
\limsup_{k\rightarrow\infty}\widehat{D}_{\max}^{\varepsilon}(\rho^{k}%
\Vert\sigma^{k})\leq\widehat{D}_{\max}^{\varepsilon}(\rho\Vert\sigma).
\label{eq:dual-smooth-dmax-ineq-seq}%
\end{equation}

Finally, by combining \eqref{eq:weak-duality-starting-point},
\eqref{eq:strong-duality-conseq}, and \eqref{eq:dual-smooth-max-lim}, we
conclude that%
\begin{equation}
\widehat{D}_{\max}^{\varepsilon}(\rho\Vert\sigma)=D_{\max}^{\varepsilon}%
(\rho\Vert\sigma).
\end{equation}
This completes the proof.
\end{proof}

\section{Relation between smooth max-relative entropy and hypothesis testing
relative entropy}

\label{app:smooth-dmax-hypo-test}This appendix establishes inequalities
relating the smooth max-relative entropy and the hypothesis testing relative
entropy of states acting on a separable Hilbert space.

\begin{lemma}
\label{lem:smooth-dmax-to-hypo-test}Let $\rho$ be a state and $\sigma$ 
a trace-class positive semi-definite operator  acting
on a separable Hilbert space~$\mathcal{H}$. Let $\varepsilon\in(0,1)$, and set
$\delta\in(0,1)$ such that $\varepsilon+\delta<1$. Suppose that $D_{\max
}^{\sqrt{\varepsilon}}(\rho\Vert\sigma)<\infty$. Then the following bounds
hold%
\begin{align}
D_{\max}^{\sqrt{\varepsilon}}(\rho\Vert\sigma)  &  \leq D_{H}^{1-\varepsilon
}(\rho\Vert\sigma)+\log_{2}\!\left(  \frac{1}{1-\varepsilon}\right)
,\label{eq:smooth-dmax-up-bnd}\\
D_{H}^{1-\varepsilon-\delta}(\rho\Vert\sigma)  &  \leq D_{\max}^{\sqrt
{\varepsilon}}(\rho\Vert\sigma)+\log_{2}\!\left(  \frac{4\left(
1-\varepsilon\right)  }{\delta^{2}}\right)  . \label{eq:smooth-dmax-low-bnd}%
\end{align}

\end{lemma}

\begin{proof}
These bounds were established for the finite-dimensional case in
\cite[Theorem~4]{ABJT19}. Here we check that the arguments presented in the
proof of \cite[Theorem~4]{ABJT19} hold for the more general case in the
statement of the lemma, and we combine with the result from
Lemma~\ref{lem:smooth-dmax-primal-dual-equal}.

We first consider the inequality in \eqref{eq:smooth-dmax-up-bnd}. Let
$M\geq0$ be an arbitrary compact bounded operator satisfying
$\operatorname{Tr}[M\sigma]\leq1$. Since $M$ is compact, it has a countable
spectral decomposition as $M=\sum_{i}m_{i}|\phi_{i}\rangle\!\langle\phi_{i}|$,
where $\{m_{i}\}_{i}$ is a countable set of non-negative eigenvalues and
$\{|\phi_{i}\rangle\}_{i}$ is a countable orthonormal basis. We define a
measurement (or completely dephasing)\ channel as
\begin{equation}
\mathcal{M}(\cdot)\coloneqq \sum_{i}|\phi_{i}\rangle\!\langle\phi_{i}|(\cdot)|\phi
_{i}\rangle\!\langle\phi_{i}|,
\end{equation}
and we set the probability distributions $P$ and $Q$ such that
\begin{align}
\mathcal{M} (\rho)  &  =\sum_{i}P(i)|\phi_{i}\rangle\!\langle\phi_{i}|,\\
\mathcal{M} (\sigma)  &  =\sum_{i}Q(i)|\phi_{i}\rangle\!\langle\phi_{i}|.
\end{align}
Then the data processing inequality for the hypothesis testing relative
entropy implies that%
\begin{equation}
D_{H}^{1-\varepsilon}(\rho\Vert\sigma)\geq D_{H}^{1-\varepsilon}(P\Vert Q)\geq
D_{s}^{1-\varepsilon}(P\Vert Q)\coloneqq K,
\end{equation}
where the $\varepsilon$-information spectrum relative entropy of states $\tau$
and $\omega$ is defined as \cite{TH12}
\begin{equation}
D_{s}^{\varepsilon}(\tau\Vert\omega)\coloneqq \sup\left\{  \lambda\in\mathbb{R}%
:\operatorname{Tr}[\tau\left\{  \tau\leq2^{\lambda}\omega\right\}  ]\leq
\varepsilon\right\}  ,
\end{equation}
and $\left\{  \tau\leq2^{\lambda}\omega\right\} $ denotes the projection onto the
positive part of $2^{\lambda}\omega- \tau$. The second inequality above follows by
picking $\Lambda\coloneqq \left\{  \tau>2^{\lambda}\omega\right\}  $ for $\lambda=D_{s}%
^{\varepsilon}(\tau\Vert\omega)-\xi$ and $\xi>0$. Since $\Lambda$ satisfies
$\operatorname{Tr}[\Lambda\tau]\geq1-\varepsilon$, and we also have that%
\begin{align}
\operatorname{Tr}[\Lambda\omega]  &  =\operatorname{Tr}[\left\{  \tau
>2^{\lambda}\omega\right\}  \omega]\\
&  \leq2^{-\lambda}\operatorname{Tr}[\tau\left\{  \tau>2^{\lambda}\omega\right\}  ]\\
&  \leq2^{-\lambda},
\end{align}
it follows from definitions that%
\begin{equation}
\lambda=D_{s}^{\varepsilon}(\tau\Vert\omega)-\xi\leq D_{H}^{\varepsilon}(\tau
\Vert\omega).
\end{equation}
Since the inequality has been shown for all $\xi>0$, it follows that
$D_{s}^{\varepsilon}(\tau\Vert\omega)\leq D_{H}^{\varepsilon}(\tau\Vert
\omega)$. Now set $\eta>0$, and define the set%
\begin{equation}
I\coloneqq \left\{  i:P(i)\leq2^{K+\eta}Q(i)\right\}  ,
\end{equation}
which implies that $P(I)\coloneqq  \sum_{i \in I} P(i) > 1-\varepsilon$ from the
definition of $D_{s}^{1-\varepsilon}(P\Vert Q)$ (indeed, from the definition of $D_{s}^{1-\varepsilon}(P\Vert Q)$, $K$ is the largest possible value such that $\sum_{i : P(i) \leq 2^K Q(i)} P(i) \leq 1-\varepsilon$, so that if it is increased by $\eta > 0 $, then $P(I) >  1-\varepsilon$). Now define the projection
$\Pi\coloneqq \sum_{i\not \in I}|\phi_{i}\rangle\!\langle\phi_{i}|$. Then it follows
that%
\begin{align}
\operatorname{Tr}[\Pi\rho]  &  =\operatorname{Tr}[\mathcal{M}(\Pi)\rho]\\
&  =\operatorname{Tr}[\Pi\mathcal{M}(\rho)]\\
&  =1-P(I)\\
&  \leq\varepsilon.
\end{align}
Set%
\begin{equation}
\widetilde{\rho}\coloneqq \frac{\left(  I-\Pi\right)  \rho\left(  I-\Pi\right)
}{1-\operatorname{Tr}[\Pi\rho]}.
\end{equation}
Then it follows that $P(\rho,\widetilde{\rho})\leq\sqrt{\varepsilon}$ because%
\begin{align}
&  F(\rho,\widetilde{\rho})\nonumber\\
&  =\left\Vert \sqrt{\rho}\sqrt{\widetilde{\rho}}\right\Vert _{1}^{2}\\
&  =\frac{1}{1-\operatorname{Tr}[\Pi\rho]}\left\Vert \sqrt{\rho}\sqrt{\left(
I-\Pi\right)  \rho\left(  I-\Pi\right)  }\right\Vert _{1}^{2}\\
&  =\frac{1}{1-\operatorname{Tr}[\Pi\rho]}\times\nonumber\\
&  \quad\left\Vert \sqrt{\left(  I-\Pi\right)  \rho\left(  I-\Pi\right)
}\sqrt{\left(  I-\Pi\right)  \rho\left(  I-\Pi\right)  }\right\Vert _{1}^{2}\\
&  =\frac{1}{1-\operatorname{Tr}[\Pi\rho]}\left\Vert \left(  I-\Pi\right)
\rho\left(  I-\Pi\right)  \right\Vert _{1}^{2}\\
&  =1-\operatorname{Tr}[\Pi\rho].
\end{align}
Then since $1-\operatorname{Tr}[\Pi\rho] \geq 1-\varepsilon$, it follows that%
\begin{align}
\operatorname{Tr}[M\widetilde{\rho}]  &  =\frac{1}{1-\operatorname{Tr}[\Pi
\rho]}\operatorname{Tr}[M\left(  I-\Pi\right)  \rho\left(  I-\Pi\right)  ]\\
&  =\frac{1}{1-\operatorname{Tr}[\Pi\rho]}\sum_{i\in I}m_{i}P(i)\\
&  \leq\frac{1}{1-\varepsilon}2^{K+\eta}\sum_{i\in I}m_{i}Q(i)\\
&  \leq\frac{1}{1-\varepsilon}2^{K+\eta}\sum_{i}m_{i}Q(i)\\
&  =\frac{1}{1-\varepsilon}2^{K+\eta}\operatorname{Tr}[M\sigma]\\
&  \leq\frac{1}{1-\varepsilon}2^{D_{H}^{1-\varepsilon}(\rho\Vert\sigma)+\eta}.
\end{align}
Taking a logarithm and an infimum over all states $\widetilde{\rho}$
satisfying $P(\rho,\widetilde{\rho})\leq\sqrt{\varepsilon}$, we find that%
\begin{multline}
\log_{2}\inf_{\widetilde{\rho}:P(\rho,\widetilde{\rho})\leq\sqrt{\varepsilon}%
}\operatorname{Tr}[M\widetilde{\rho}]\\
\leq D_{H}^{1-\varepsilon}(\rho\Vert\sigma)+\eta+\log_{2}\!\left(  \frac
{1}{1-\varepsilon}\right)  .
\end{multline}
Since the inequality has been shown for all compact bounded operators $M\geq0$
satisfying $\operatorname{Tr}[M\sigma]\leq1$, it follows that%
\begin{align}
& \widehat{D}_{\max}^{\sqrt{\varepsilon}}(\rho\Vert\sigma)  \nonumber \\
&  =\log_{2}%
\sup_{M\geq0, \operatorname{Tr}[M\sigma]\leq1}\left[\inf_{\widetilde{\rho}:P(\rho,\widetilde{\rho})\leq
\sqrt{\varepsilon}}\operatorname{Tr}[M\widetilde{\rho}]\right]\\
&  \leq D_{H}^{1-\varepsilon}(\rho\Vert\sigma)+\eta+\log_{2}\!\left(  \frac
{1}{1-\varepsilon}\right)  .
\end{align}
Since the inequality has been shown for all $\eta>0$, it follows that%
\begin{equation}
\widehat{D}_{\max}^{\sqrt{\varepsilon}}(\rho\Vert\sigma)\leq D_{H}%
^{1-\varepsilon}(\rho\Vert\sigma)+\log_{2}\!\left(  \frac{1}{1-\varepsilon
}\right)  .
\end{equation}
Now applying Lemma~\ref{lem:smooth-dmax-primal-dual-equal}, we conclude the
inequality in~\eqref{eq:smooth-dmax-up-bnd}.

We now consider the other inequality in \eqref{eq:smooth-dmax-low-bnd}. Let
$\lambda\geq0$ and $\widetilde{\rho}$ be a state satisfying $P(\widetilde
{\rho},\rho)\leq\sqrt{\varepsilon}$ and
\begin{equation}
\widetilde{\rho}\leq2^{\lambda}\sigma. \label{eq:dmax-smooth-unraveled}%
\end{equation}
Consider an arbitrary operator $\Lambda$ satisfying $0\leq\Lambda\leq I$ and
$\operatorname{Tr}[\left(  I-\Lambda\right)  \rho]=1-\varepsilon-\delta$. Then
the data processing inequality for the fidelity, with respect to the
measurement $\{ \Lambda, I -\Lambda\}$, implies that%
\begin{align}
&  \sqrt{1-\varepsilon}\nonumber\\
&  \leq\sqrt{F}(\widetilde{\rho},\rho)\nonumber\\
&  \leq\sqrt{\operatorname{Tr}[\Lambda\widetilde{\rho}]\operatorname{Tr}%
[\Lambda\rho]}+\sqrt{\operatorname{Tr}[\left(  I-\Lambda\right)
\widetilde{\rho}]\operatorname{Tr}[\left(  I-\Lambda\right)  \rho]}\nonumber\\
&  \leq\sqrt{2^{\lambda}\operatorname{Tr}[\Lambda\sigma]}+\sqrt{1-\varepsilon
-\delta}.
\end{align}
Now taking an infimum over all $\lambda$ such that
\eqref{eq:dmax-smooth-unraveled} holds and all $\Lambda$ satisfying
$0\leq\Lambda\leq I$ and $\operatorname{Tr}[\left(  I-\Lambda\right)
\rho]=1-\varepsilon-\delta$, and applying definitions, we find that%
\begin{equation}
\sqrt{1-\varepsilon}\leq\sqrt{2^{D_{\max}^{\sqrt{\varepsilon}}(\rho\Vert
\sigma)}2^{-D_{H}^{1-\varepsilon-\delta}(\rho\Vert\sigma)}}+\sqrt
{1-\varepsilon-\delta}.
\end{equation}
Rewriting this gives%
\begin{multline}
\log_{2}\left[  \sqrt{1-\varepsilon}-\sqrt{1-\varepsilon-\delta}\right]
^{2}\\
\leq D_{\max}^{\sqrt{\varepsilon}}(\rho\Vert\sigma)-D_{H}^{1-\varepsilon
-\delta}(\rho\Vert\sigma).
\end{multline}
Now applying the inequality $\sqrt{1-\varepsilon}-\sqrt{1-\varepsilon-\delta
}\geq\frac{\delta}{2\sqrt{1-\varepsilon}}$, we conclude \eqref{eq:smooth-dmax-low-bnd}.
\end{proof}

\section{Second-order asymptotics of smooth max-relative entropy for states
acting on a separable Hilbert space}

\label{app:2nd-order-smooth-dmax}Recall the quantities defined in
\eqref{eq:rho-sig-spectral-decomps}--\eqref{eq:rel-ent-var-sep}. We also
define%
\begin{align}
\Phi(a)  &  \coloneqq \frac{1}{\sqrt{2\pi}}\int_{-\infty}^{a}dx\ \exp\!\left(
\frac{-x^{2}}{2}\right)  ,\\
\Phi^{-1}(\varepsilon)  &  \coloneqq \sup\left\{  a\in\mathbb{R}\ |\ \Phi
(a)\leq\varepsilon\right\}  .
\end{align}
For $\varepsilon\in(0,1)$, observe that%
\begin{equation}
\Phi^{-1}(1-\varepsilon)=-\Phi^{-1}(\varepsilon).
\label{eq:cumula-normal-symmetry}%
\end{equation}
Let us define the quantity $T(\rho\Vert\sigma)$ \cite{TH12,li12,KW17a} as%
\begin{equation}
T(\rho\Vert\sigma)\coloneqq \sum_{x,y}\left\vert \langle\phi_{y}|\psi_{x}%
\rangle\right\vert ^{2}\lambda_{x}\left\vert \log_{2}\!\left(  \frac
{\lambda_{x}}{\mu_{y}}\right)  -D(\rho\Vert\sigma)\right\vert ^{3},
\label{eq:T-quantity}%
\end{equation}
where the eigendecompositions of the state $\rho$ and the trace-class positive semi-definite operator $\sigma$ are given in
\eqref{eq:rho-sig-spectral-decomps}. Then the following expansion holds for
$\rho$ and $\sigma$ such that $D(\rho\Vert\sigma),V(\rho\Vert\sigma
),T(\rho\Vert\sigma)<\infty$ and $V(\rho\Vert\sigma)>0$:%
\begin{multline}
\frac{1}{n}D_{H}^{\varepsilon}(\rho^{\otimes n}\Vert\sigma^{\otimes n})=\\
D(\rho\Vert\sigma)+\sqrt{\frac{1}{n}V(\rho\Vert\sigma)}\Phi^{-1}%
(\varepsilon)+O\!\left(  \frac{\log n}{n}\right)  .
\label{eq:hypo-test-rel-ent-2nd-order}%
\end{multline}
The equality in \eqref{eq:hypo-test-rel-ent-2nd-order} was shown for the
finite-dimensional case in \cite{li12,TH12}. For a state $\rho$ and positive-semidefinite trace-class $\sigma$ acting on a
separable Hilbert space, the inequality $\leq$ was shown in \cite{DPR15,KW17a}%
, while the inequality $\geq$ was established in \cite{li12,DPR15,OMW19}. The
term $O\!\left(  \frac{\log n}{n}\right)  $ above hides constants involving
$V(\rho\Vert\sigma)$, $T(\rho\Vert\sigma)$, and $\varepsilon$.

By combining the equality in \eqref{eq:hypo-test-rel-ent-2nd-order} with
Lemma~\ref{lem:smooth-dmax-to-hypo-test}, we arrive at the following
corollary, giving the second-order asymptotics for smooth max-relative entropy
of trace-class operators acting on a separable Hilbert space:

\begin{corollary}
\label{cor-2nd-order-smooth-dmax}Let $\varepsilon\in(0,1)$. Let $\rho$ be a state and
$\sigma$ a trace-class positive semi-definite operator acting on a separable Hilbert space $\mathcal{H}$, such
that $D(\rho\Vert\sigma),V(\rho\Vert\sigma),T(\rho\Vert\sigma)<\infty$ and
$V(\rho\Vert\sigma)>0$. Then the following second-order expansion holds for
sufficiently large $n$:%
\begin{multline}
\frac{1}{n}D_{\max}^{\sqrt{\varepsilon}}(\rho^{\otimes n}\Vert\sigma^{\otimes
n})=\\
D(\rho\Vert\sigma)-\sqrt{\frac{1}{n}V(\rho\Vert\sigma)}\Phi^{-1}%
(\varepsilon)+O\!\left(  \frac{\log n}{n}\right)  .
\end{multline}

\end{corollary}

\begin{proof}
Exploiting the inequality in \eqref{eq:smooth-dmax-up-bnd}, we find that
\begin{align}
&  \frac{1}{n}D_{\max}^{\sqrt{\varepsilon}}(\rho^{\otimes n}\Vert
\sigma^{\otimes n})\nonumber\\
&  \leq\frac{1}{n}D_{H}^{1-\varepsilon}(\rho^{\otimes n}\Vert\rho^{\otimes
n})+\frac{1}{n}\log_{2}\!\left(  \frac{1}{1-\varepsilon}\right) \\
&  =D(\rho\Vert\sigma)+\sqrt{\frac{1}{n}V(\rho\Vert\sigma)}\Phi^{-1}%
(1-\varepsilon)+O\!\left(  \frac{\log n}{n}\right) \\
&  =D(\rho\Vert\sigma)-\sqrt{\frac{1}{n}V(\rho\Vert\sigma)}\Phi^{-1}%
(\varepsilon)+O\!\left(  \frac{\log n}{n}\right)  .
\end{align}
where the first equality follows from \eqref{eq:hypo-test-rel-ent-2nd-order}
and the last equality follows from \eqref{eq:cumula-normal-symmetry}. Now
exploiting the inequality in \eqref{eq:smooth-dmax-low-bnd} and choosing
$\delta=1/\sqrt{n}$, we find that%
\begin{align}
&  \frac{1}{n}D_{\max}^{\sqrt{\varepsilon}}(\rho^{\otimes n}\Vert
\sigma^{\otimes n})\nonumber\\
&  \geq\frac{1}{n}D_{H}^{1-\varepsilon-\delta}(\rho^{\otimes n}\Vert
\sigma^{\otimes n})-\frac{1}{n}\log_{2}\!\left(  \frac{4\left(  1-\varepsilon
\right)  }{\delta^{2}}\right) \nonumber\\
&  =D(\rho\Vert\sigma)+\sqrt{\frac{1}{n}V(\rho\Vert\sigma)}\Phi^{-1}%
(1-\varepsilon)+O\!\left(  \frac{\log n}{n}\right) \nonumber\\
&  =D(\rho\Vert\sigma)-\sqrt{\frac{1}{n}V(\rho\Vert\sigma)}\Phi^{-1}%
(\varepsilon)+O\!\left(  \frac{\log n}{n}\right)  .
\end{align}
In the first equality, we have invoked a standard step from \cite[Footnote~6]%
{TH12} applied to $\Phi^{-1}(1-\varepsilon-\delta)$, which is an invocation of
Taylor's theorem: for $f$ continuously differentiable, $c$ a positive
constant, and $n \geq n_{0}$, the following equality holds
\begin{equation}
\sqrt{n} f(x + c/\sqrt{n}) = \sqrt{n}f(x) + c f^{\prime}(a)
\end{equation}
for some $a \in[x,x+c/\sqrt{n_{0}}]$. This concludes the proof.
\end{proof}

\section{Mutual information variance and $T$ quantity of classical--quantum
states}

\label{app:mut-info-var-and-T}

\begin{proposition}
Given is a classical--quantum state of the following form:%
\begin{equation}
\rho_{XB}\coloneqq \sum_{x}p_{X}(x)|x\rangle\!\langle x|_{X}\otimes\rho_{B}^{x}.
\end{equation}
Suppose that $\rho_{B}^{x}$ and $\rho_{B}\coloneqq \sum_{x}p_{X}(x)\rho_{B}^{x}$ have
the following spectral decompositions:%
\begin{align}
\rho_{B}^{x}  &  =\sum_{y}p(y|x)|\psi^{x,y}\rangle\!\langle\psi^{x,y}|_{B},\\
\rho_{B}  &  =\sum_{z}q(z)|\phi^{z}\rangle\!\langle\phi^{z}|_{B},
\end{align}
for $p(y|x)$ a conditional probability distribution, $\{|\psi^{x,y}\rangle
_{B}\}_{y}$ an orthonormal set of eigenvectors (for fixed $x$), $q(z)$ a
probability distribution, and $\{|\phi^{z}\rangle_{B}\}_{z}$ an orthonormal
set of eigenvectors. Then%
\begin{align}
&  V(X;B)_{\rho}\nonumber\\
&  =\sum_{x}p_{X}(x)\left(  V(\rho_{B}^{x}\Vert\rho_{B})+\left[  D(\rho
_{B}^{x}\Vert\rho_{B})\right]  ^{2}\right)  -\left[  I(X;B)\right]
^{2}\label{eq:rel-ent-var-cq}\\
&  =\sum_{x}p(x)\sum_{y,z}\left\vert \langle\psi^{x,y}|\phi^{z}\rangle
_{B}\right\vert ^{2}p(y|x)\left\vert f(x,y,z)\right\vert ^{2},
\label{eq:rel-ent-var-cq-other}%
\end{align}%
\begin{multline}
T(X;B)_{\rho}=\\
\sum_{x}p(x)\sum_{y,z}\left\vert \langle\psi^{x,y}|\phi^{z}\rangle
_{B}\right\vert ^{2}p(y|x)\left\vert f(x,y,z)\right\vert ^{3},
\end{multline}
where%
\begin{equation}
f(x,y,z)\coloneqq \log_{2}\left(  p(y|x)/q(z)\right)  -I(X;B).
\end{equation}

\end{proposition}

\begin{proof}
To see the first expression in \eqref{eq:rel-ent-var-cq}, consider that%
\begin{align}
&  \log_{2}\rho_{XB}-\log_{2}\rho_{X}\otimes\rho_{B}\nonumber\\
&  =\log_{2}\left[  \sum_{x}|x\rangle\!\langle x|_{X}\otimes p_{X}(x)\rho
_{B}^{x}\right] \nonumber\\
&  \qquad\qquad-\log_{2}\left[  \sum_{x}|x\rangle\!\langle x|_{X}\otimes
p_{X}(x)\rho_{B}\right] \nonumber\\
&  =\sum_{x}|x\rangle\!\langle x|_{X}\otimes\left(  \log_{2}\left[  p_{X}%
(x)\rho_{B}^{x}\right]  -\log_{2}\left[  p_{X}(x)\rho_{B}\right]  \right)
\nonumber\\
&  =\sum_{x}|x\rangle\!\langle x|_{X}\otimes\left(  \log_{2}\rho_{B}^{x}%
-\log_{2}\rho_{B}\right)  ,
\end{align}
and we find that%
\begin{multline}
\left[  \log_{2}\rho_{XB}-\log_{2}\rho_{X}\otimes\rho_{B}-I(X;B)\right]
^{2}\\
=\sum_{x}|x\rangle\!\langle x|_{X}\otimes\left(  \log_{2}\rho_{B}^{x}-\log
_{2}\rho_{B}-I(X;B)\right)  ^{2},
\end{multline}
so that%
\begin{align}
&  \operatorname{Tr}[\rho_{XB}\left(  \log_{2}\rho_{XB}-\log_{2}\rho
_{X}\otimes\rho_{B}-I(X;B)\right)  ^{2}]\nonumber\\
&  =\operatorname{Tr}\left[
\begin{array}
[c]{c}%
\left(  \sum_{x^{\prime}}p_{X}(x^{\prime})|x^{\prime}\rangle\!\langle x^{\prime
}|_{X}\otimes\rho_{B}^{x^{\prime}}\right)  \times\notag\\
\left(  \sum_{x}|x\rangle\!\langle x|_{X}\otimes\left(  \log_{2}\rho_{B}%
^{x}-\log_{2}\rho_{B}-I(X;B)\right)  ^{2}\right)
\end{array}
\right] \nonumber\\
&  =\sum_{x}p_{X}(x)\operatorname{Tr}[\rho_{B}^{x}\left(  \log_{2}\rho_{B}%
^{x}-\log_{2}\rho_{B}-I(X;B)\right)  ^{2}]\nonumber\\
&  =\sum_{x}p_{X}(x)\operatorname{Tr}[\rho_{B}^{x}\left(  \log_{2}\rho_{B}%
^{x}-\log_{2}\rho_{B}\right)  ^{2}]-\left[  I(X;B)\right]  ^{2}.
\label{eq:MI-variance-cq}%
\end{align}
Now consider that%
\begin{align}
&  \operatorname{Tr}[\rho_{B}^{x}\left(  \log_{2}\rho_{B}^{x}-\log_{2}\rho
_{B}\right)  ^{2}]\nonumber\\
&  =\operatorname{Tr}[\rho_{B}^{x}\left(  \log_{2}\rho_{B}^{x}-\log_{2}%
\rho_{B}\right)  ^{2}]\nonumber\\
&  \qquad\qquad-\left[  D(\rho_{B}^{x}\Vert\rho_{B})\right]  ^{2}+\left[
D(\rho_{B}^{x}\Vert\rho_{B})\right]  ^{2}\nonumber\\
&  =\operatorname{Tr}[\rho_{B}^{x}\left(  \log_{2}\rho_{B}^{x}-\log_{2}%
\rho_{B}-D(\rho_{B}^{x}\Vert\rho_{B})\right)  ^{2}]\nonumber\\
&  \qquad\qquad+\left[  D(\rho_{B}^{x}\Vert\rho_{B})\right]  ^{2}\nonumber\\
&  =V(\rho_{B}^{x}\Vert\rho_{B})+\left[  D(\rho_{B}^{x}\Vert\rho_{B})\right]
^{2}.
\end{align}
Substituting the last line back in \eqref{eq:MI-variance-cq} gives the formula
in \eqref{eq:rel-ent-var-cq}.

To see the other expressions, consider that%
\begin{align}
&  T(\rho_{XB}\Vert\rho_{X}\otimes\rho_{B})\nonumber\\
&  =\sum_{x,x^{\prime},y,z}\left\vert \left(  \langle x|_{X}\otimes\langle
\psi^{x,y}|_{B}\right)  \left(  |x^{\prime}\rangle_{X}\otimes|\phi^{z}%
\rangle_{B}\right)  \right\vert ^{2}p(x)p(y|x)\times\nonumber\\
&  \qquad\qquad\left\vert \log_{2}\left(  p(x)p(y|x)/\left[  p(x^{\prime
})q(z)\right]  \right)  -I(X;B)\right\vert ^{3}\\
&  =\sum_{x,y,z}\left\vert \langle\psi^{x,y}|\phi^{z}\rangle_{B}\right\vert
^{2}p(x)p(y|x)\left\vert f(x,y,z)\right\vert ^{3}\\
&  =\sum_{x}p(x)\sum_{y,z}\left\vert \langle\psi^{x,y}|\phi^{z}\rangle
_{B}\right\vert ^{2}p(y|x)\left\vert f(x,y,z)\right\vert ^{3}.
\end{align}
By employing a similar development as above, the formula in
\eqref{eq:rel-ent-var-cq-other} is then clear.
\end{proof}

\section{Holevo information and Holevo information variance of a Gaussian
ensemble of Gaussian states}

\label{app:Holevo-info-and-var-Gaussian}

In what follows, we determine formulas for the Holevo information and Holevo
information variance of a Gaussian ensemble of Gaussian states. The first is
well known, but the expression we give for it seems to be novel. The latter
has not been presented yet to our knowledge. Before presenting the formulas,
we provide a brief review of quantum Gaussian states in order to set some
notation. We refer to \cite{Ser17} for an in-depth overview of quantum
Gaussian states and measurements.

We consider $m$-mode Gaussian states, where $m$ is some fixed positive
integer. Let $\hat{r}_{j}$ denote each quadrature operator ($2m$ of them for
an $m$-mode state), and let%
\begin{align}
\hat{r}  &  \equiv\left[  \hat{x}_{1},\hat{p}_{1},\ldots,\hat{x}_{m},\hat
{p}_{m}\right] \\
&  \equiv\left[  \hat{r}_{1},\ldots,\hat{r}_{2m}\right]
\end{align}
denote the vector of quadrature operators, so that the odd entries correspond
to position-quadrature operators and the even entries to momentum-quadrature
operators. The quadrature operators satisfy the following commutation
relations:%
\begin{equation}
\left[  \hat{r}_{j},\hat{r}_{k}\right]  =i\Omega_{j,k},\quad\mathrm{where}%
\quad\Omega\coloneqq I_{m}\otimes%
\begin{bmatrix}
0 & 1\\
-1 & 0
\end{bmatrix}
\text{,} \label{eq:symplectic-form}%
\end{equation}
and $I_{m}$ is the $m\times m$ identity matrix. We also take the annihilation
operator $\hat{a}=\left(  \hat{x}+i\hat{p}\right)  /\sqrt{2}$.

Let $\rho$ be a Gaussian state, with the mean-vector entries%
\begin{equation}
\mu_{j}\coloneqq \left\langle \hat{r}_{j}\right\rangle _{\rho},
\end{equation}
and let $\mu$ denote the mean vector. The entries of the covariance matrix
$V$\ of $\rho$ are given by
\begin{equation}
V_{j,k}\equiv\left\langle \left\{  \hat{r}_{j}-\mu_{j},\hat{r}_{k}-\mu
_{k}\right\}  \right\rangle _{\rho},\label{eq:covariance-matrices}%
\end{equation}
and they satisfy the uncertainty principle $V+i\Omega\geq0$. A $2m\times2m$
matrix $S$ is symplectic if it preserves the symplectic form: $S\Omega
S^{T}=\Omega$. According to Williamson's theorem \cite{W36}, there is a
diagonalization of the covariance matrix $V$ of the form,
\begin{equation}
V=S\left(  D\otimes I_{2}\right)  S^{T},
\end{equation}
where $S$ is a symplectic matrix and $D\equiv\operatorname{diag}(\nu
_{1},\ldots,\nu_{m})$ is a diagonal matrix of symplectic eigenvalues such that
$\nu_{i}\geq1$ for all $i\in\left\{  1,\ldots,m\right\}  $. We say that a
quantum Gaussian state is faithful if all of its symplectic eigenvalues are
strictly greater than one (this also means that the state is positive
definite). We can write the density operator $\rho$ of a faithful state in the
following exponential form \cite{PhysRevA.71.062320,K06,H10} (see also
\cite{H13book,Ser17}):
\begin{align}
&  \rho=Z^{-1/2}\exp\left[  -\frac{1}{2}(\hat{r}-\mu)^{T}G(\hat{r}%
-\mu)\right]  ,\label{eq:exp-form}\\
&  \mathrm{with}\quad Z\coloneqq \det(\left[  V+i\Omega\right]  /2)\\
&  \mathrm{and}\quad G\coloneqq -2\Omega S\left[  \operatorname{arcoth}(D)\otimes
I_{2}\right]  S^{T}\Omega,\label{eq:G_rho}%
\end{align}
where $\operatorname{arcoth}(x)\equiv\frac{1}{2}\ln\!\left(  \frac{x+1}%
{x-1}\right)  $ with domain $\left(  -\infty,-1\right)  \cup\left(
1,+\infty\right)  $. The matrix $G$ is known as the Hamiltonian matrix
\cite{Ser17}. Note that we can also write%
\begin{equation}
G=2i\Omega\operatorname{arcoth}(iV\Omega),\label{eq:more-compact-G-rho}%
\end{equation}
so that $G$ is represented directly in terms of the covariance matrix $V$.
Faithfulness of Gaussian states is required to ensure that $G$ is
non-singular. By inspection, the Hamiltonian matrix $G$ and the covariance
matrix $V$ are symmetric.

Let $\mu^{\rho}$, $V^{\rho}$, $G^{\rho}$, $Z^{\rho}$ denote the various
quantities above for an $m$-mode quantum Gaussian state $\rho$, and let
$\mu^{\sigma}$, $V^{\sigma}$, $G^{\sigma}$, $Z^{\sigma}$ denote the various
quantities for an $m$-mode quantum Gaussian state $\sigma$. Then the relative
entropy $D(\rho\Vert\sigma)$ is given by the following formula:%
\begin{equation}
D(\rho\Vert\sigma)=\frac{1}{2\ln2}\left[  \ln\!\left(  \frac{Z_{\sigma}%
}{Z_{\rho}}\right)  +\frac{1}{2}\operatorname{Tr}[V^{\rho}\Delta]+\delta
^{T}G^{\sigma}\delta\right]  . \label{eq:rel-ent-Gaussian}%
\end{equation}
The formula in \eqref{eq:rel-ent-Gaussian} was established by
\cite{PhysRevA.71.062320} for the zero-mean case, by \cite{K06} for the case
of non-zero mean but equal covariance matrices, and then extended by
\cite{PLOB15} to the case of general multi-mode Gaussian states. The relative
entropy variance is given by the following formula \cite{BLTW16}:
\begin{multline}
V(\rho\Vert\sigma)=\frac{1}{8\ln^{2}2}\left[  \operatorname{Tr}[\left(  \Delta
V^{\rho}\right)  ^{2}]+\operatorname{Tr}[\left(  \Delta\Omega\right)
^{2}]\right] \label{eq:rel-ent-var-Gaussian}\\
+\frac{1}{2\ln^{2}2}\delta^{T}G^{\sigma}V^{\rho}G^{\sigma}\delta,
\end{multline}
where%
\begin{align}
\delta &  \coloneqq \mu^{\rho}-\mu^{\sigma},\\
\Delta &  \coloneqq G^{\sigma}-G^{\rho}.
\end{align}

Recall from \eqref{eq:Holevo-info-def-rel-ent} and
\eqref{eq:holevo-info-var-reverse} that the Holevo information and Holevo
information variance of a general ensemble $\{p_{Y}(y),\rho_{E}^{y}\}_{y}$ are
given by the following, with sums replaced by integrals:%
\begin{align}
I(Y;E) &  =\int dy\ p_{Y}(y)D(\rho_{E}^{y}\Vert\rho_{E}),\\
V(Y;E) &  =\int dy\ p_{Y}(y)\left[  V(\rho_{E}^{y}\Vert\rho_{E})+\left(
D(\rho_{E}^{y}\Vert\rho_{E})\right)  ^{2}\right]  \nonumber\\
&  \qquad-\left[  I(Y;E)\right]  ^{2},
\end{align}
where%
\begin{equation}
\rho_{E}\coloneqq \int dy\ p_{Y}(y)\rho_{E}^{y}.\label{eq:rho_E_gaussian}%
\end{equation}
Then the proof of Proposition~\ref{prop:gaussian-formulas-HI-Hi-var}\ is as follows:

\bigskip

\begin{proof}
[Proof of Proposition~\ref{prop:gaussian-formulas-HI-Hi-var}]Recall
Definition~\ref{def:gaussian-ensemble}\ for a Gaussian ensemble of quantum
Gaussian states. First, consider that the expected density operator $\rho_{E}$
in \eqref{eq:rho_E_gaussian} is a quantum Gaussian state because it is a
Gaussian mixture of Gaussian states. Furthermore, the entries $\mu_{E}^{j}$ of
the mean vector $\mu_{E}$ of $\rho_{E}$ are given by%
\begin{align}
\mu_{E}^{j} &  \coloneqq \operatorname{Tr}[\hat{r}_{j}\rho_{E}]\\
&  =\int dy\ p_{Y}(y)\operatorname{Tr}[\hat{r}_{j}\rho_{E}^{y}]\\
&  =\int dy\ p_{Y}(y)\left(  \left[  Wy\right]  _{j}+\nu_{j}\right)  \\
&  =\left[  W\mu+\nu\right]  _{j},
\end{align}
so that%
\begin{equation}
\mu_{E}=W\mu+\nu.\label{eq:mean-of-avg-E}%
\end{equation}
The entries $V_{E}^{jk}$ of the covariance matrix $V_{E}$\ of $\rho_{E}$ are
given by%
\begin{align}
V_{E}^{jk} &  \coloneqq \operatorname{Tr}\left[  \left\{  \hat{r}_{j}-\mu_{E}^{j}%
,\hat{r}_{k}-\mu_{E}^{k}\right\}  \rho_{E}\right]  \\
&  =\operatorname{Tr}\left[  \left\{  \hat{r}_{j},\hat{r}_{k}\right\}
\rho_{E}\right]  -2\mu_{E}^{j}\operatorname{Tr}\left[  \hat{r}_{k}\rho
_{E}\right]  \nonumber\\
&  \qquad-2\mu_{E}^{k}\operatorname{Tr}\left[  \hat{r}_{j}\rho_{E}\right]
+2\mu_{E}^{j}\mu_{E}^{k}\operatorname{Tr}[\rho_{E}]\\
&  =\operatorname{Tr}\left[  \left\{  \hat{r}_{j},\hat{r}_{k}\right\}
\rho_{E}\right]  -2\mu_{E}^{j}\mu_{E}^{k}.\label{eq:avg-E-cov-matrix-progress}%
\end{align}
Let us focus on the first term $\operatorname{Tr}\left[  \left\{  \hat{r}%
_{j},\hat{r}_{k}\right\}  \rho_{E}\right]  $. Set%
\begin{equation}
y_{c}\coloneqq y-\mu,\label{eq:centered-y}%
\end{equation}
and consider that%
\begin{align}
&  \operatorname{Tr}\left[  \left\{  \hat{r}_{j},\hat{r}_{k}\right\}  \rho
_{E}\right]  \nonumber\\
&  =\int dy\ p_{Y}(y)\ \operatorname{Tr}\left[  \left\{  \hat{r}_{j},\hat
{r}_{k}\right\}  \rho_{E}^{y}\right]  \\
&  =\int dy\ p_{Y}(y)\ \left(  V^{jk}+2\left[  Wy+\nu\right]  _{j}\left[
Wy+\nu\right]  _{k}\right)  \\
&  =V^{jk}+2\int dy\ p_{Y}(y)\ \left[  Wy+\nu\right]  _{j}\left[
Wy+\nu\right]  _{k}\\
&  =V^{jk}+2\int dy\ p_{Y}(y)\ \left[  Wy_{c}+\mu_{E}\right]  _{j}\left[
Wy_{c}+\mu_{E}\right]  _{k},
\end{align}
where the second equality follows from the definition of the quantum
covariance matrix $V$ of $\rho_{E}^{y}$ and the fact that $Wy+\nu$ is the mean
vector of $\rho_{E}^{y}$. The last equality follows from
\eqref{eq:mean-of-avg-E} and \eqref{eq:centered-y}. Focusing on the second
term, we find that%
\begin{align}
&  \int dy\ p_{Y}(y)\ \left[  Wy_{c}+\mu_{E}\right]  _{j}\left[  Wy_{c}%
+\mu_{E}\right]  _{k}\nonumber\\
&  =\int dy\ p_{Y}(y)\left(  \mu_{E}^{j}+\sum_{\ell}W_{j\ell}y_{c}^{\ell
}\right)  \left(  \mu_{E}^{k}+\sum_{m}W_{km}y_{c}^{m}\right)  \\
&  =\mu_{E}^{j}\mu_{E}^{k}+\sum_{\ell,m}W_{j\ell}W_{km}\int dy\ p_{Y}%
(y)y_{c}^{\ell}y_{c}^{m}\\
&  =\mu_{E}^{j}\mu_{E}^{k}+\sum_{\ell,m}W_{j\ell}W_{km}\Sigma_{\ell m}\\
&  =\mu_{E}^{j}\mu_{E}^{k}+\left[  W\Sigma W^{T}\right]  _{jk},
\end{align}
where the second equality follows because%
\begin{equation}
\int dy\ p_{Y}(y)\sum_{\ell}W_{j\ell}y_{c}^{\ell}\mu_{E}^{k}=0,
\end{equation}
with similar reasoning for the other vanishing term. It follows that%
\begin{equation}
\operatorname{Tr}\left[  \left\{  \hat{r}_{j},\hat{r}_{k}\right\}  \rho
_{E}\right]  =V^{jk}+2\left(  \left[  \mu_{E}^{j}\mu_{E}^{k}+W\Sigma
W^{T}\right]  _{jk}\right)  ,
\end{equation}
and we find from combining with \eqref{eq:avg-E-cov-matrix-progress}\ that%
\begin{equation}
V_{E}^{jk}=V^{jk}+2\left[  W\Sigma W^{T}\right]  _{jk},
\end{equation}
or equivalently,%
\begin{equation}
V_{E}=V+2W\Sigma W^{T}.\label{eq:QCM-of-avg-E}%
\end{equation}
So the expected density operator $\rho_{E}$ is a quantum Gaussian state with
mean vector given by \eqref{eq:mean-of-avg-E} and quantum covariance matrix
given by \eqref{eq:QCM-of-avg-E}.\ The normalization $Z_{E}$ of $\rho_{E}$ is
thus given by%
\begin{equation}
Z_{E}\coloneqq \det(\left[  V_{E}+i\Omega\right]  /2),
\end{equation}
and the Hamiltonian matrix $G_{E}$ of $\rho_{E}$ is given by%
\begin{equation}
G_{E}\coloneqq 2i\Omega\operatorname{arcoth}(iV_{E}\Omega).
\end{equation}

Thus, the Holevo information $I(Y;E)$\ works out to%
\begin{align}
&  I(Y;E)\\
&  =\int dy\ p_{Y}(y)D(\rho_{E}^{y}\Vert\rho_{E})\\
&  =\int dy\ \frac{p_{Y}(y)}{2\ln2}\left[  \ln\!\left(  \frac{Z_{E}}%
{Z}\right)  +\frac{1}{2}\operatorname{Tr}[V\Delta]+\delta^{T}G_{E}%
\delta\right] \label{eq:d(rho-y-to-rho)}\\
&  =\frac{1}{2\ln2}\left[  \ln\!\left(  \frac{Z_{E}}{Z}\right)  +\frac{1}%
{2}\operatorname{Tr}[V\Delta]+\int dy\ p_{Y}(y)\delta^{T}G_{E}\delta\right]  ,
\end{align}
where%
\begin{equation}
\delta\coloneqq Wy+\nu-\left(  W\mu+\nu\right)  =W(y-\mu).
\end{equation}
For the equality in \eqref{eq:d(rho-y-to-rho)}, we applied the formula in
\eqref{eq:rel-ent-Gaussian}. To evaluate the last integral, consider that%
\begin{align}
&  \int dy\ p_{Y}(y)\delta^{T}G_{E}\delta\nonumber\\
&  =\int dy\ p_{Y}(y)\operatorname{Tr}[\delta\delta^{T}G_{E}]\\
&  =\int dy\ p_{Y}(y)\operatorname{Tr}[W(y-\mu)(y-\mu)^{T}W^{T}G_{E}]\\
&  =\operatorname{Tr}[W\Sigma W^{T}G_{E}]. \label{eq:Gaussian-avg-quad-form}%
\end{align}
Then the formula for the Holevo information $I(Y;E)$\ reduces to%
\begin{multline}
I(Y;E)=\\
\frac{1}{2\ln2}\left[  \ln\!\left(  \frac{Z_{E}}{Z}\right)  +\frac{1}%
{2}\operatorname{Tr}[V\Delta]+\operatorname{Tr}[W\Sigma W^{T}G_{E}]\right]  ,
\end{multline}
as claimed.

We now determine a formula for the Holevo information variance. We first
consider the difference%
\begin{equation}
\int dy\ p_{Y}(y)\left(  D(\rho_{E}^{y}\Vert\rho_{E})\right)  ^{2}-\left[
I(Y;E)\right]  ^{2}.
\end{equation}
Consider from \eqref{eq:d(rho-y-to-rho)} that%
\begin{align}
&  \left(  2\ln2\right)  ^{2}\left[  \left(  D(\rho_{E}^{y}\Vert\rho
_{E})\right)  ^{2}-\left[  I(Y;E)\right]  ^{2}\right] \nonumber\\
&  =\left[  \ln\!\left(  \frac{Z_{E}}{Z}\right)  +\frac{1}{2}\operatorname{Tr}%
[V\Delta]+\delta^{T}G_{E}\delta\right]  ^{2}\nonumber\\
&  \qquad-\left[  \ln\!\left(  \frac{Z_{E}}{Z}\right)  +\frac{1}%
{2}\operatorname{Tr}[V\Delta]+\operatorname{Tr}[W\Sigma W^{T}G_{E}]\right]
^{2}\\
&  =2\left(  \ln\!\left(  \frac{Z_{E}}{Z}\right)  +\frac{1}{2}%
\operatorname{Tr}[V\Delta]\right)  \left(  \delta^{T}G_{E}\delta
-\operatorname{Tr}[W\Sigma W^{T}G_{E}]\right) \nonumber\\
&  \qquad+\left(  \delta^{T}G_{E}\delta\right)  ^{2}-\left(  \operatorname{Tr}%
[W\Sigma W^{T}G_{E}]\right)  ^{2}.
\end{align}
This means that%
\begin{multline}
\left(  2\ln2\right)  ^{2}\int dy\ p_{Y}(y)\left(  D(\rho_{E}^{y}\Vert\rho
_{E})\right)  ^{2}-\left[  I(Y;E)\right]  ^{2}%
\label{eq:diff-square-rel-ent-square-MI}\\
=\int dy\ p_{Y}(y)\left(  \delta^{T}G_{E}\delta\right)  ^{2}-\left(
\operatorname{Tr}[W\Sigma W^{T}G_{E}]\right)  ^{2},
\end{multline}
due to \eqref{eq:Gaussian-avg-quad-form}. To evaluate the first term, consider
that%
\begin{multline}
\int dy\ p_{Y}(y)\left(  \delta^{T}G_{E}\delta\right)  ^{2}\\
=\int dy\ p_{Y}(y)\left(  y_{c}^{T}W^{T}G_{E}Wy_{c}\right)  ^{2}.
\end{multline}
Then write the above as%
\begin{align}
&  \int dy\ p_{Y}(y)\sum_{ijk\ell}y_{c}^{i}\left[  W^{T}G_{E}W\right]
_{ij}y_{c}^{j}y_{c}^{k}\left[  W^{T}G_{E}W\right]  _{k\ell}y_{c}^{\ell
}\nonumber\\
&  =\sum_{ijk\ell}\left[  W^{T}G_{E}W\right]  _{ij}\left[  W^{T}G_{E}W\right]
_{k\ell}\int dy\ p_{Y}(y)y_{c}^{i}y_{c}^{j}y_{c}^{k}y_{c}^{\ell}\\
&  =\sum_{ijk\ell}\left[  W^{T}G_{E}W\right]  _{ij}\left[  W^{T}G_{E}W\right]
_{k\ell}\times\nonumber\\
&  \qquad\left[  \Sigma^{ij}\Sigma^{k\ell}+\Sigma^{ik}\Sigma^{j\ell}%
+\Sigma^{i\ell}\Sigma^{jk}\right] \\
&  =\left(  \operatorname{Tr}[W^{T}G_{E}W\Sigma]\right)  ^{2}%
+2\operatorname{Tr}[\left(  W^{T}G_{E}W\Sigma\right)  ^{2}]\\
&  =\left(  \operatorname{Tr}[W\Sigma W^{T}G_{E}]\right)  ^{2}%
+2\operatorname{Tr}[\left(  W\Sigma W^{T}G_{E}\right)  ^{2}],
\end{align}
where we applied Isserlis' theorem \cite{Isserlis18} to evaluate the fourth
moment and we employed the facts that $W^{T}G_{E}W$ and $\Sigma$ are symmetric
matrices. So then by combining with \eqref{eq:diff-square-rel-ent-square-MI},
we find that%
\begin{multline}
\int dy\ p_{Y}(y)\left(  D(\rho_{E}^{y}\Vert\rho_{E})\right)  ^{2}-\left[
I(Y;E)\right]  ^{2}\\
=\frac{1}{2\ln^{2}2}\operatorname{Tr}[\left(  W\Sigma W^{T}G_{E}\right)
^{2}].
\end{multline}
It remains to evaluate the term%
\begin{align}
&  \int dy\ p_{Y}(y)V(\rho_{E}^{y}\Vert\rho_{E})\nonumber\\
&  =\int dy\ \frac{p_{Y}(y)}{8\ln^{2}2}\left[  \operatorname{Tr}[\left(
\Delta V\right)  ^{2}]+\operatorname{Tr}[\left(  \Delta\Omega\right)
^{2}]+4\delta^{T}G_{E}VG_{E}\delta\right] \\
&  =\frac{1}{8\ln^{2}2}\left[  \operatorname{Tr}[\left(  \Delta V\right)
^{2}]+\operatorname{Tr}[\left(  \Delta\Omega\right)  ^{2}]\right] \nonumber\\
&  \qquad+\frac{1}{2\ln^{2}2}\int dy\ p_{Y}(y)\delta^{T}G_{E}VG_{E}\delta,
\end{align}
where we applied the formula in \eqref{eq:rel-ent-var-Gaussian}. Then it
follows that%
\begin{align}
&  \int dy\ p_{Y}(y)\delta^{T}G_{E}VG_{E}\delta\nonumber\\
&  =\int dy\ p_{Y}(y)\operatorname{Tr}[\delta\delta^{T}G_{E}VG_{E}]\\
&  =\int dy\ p_{Y}(y)\operatorname{Tr}[W\left(  y-\mu\right)  \left(
y-\mu\right)  ^{T}W^{T}G_{E}VG_{E}]\\
&  =\operatorname{Tr}[W\Sigma W^{T}G_{E}VG_{E}].
\end{align}
Putting everything together, we find that%
\begin{multline}
V(Y;E)=\frac{1}{8\ln^{2}2}\left[  \operatorname{Tr}[\left(  \Delta V\right)
^{2}]+\operatorname{Tr}[\left(  \Delta\Omega\right)  ^{2}]\right] \\
+\frac{1}{2\ln^{2}2}\left[  \operatorname{Tr}[W\Sigma W^{T}G_{E}%
VG_{E}]+\operatorname{Tr}[\left(  W\Sigma W^{T}G_{E}\right)  ^{2}]\right]  ,
\end{multline}
as claimed.
\end{proof}

\bigskip

The formulas from Proposition~\ref{prop:gaussian-formulas-HI-Hi-var}\ can be
applied to the scenario in which some modes of a Gaussian state are measured
according to a \textquotedblleft general-dyne\textquotedblright\ Gaussian
measurement \cite{GLS16,Ser17}, which leaves a Gaussian ensemble of Gaussian
states on the remaining modes. To see how this works, let $\rho_{AB}$ denote a
bipartite Gaussian state of $m+n$ modes, with $m$ modes for system $A$ and $n$
modes for system $B$. Suppose that the $2\left(  m+n\right)  \times1$\ mean
vector of $\rho_{AB}$ is%
\begin{equation}%
\begin{bmatrix}
\overline{r}_{A}\\
\overline{r}_{B}%
\end{bmatrix}
,
\end{equation}
and the $2\left(  m+n\right)  \times2\left(  m+n\right)  $ quantum covariance
matrix is%
\begin{equation}%
\begin{bmatrix}
V_{A} & V_{AB}\\
V_{AB}^{T} & V_{B}%
\end{bmatrix}
.
\end{equation}
A general-dyne measurement of system $B$ is described by a quantum Gaussian
state $\omega_{M}$ with covariance matrix $V_{M}$ satisfying the uncertainty
principle $V_{M}+i\Omega\geq0$ \cite{GLS16,Ser17}. The POVM\ elements of this
general-dyne detection are given by%
\begin{equation}
\left\{  \frac{1}{\left(  2\pi\right)  ^{n}}\hat{D}_{-y}\omega_{M}\hat{D}%
_{y}\right\}  _{y\in\mathbb{R}^{2n}},
\end{equation}
where the displacement operator is defined as $\hat{D}_{y}\coloneqq \exp( iy^{T}%
\Omega\hat{r}) $, and the following completeness relation holds%
\begin{equation}
\frac{1}{\left(  2\pi\right)  ^{n}}\int_{\mathbb{R}^{2n}}dy\ \hat{D}%
_{-y}\omega_{M}\hat{D}_{y}.
\end{equation}
(Note that ideal heterodyne detection corresponds to $V_{M}=I_{2n}$.) If this
measurement is performed on system~$B$ of the state $\rho_{AB}$ as specified
above, then the induced ensemble on system $A$ is $\left\{  p_{Y}(y),\rho
_{A}^{y}\right\}  $, where%
\begin{equation}
p_{Y}(y)=\frac{\exp\left(  -\frac{1}{2}\left(  y-\overline{r}_{B}\right)
^{T}\left(  \frac{V_{B}+V_{M}}{2}\right)  ^{-1}\left(  y-\overline{r}%
_{B}\right)  \right)  }{\sqrt{\left(  2\pi\right)  ^{2n}\det\left(
\frac{V_{B}+V_{M}}{2}\right)  }},
\end{equation}
and $\rho_{A}^{y}$ is a quantum Gaussian state with mean vector%
\begin{equation}
\overline{r}_{A}+V_{AB}\left(  V_{B}+V_{M}\right)  ^{-1}\left(  y-\overline
{r}_{B}\right)  ,
\end{equation}
and quantum covariance matrix%
\begin{equation}
V_{A}-V_{AB}\left(  V_{B}+V_{M}\right)  ^{-1}V_{AB}^{T}.
\end{equation}
Thus, we can use Proposition~\ref{prop:gaussian-formulas-HI-Hi-var}\ to
calculate both the Holevo information and the Holevo information variance of
this ensemble with the following identifications in
Definition~\ref{def:gaussian-ensemble}:%
\begin{align}
\mu &  =\overline{r}_{B},\\
\Sigma &  =\frac{V_{B}+V_{M}}{2},\\
W  &  =V_{AB}\left(  V_{B}+V_{M}\right)  ^{-1},\\
\nu &  =\overline{r}_{A}-V_{AB}\left(  V_{B}+V_{M}\right)  ^{-1}\overline
{r}_{B},\\
V  &  =V_{A}-V_{AB}\left(  V_{B}+V_{M}\right)  ^{-1}V_{AB}^{T}.
\end{align}

\newpage

%{\begin{center}\Large Supplemental material\end{center}}
%
%

\end{document}